\documentclass[letterpaper, 12pt]{report}                  
\usepackage[reqno]{amsmath}
\usepackage{adjustbox}
\usepackage{multirow}
\newcommand{\uprightcell}[1]{\rotatebox{-90}{#1}} 

\usepackage{lipsum}
\usepackage{textcase}
\usepackage{mathtools}
\usepackage{multicol}
\usepackage{comment}
\usepackage{geometry}
\usepackage{fancyhdr} 

\usepackage{gscale_thesis_doublespace} 

\usepackage{setspace}                           

\halftitle{AI-Driven Cardiorespiratory Signal Processing: Separation, Clustering, and Anomaly Detection}

\title{AI-Driven Cardiorespiratory Signal Processing: Separation, Clustering, and Anomaly Detection}

\author{Yasaman Torabi}
\shortauthor{Y. Torabi} 

\dept{\href{https://www.eng.mcmaster.ca/ece}{Electrical and Computer Engineering}} 
\gschoolname{\href{https://gs.mcmaster.ca/}{School of Graduate Studies}}
\univname{\href{http://www.mcmaster.ca/}{McMaster University}} 
\macaddress{Hamilton, Ontario, Canada}

\prevdegreeone{B.Sc. (Electrical Engineering),\\ Sharif University of Technology, Tehran, Iran}
\prevdegreetwo{B.Sc.} 
\degreename{Doctor of Philosophy}

\submitmonthyear{2025} 
\submitdate{\today} 
\copyrightyear{2025} 

\principaladviser{Prof. Shahram Shirani and Prof. James P. Reilly} 
\usepackage{graphicx}                                   
\usepackage{subcaption}                               
\usepackage[justification=centering]{caption} 
\usepackage[dvipsnames]{xcolor}
\usepackage[
  colorlinks=false,          
  pdfborderstyle={/S/S/W 1}, 
  citebordercolor=ForestGreen, 
]{hyperref}
\usepackage{booktabs}                                 
\usepackage{longtable}                                 
\usepackage{float}                                         
\usepackage{enumerate}   
\usepackage[shortlabels]{enumitem}                            
\usepackage[shortcuts]{extdash}                  
\usepackage{amsmath}
\usepackage{braket} 
\usepackage{amsfonts}                                
\numberwithin{equation}{chapter}              
\usepackage[backend=bibtex,style=numeric,sorting=none,defernums=false]{biblatex}
\newcommand{\figsubref}[2]{\hyperref[#1]{\ref*{#1}#2}}

\usepackage{amsmath, amssymb, amsthm}
\usepackage{algorithm}
\usepackage{algpseudocode}
\usepackage{amsmath}
\usepackage{dsfont}
\usepackage{chngcntr}

\usepackage{tikz}
\usetikzlibrary{arrows.meta,positioning,shapes.geometric}

\counterwithin{algorithm}{chapter} 
\renewcommand{\thealgorithm}{\thechapter.\arabic{algorithm}}

\usepackage{bm}
\newtheorem{theorem}{Theorem}[section]
\newtheorem{lemma}{Lemma}[section]
\theoremstyle{definition}
\newtheorem{definition}{Definition}[section]

\newtheorem{remark}{Remark}[section]
\usepackage{titling}  
\usepackage{ragged2e}

\usepackage{algorithm,algpseudocode}
\usepackage{lipsum}

\makeatletter
\newenvironment{breakablealgorithm}
  {
   \begin{center}
     \refstepcounter{algorithm}
     \hrule height.8pt depth0pt \kern2pt
     \renewcommand{\caption}[2][\relax]{
       {\raggedright\textbf{\fname@algorithm~\thealgorithm} ##2\par}%
       \ifx\relax##1\relax 
         \addcontentsline{loa}{algorithm}{\protect\numberline{\thealgorithm}##2}%
       \else 
         \addcontentsline{loa}{algorithm}{\protect\numberline{\thealgorithm}##1}%
       \fi
       \kern2pt\hrule\kern2pt
     }
  }{
     \kern2pt\hrule\relax
   \end{center}
  }
\makeatother

\begin{document}

\hypersetup{
  hidelinks,         
  colorlinks=false   
}
\beforepreface 
\null\vfill
\begin{center}
\textsl{To my beloved parents, my wonderful supervisors, and all my dearest teachers \\ who have supported me along this academic journey.}
\end{center}
\vfill
    \prefacesection{Lay Abstract}

This research applies artificial intelligence (AI) to separate, cluster, and analyze cardiorespiratory sounds. We recorded a new dataset and developed several AI models, including generative AI methods based on large language models (LLMs) for guided separation, explainable AI (XAI) techniques to interpret latent representations, variational autoencoders (VAEs) for waveform separation, a chemistry-inspired non-negative matrix factorization (NMF) algorithm for clustering, and a quantum convolutional neural network (QCNN) designed to detect abnormal physiological patterns. The performance of these AI models depends on the quality of the recorded signals. Therefore, this thesis also reviews the biosensing technologies used to capture biomedical data. Together, these studies show how AI and next-generation sensors can support more intelligent diagnostic systems for future healthcare.                                  
  \prefacesection{Abstract}

This thesis presents artificial intelligence algorithms for cardiorespiratory signal processing, focusing on blind source separation, clustering, and anomaly detection. We created a heart and lung sound dataset, called \textit{HLS-CMDS}, containing normal and abnormal cardiorespiratory recordings collected from clinical manikins using a digital stethoscope at a sampling frequency of 22 kHz. The dataset supports algorithmic development for analyzing mixed biomedical sounds. 

Beyond algorithm design, this thesis also reviews the sensing technologies that enable these analyses. It summarizes developments in microelectromechanical systems (MEMS) acoustic sensors and quantum biosensors, such as quantum dots and nitrogen-vacancy centers. It further outlines the transition from electronic integrated circuits (EICs) to photonic integrated circuits (PICs) and early progress toward integrated quantum photonics (IQP) for chip-based biosensing. 

For blind source separation, we developed \textit{LingoNMF}, a large language model-based non-negative matrix factorization (NMF) algorithm that combines the periodic nature of the biological sounds with language-guided reasoning. It improves separation performance, increasing heart sound signal to distortion ratio (SDR) and signal to interference ratio (SIR) by up to 3–4 dB. In addition, we introduced \textit{XVAE-WMT}, a masked wavelet-based variational autoencoder (VAE) with temporal-consistency loss and explainable latent-space analysis for separation of overlapping heart and lung sounds. XVAE-WMT yields SDR = 26.8 dB and SIR = 32.8 dB, with a latent-space Silhouette score of 0.345, showing consistently stronger separation and clustering performance than the other VAE variants.
\noindent
For unsupervised clustering, we proposed \textit{Chem-NMF}, a multi-layer $\alpha$-divergence matrix factorization algorithm inspired by physical chemistry phenomena. Just as catalysts reduce activation barriers and regulate the reaction rate, the algorithm controls the initialization steps and stabilizes the algorithm’s convergence.
\noindent
For anomaly detection, we designed \textit{QuPCG}, a quantum convolutional neural network that encodes physiological features into qubits to detect cardiac abnormalities through quantum-classical computation.  It achieves $93.33 \pm 2.9\%$ test accuracy, suggesting potential for binary heart-sound classification.
Together, these works demonstrate how generative AI, physics-inspired algorithms, and quantum machine learning approaches can transform the analysis of cardiorespiratory signals.

  \prefacesection{Acknowledgements}

I would like to express my gratitude to my supervisor, Professor Shahram Shirani, whose mentorship, encouragement, and insightful guidance have been central to my doctoral journey. His depth of knowledge, patience, and vision have influenced my research development. I also wish to express my appreciation to my co-supervisor, Professor James P. Reilly, for his valuable advice, constructive feedback, and continued support throughout the progression of this work. My thanks extend to my committee members, Professor Michael Noseworthy and Professor Thia Kirubarajan, for their thoughtful comments and helpful suggestions, which significantly enhanced the quality of this dissertation. Lastly, I am deeply grateful to my family for their constant love, understanding, and encouragement, which have sustained me throughout this academic journey.                 
    \referencepages 

\hypersetup{
  colorlinks=false,
  pdfborderstyle={/S/S/W 1},
}
\afterpreface                      
  \chapter{Introduction}
\section{Motivation and Background}

Cardiopulmonary diseases are among the most critical global health challenges. Accurate assessment of both heart and lung function is therefore essential for early diagnosis and prevention. Auscultation, i.e. the practice of listening to internal body sounds, has long served as a cornerstone of cardiopulmonary evaluation~\cite{C71ref3}. In 1816, René Laennec invented the first stethoscope to amplify chest sounds~\cite{C71ref5}. Since then, stethoscopes have evolved from simple mechanical tubes into advanced electronic instruments that convert acoustic waves into electrical signals for amplification and recording~\cite{C71ref6}. The emergence of digital stethoscopes further enhanced auscultation by enabling real-time visualization and data storage through cloud-based systems~\cite{Torabi2024arxiv}. Modern devices integrate microelectromechanical technologies to improve sensitivity~\cite{C71ref8}. The \textit{cardiac cycle} is the rhythmic contraction and relaxation of the atria and ventricles that maintain blood circulation. Heart sounds arise from valve motion during systole and diastole~\cite{C71ref17}, while murmurs reflect turbulent flow~\cite{C71ref20}. The \textit{respiratory cycle} alternates inspiration and expiration as air moves through the alveoli~\cite{C71ref25}. Normal breath sounds show low expiratory and higher inspiratory frequencies~\cite{C71ref30}. Crackles are brief discontinuous bursts from airway reopening~\cite{C71ref31}, and wheezes are continuous tones from airway narrowing~\cite{C71ref35}. Together, these acoustic phenomena provide a non-invasive window into cardiopulmonary physiology. Their accurate sensing and intelligent interpretation through AI form the foundation of this thesis.

\section{A Review on Biosensors}

Biosensing technologies play a central role in clinical diagnostics~\cite{Baraeinejad2025IoTResp}. Since the first oxygen biosensor developed by Clark and Lyons in 1962, major advances in microfabrication and materials science have shifted biosensing from bulky laboratory setups to compact, chip-scale platforms~\cite{Hassan2022}. Electronic integrated circuits (EICs) marked the major breakthrough in on-chip biosensing by converting biochemical interactions into measurable electrical signals. CMOS-based biosensors have been widely used for bioimpedance monitoring~\cite{Kim2020}, cancer biomarker detection~\cite{Alhoshany2017}, and glucose sensing~\cite{AlMamun2016}. However, electrical noise, limited sensitivity, and power dissipation have constrained their detection ability. In addition to CMOS, miniature microelectromechanical systems (MEMS) technologies have advanced the performance of EIC biosensors, offering superior sensitivity, greater miniaturization, and higher power efficiency compared to conventional electronic devices. MEMS-based biosensors use mechanical structures that respond to forces or mass changes induced by biomolecular interactions. For example, Mehdipoor et al.~\cite{mehdi} introduced a MEMS resonator for microfluidic analysis with high mass sensitivity. Photonic integrated circuits (PICs) are considered as a promising alternative to EICs, using light–matter interaction for high sensitivity, fast response, and immunity to electromagnetic interference~\cite{Torabi2025Microring}. Waveguide and microring resonator-based PIC biosensors have enabled label-free detection of water pollutants~\cite{Prieto2003}, viruses~\cite{Ning2023}, and antibody profiling~\cite{Bryan2024}. Yet, fabrication complexity and optical losses remain technical challenges~\cite{Sharma2021}. Recent developments in quantum technologies have introduced biosensors based on quantum dots (QDs)~\cite{Wei2022} and nitrogen-vacancy (NV) centers~\cite{Zhang2020}. Despite their superior sensitivity, these devices suffer from decoherence and integration challenges. The growing field of integrated quantum photonics (IQP) aims to merge electronic, photonic, and quantum functionalities on a single chip. Despite these advances, many challenges remain in implementing fully integrated quantum photonic sensors. Current IQP circuits face notable challenges, including maintaining quantum coherence at room temperature, efficient photon coupling, minimizing optical losses, and achieving reliable quantum state readout without extensive cryogenic infrastructure. In 2025, Kramnik et al.~\cite{Kramnik2025} demonstrated the first CMOS-compatible IQP platform, marking a key step toward scalable quantum biosensing on silicon. Figure~\ref{fig:biosensors} illustrates the development timeline of biosensing devices.

\begin{figure}[H]
    \centering
    \includegraphics[width=\linewidth, trim= 0 100 0 120, clip]{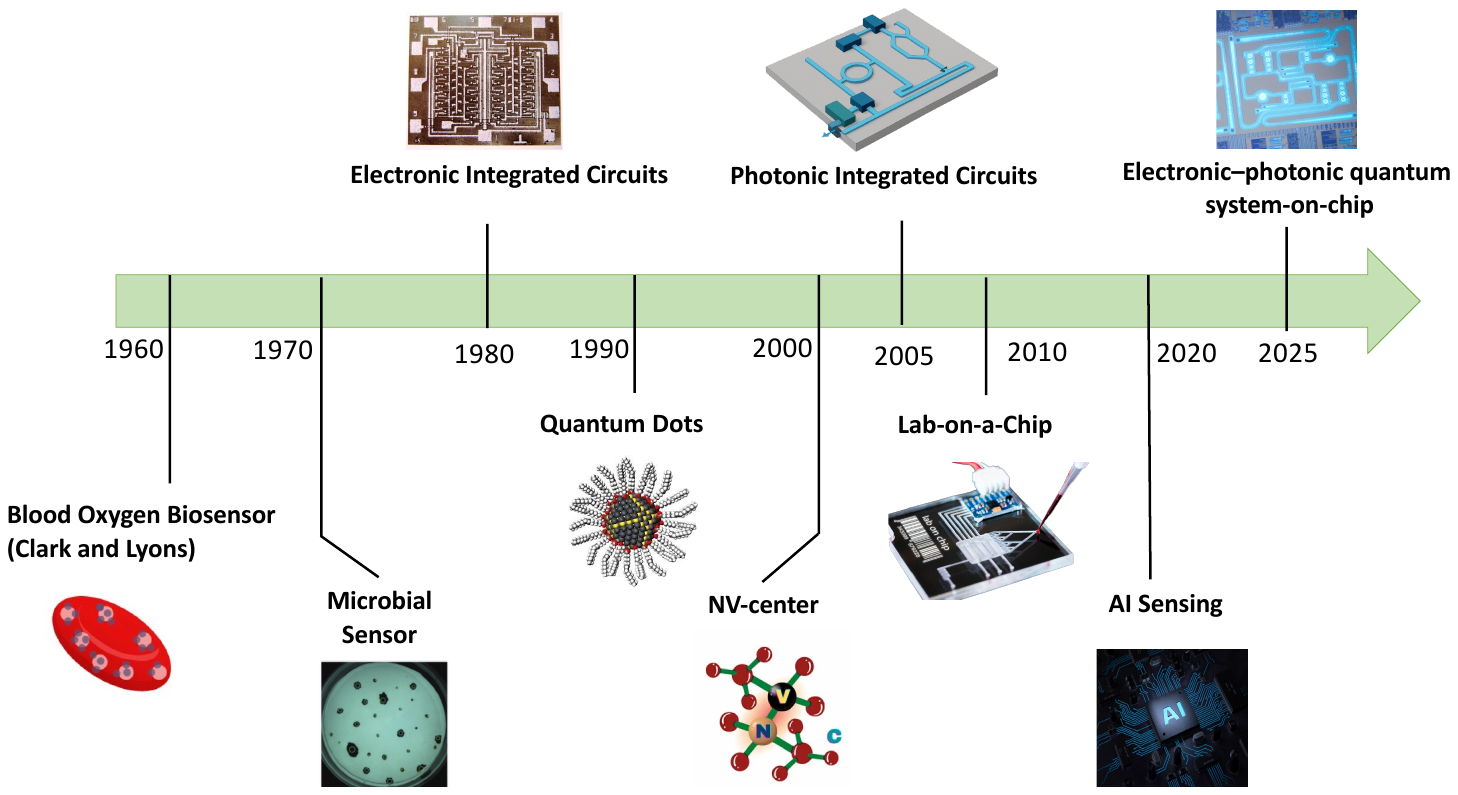}
    \captionsetup{justification=justified,font=footnotesize}
    \caption{Biosensor development timeline. It highlights key technological milestones from the first blood oxygen biosensor, integrated circuits, and quantum sensors, to recent advances toward integrated quantum–photonic systems and AI-assisted lab-on-chip biosensing.}
    \label{fig:biosensors}
\end{figure}

\section{NMF Algorithms}

Non-negative matrix factorization (NMF) is an interpretable representation learning method that decomposes data into low-rank components~\cite{Torabi2023arxiv}. Non-negativity enforces the additive, physically meaningful structure of cardiorespiratory signals and helps improve separation stability and interpretability. As a clustering-based approach, NMF generates compact feature representations that preserve the physical meaning of the data~\cite{C71ref28}. It has been successfully used in image recognition, text mining, and blind source separation, where each non-negative basis vector corresponds to a distinct underlying signal pattern~\cite{Torabi2023}. Unlike classical multi-microphone approaches such as Independent Component Analysis (ICA), which rely on spatial diversity and statistical independence across sensors, the proposed NMF-based methods
operate in a single-channel setting and exploit spectral structure and non-negativity. As a result, they are better suited to practical auscultation scenarios where only one
recording channel is available.

Numerous variants have been developed to enhance its robustness, including graph-regularized NMF~\cite{LI2025111679}, locality-preserving NMF~\cite{IMANI202510967457}, and robust distributionally-regularized NMF~\cite{GILLIS20224052}, each designed to improve clustering performance in noisy or high-dimensional datasets. More recent approaches, such as encoder–decoder NMF, incorporate autoencoder architectures to improve feature separability~\cite{Soleymanbaigi2025}. Similarly, multi-view tensor decomposition frameworks combine low-rank representation learning with clustering indicators to unify multiple modalities within a common feature space~\cite{Wang2025}. Classical NMF algorithms, however, rely heavily on initialization, cost-function design, and convergence control, which limit their adaptability when handling nonstationary biomedical data~\cite{C71ref27}. To overcome these limitations, recent studies have explored integrating large language models (LLMs) and generative AI methods with NMF to enable dynamic parameter adjustment during training~\cite{Torabi2025LLMNMF}. LLMs are capable of identifying temporal and contextual dependencies in data and can refine the factorization process iteratively~\cite{C71ref22}. For example, Zhang et al.~\cite{Zhang2024} applied LLMs for automated hyperparameter tuning, while Wanna et al.~\cite{Wanna2024} introduced prompt-based reasoning for improved component labelling and cost-function adaptation. In biomedical signal analysis, the BioSignal Copilot framework demonstrated the potential of LLMs to interpret physiological data and organize extracted features in an explainable format~\cite{C71ref24}. These LLM-based approaches integrate interpretable matrix decomposition with adaptive learning.

\section{VAE Algorithms}

Variational Autoencoders (VAEs) are generative models that learn a structured latent space of the input data and generate new signals with similar patterns~\cite{Torabi2025VariationalBSS}. Compared to standard autoencoders, VAEs use a probabilistic learning process. Different forms of VAEs, including conditional VAE~\cite{sohn2015cvae}, vector-quantized VAE~\cite{oord2017vqvae}, and residual quantized VAE~\cite{berti2024rqvae}, have been developed to improve the latent space. Explainable AI (XAI) methods are now used to study how VAEs organize information in the latent space~\cite{mishra2025explainable}. Such techniques identify which latent dimensions contribute most to reconstruction, and help in reducing model complexity while keeping accuracy. The combined use of VAE and XAI supports interpretable separation of biomedical acoustic signals.

\section{Quantum AI Algorithms}

The future of AI is moving toward quantum computing. Classical AI systems process information sequentially, while quantum AI uses qubits that can represent multiple states at once~\cite{Patil2024}. Quantum machine learning (QML) combines the probabilistic nature of quantum systems with classical optimization to improve computational efficiency~\cite{Torabi2025QuPCG}. In hybrid models, quantum circuits extract correlations in the encoded feature space and classical layers update the parameters. Variational quantum circuits (VQCs) are widely used in pattern classification and clustering~\cite{Cerezo2021}. Quantum convolutional neural networks (QCNNs) extend this concept by using layers of quantum convolution and pooling gates that learn hierarchical representations~\cite{Cong2019}. Recent studies have applied quantum neural networks to biomedical data. Ullah et al.~\cite{Ullah2022} developed a fully connected QCNN for ischemic heart disease classification, showing higher accuracy and reduced parameter count than classical CNNs. Li et al.~\cite{Li2026} proposed a quantum-inspired convolutional model for pneumonia diagnosis. Similar work explored hybrid quantum networks for EEG classification~\cite{KoikeAkino2022} and stress detection~\cite{Nath2021}. These results confirm that quantum models can process physiological signals effectively when information is encoded in a compact form.

        \setcounter{figure}{0}
        \setcounter{equation}{0}
        \setcounter{table}{0}     
  \chapter{Integrated Biosensing: From MEMS to Quantum Biosensors on Chip}

This chapter is based on the following publications, and all figures are reproduced with permission from the cited reference; each publication is released under a Creative Commons licence that permits reuse and reprinting.

\vspace{5mm}
\noindent
\cite{Torabi2024MEMSReview} Y. Torabi, S. Shirani, J. P. Reilly, and G. M. Gauvreau, 
``MEMS and ECM Sensor Technologies for Cardiorespiratory Sound Monitoring—A Comprehensive Review,'' 
in \textit{Sensors}, vol. 24, no. 21, p. 7036, 2024. 
\href{https://doi.org/10.3390/s24217036}{doi:10.3390/s24217036}.

\vspace{5mm}
\noindent
\cite{Torabi2025QuantumBiosensors} Y. Torabi, S. Shirani, and J. P. Reilly, 
``Quantum Biosensors on Chip: A Review from Electronic and Photonic Integrated Circuits to Future Integrated Quantum Photonic Circuits,'' 
in \textit{Microelectronics}, vol. 1, no. 2, p. 5, 2025. \href{https://doi.org/10.3390/microelectronics1020005}{doi:10.3390/microelectronics1020005}.

\newpage
Biosensing technologies have evolved from simple transducers to highly integrated platforms capable of detecting physiological signals. Microelectromechanical systems (MEMS) have laid the foundation for wearable and non-invasive monitoring of cardiorespiratory activity. However, as clinical diagnostics demand higher sensitivity and miniaturization, a paradigm shift toward quantum-enabled biosensing is arising. This chapter provides a perspective on the evolution of integrated biosensing from MEMS-based sensors used in cardiorespiratory monitoring to advanced quantum biosensors based on quantum dots and nitrogen-vacancy centers. By bridging electronic, photonic, and quantum technologies, we outline the trajectory toward chip-scale biosensing platforms capable of ultra-sensitive detection.

\section{MEMS Acoustic Sensing Technologies}

\subsection{Introduction}
MEMS have become essential in wearable biosensing thanks to their compact size and low power consumption. Unlike traditional acoustic transducers, MEMS sensors enable the precise capture of weak heart and lung sounds with higher sensitivity and improved manufacturability, making them suitable for continuous monitoring and AI-based diagnostics. Their integration into portable platforms allows real-time signal analysis and digital processing, positioning MEMS as a key technology in the shift from conventional auscultation to intelligent healthcare. Despite ongoing challenges in dynamic range and biological coupling, MEMS offer a promising pathway toward high-fidelity auscultation systems.

\subsection{MEMS Working Principle}

Compared to electret condenser microphones, which are widely used for general applications, MEMS sensors offer superior sensitivity, lower power consumption, and miniaturization~\cite{C71ref84}. Figure~\figsubref{fig:mems_figures}{a} and Figure~\figsubref{fig:mems_figures}{b} illustrate the working principle of a MEMS microphone. These sensors employ a deformable diaphragm that forms a capacitive structure with a fixed backplate. When acoustic pressure is applied, the diaphragm deflects, causing a change in capacitance and generating an electrical signal. 
\begin{figure}[H]
    \centering
    \includegraphics[width=0.7\linewidth]{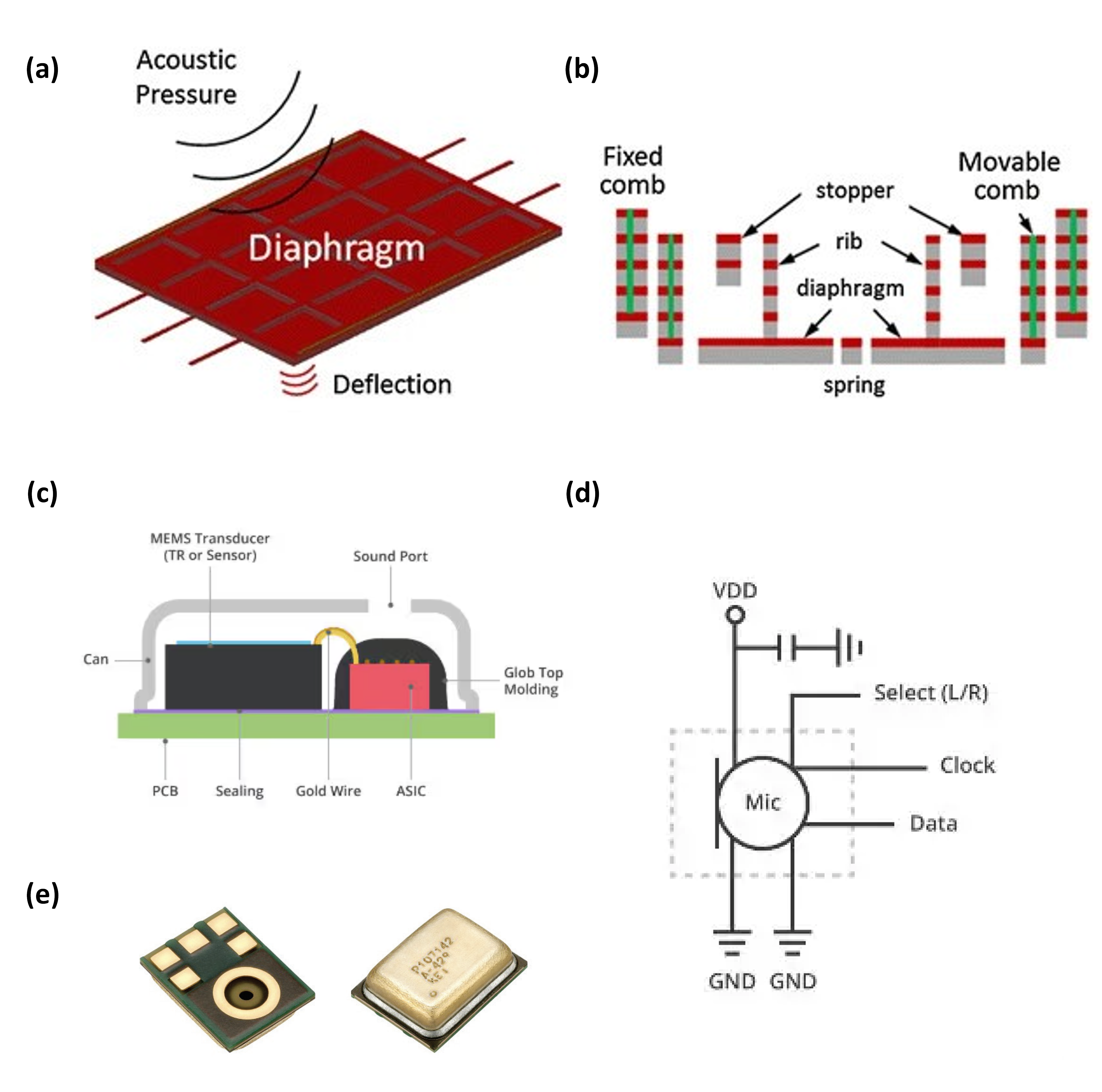}
    \captionsetup{justification=justified,font=footnotesize}
    \caption{\textbf{(a)} MEMS operating principle~\cite{C71ref90}; 
             \textbf{(b)} Cross-sectional structure of a MEMS microphone~\cite{C71ref90}; 
             \textbf{(c)} Typical MEMS microphone assembly~\cite{C71ref70}; 
             \textbf{(d)} Digital MEMS microphone application schematic~\cite{C71ref70}; 
             \textbf{(e)} Commercial MEMS microphone packages~\cite{C71ref88}. 
            }
    \label{fig:mems_figures}
\end{figure}

A representative MEMS microphone package is shown in Figure~\figsubref{fig:mems_figures}{c}. MEMS microphones may provide either analog or digital output. Digital MEMS microphones commonly utilize the I\textsuperscript{2}S protocol~\cite{C71ref88}, time-division multiplexing (TDM)~\cite{C71ref87}, or pulse density modulation (PDM)~\cite{C71ref85} for communication with processing units. These digital interfaces eliminate the need for external analog-to-digital conversion circuits, as illustrated in Figure~\figsubref{fig:mems_figures}{d} and Figure~\figsubref{fig:mems_figures}{e}. MEMS sensors are suitable for applications that demand high performance and stability, thanks to their small size, mass production capabilities, and low energy consumption.

\subsection{Recent Advances in MEMS-Based Acoustic Sensors}

In 2012, Hu et al.~\cite{C71ref98} developed a chest-worn accelerometer based on an asymmetrical gapped cantilever design to enhance the detection sensitivity of heart and lung vibrations. In 2016, Zhang et al.~\cite{C71ref99} investigated the trade-off between sensor sensitivity and the first-order resonant frequency. To address this challenge, they introduced a MEMS piezoresistive heart sound sensor with a double-beam-block configuration, optimized through analytical modelling and simulation (Figure~\figsubref{fig:mems_advanced}{a}). This architecture improved the resonant frequency while preserving high sensitivity through stress-concentration grooves. Subsequently, Wang et al.~\cite{C71ref100} proposed a bat-shaped MEMS electronic stethoscope integrating a Wheatstone bridge with piezoresistors (Figure~\figsubref{fig:mems_advanced}{b}), a configuration known for detecting minimal resistance changes with high precision~\cite{C71ref101}.

In further developments, Yilmaz et al.~\cite{C71ref102} designed a wearable stethoscope tailored for long-term respiratory monitoring. Their approach employed a diaphragm-less transducer incorporating silicone rubber and piezoelectric film to effectively capture airflow-induced acoustic signals (Figure~\figsubref{fig:mems_advanced}{c}, and Figure~\figsubref{fig:mems_advanced}{d}).

\begin{figure}[H]
    \centering
    \includegraphics[width=0.85\linewidth]{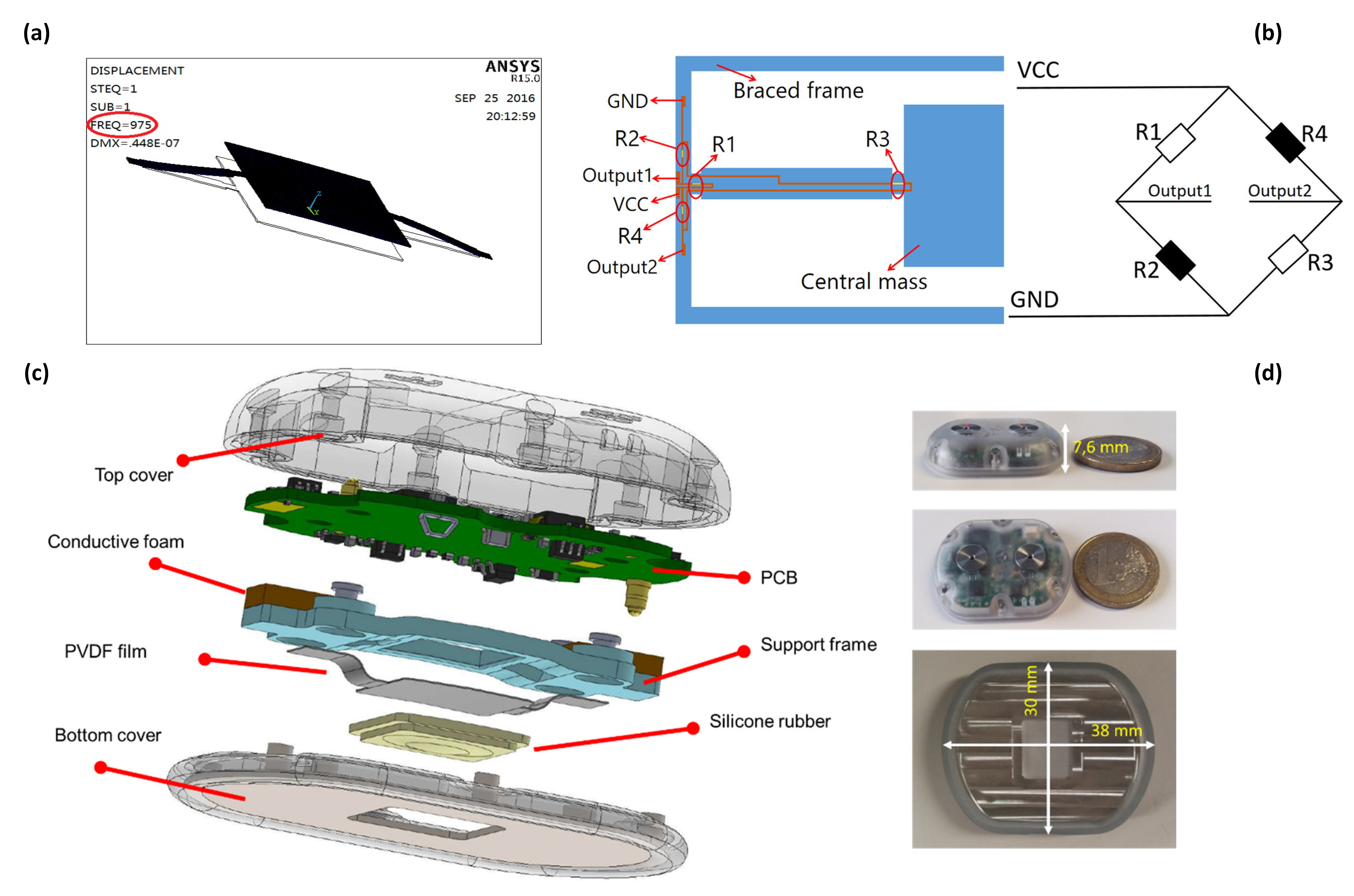}
    \captionsetup{justification=justified,font=footnotesize}
    \caption{\textbf{(a)} Coupling model analysis of MEMS piezoresistive sensor~\cite{C71ref99};
             \textbf{(b)} Resistance distribution in Wheatstone bridge configuration~\cite{C71ref100};
             \textbf{(c)} Exploded view of diaphragm-less wearable sensor module~\cite{C71ref102};
             \textbf{(d)} Assembled wearable sensor: side view with 1 Euro coin, top view, and bottom view~\cite{C71ref102}.
            }
    \label{fig:mems_advanced}
\end{figure}

In 2021, Li et al.~\cite{C71ref104} proposed a stethoscope employing a novel MEMS microstructure designed to minimize environmental noise interference. To optimize low-frequency sensitivity in heart sound acquisition, they employed a fabrication process combining Low-Pressure Chemical Vapour Deposition (LPCVD) and Deep Reactive Ion Etching (DRIE). The process began with the thermal oxidation of a Silicon-On-Insulator (SOI) wafer to establish nano-scale gaps, followed by DRIE to define the cantilever microstructure. LPCVD was subsequently used to deposit silicon dioxide, forming the sensing elements. Final integration included capping wafer bonding with through-silicon vias (TSVs) for electrical interfacing, as illustrated in Figure~\ref{fig:li_fabrication}. This fabrication strategy significantly improved device performance under clinical noise conditions, enhancing applicability in cardiology diagnostics.

\begin{figure}[H]
    \centering
    \includegraphics[width=0.75\linewidth]{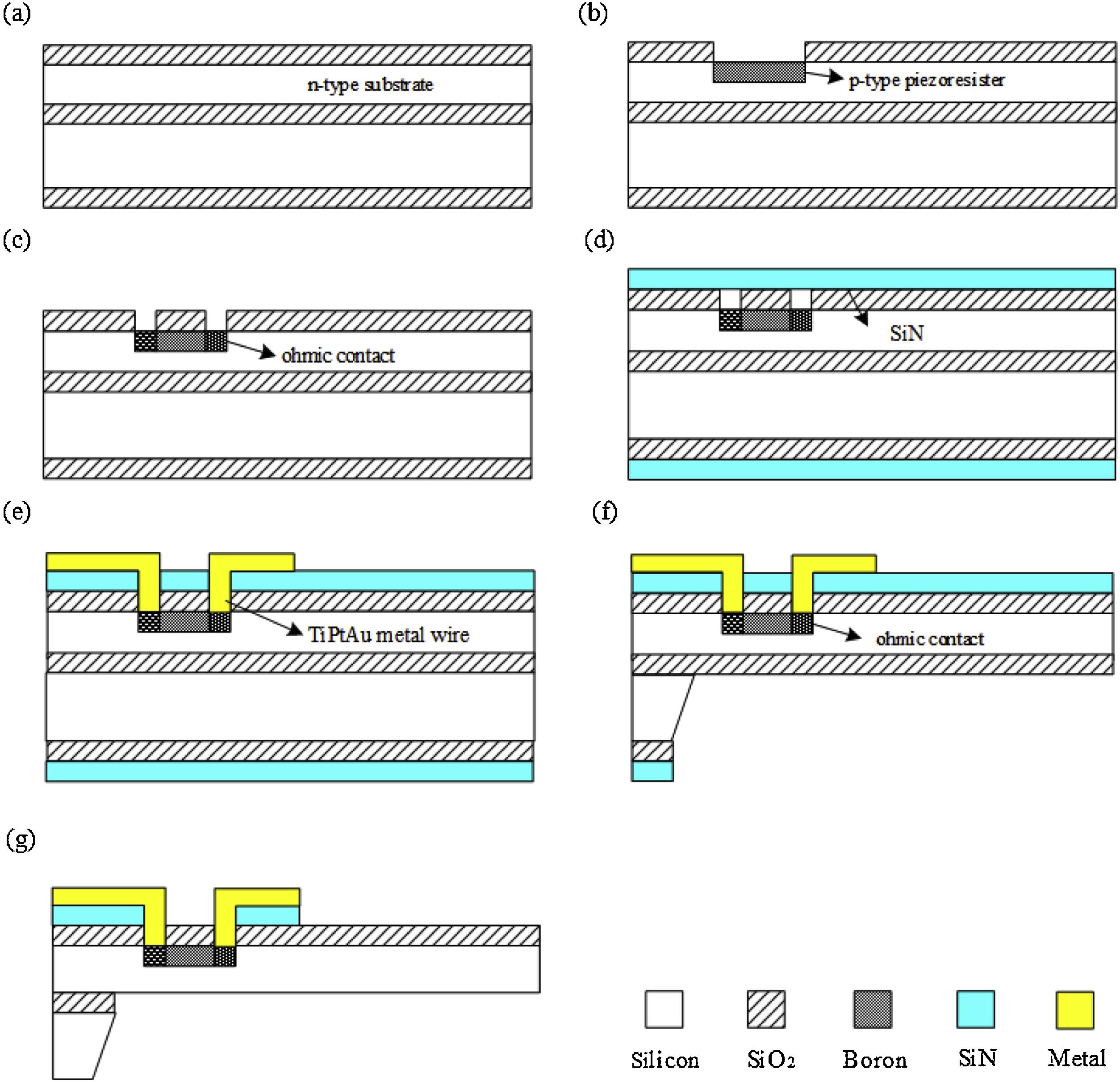}
    \captionsetup{justification=justified,font=footnotesize}
    \caption{Fabrication process of the magnetic-induction MEMS stethoscope proposed by Li et al.~\cite{C71ref104}: 
             \textbf{(a)} Wafer preparation and oxidation; 
             \textbf{(b)} Boron implantation for piezoresistors; 
             \textbf{(c)} Re-oxidation and dense boron implantation;
             \textbf{(d)} Double-sided SiN deposition; 
             \textbf{(e)} Metal sputtering for lead formation; 
             \textbf{(f)} Back cavity etching; 
             \textbf{(g)} Structural release via front etching. 
            }
    \label{fig:li_fabrication}
\end{figure}

\subsection{AI-Enabled MEMS-based Wearables}

Recent research has increasingly focused on enhancing real-time and remote monitoring capabilities for cardiorespiratory diagnostics. S.~Hoon Lee et al.~\cite{C71ref107} developed a soft, skin-conformal wearable stethoscope using nanomaterial printing of silicone elastomers and conductive hydrogels. This device employed convolutional neural networks (CNNs) for real-time sound classification, achieving high diagnostic accuracy with effective noise suppression (Figure~\figsubref{fig:ai_mems_wearable}{a}--\figsubref{fig:ai_mems_wearable}{d}). Meanwhile, S.~Hyun Lee et al.~\cite{C71ref108} introduced a flexible lung sound monitoring patch equipped with AI-based breath sound analysis. Decision tree and support vector machine (SVM) algorithms were implemented to detect wheezes and classify respiratory patterns in real time (Figure~\figsubref{fig:ai_mems_wearable}{e}--\figsubref{fig:ai_mems_wearable}{f}). In another work, Baraeinejad et al.~\cite{C71ref109} presented a multifunctional digital stethoscope integrating MEMS microphones with Internet of Things (IoT) connectivity. By incorporating digital signal processing and machine learning algorithms, the device enhanced diagnostic accuracy, noise management, and remote clinical accessibility. Collectively, these innovations underscore the evolution of wearable MEMS technologies toward miniaturization, real-time AI processing, and remote healthcare applications.

\begin{figure}[H]
    \centering
    \includegraphics[width=0.85\linewidth]{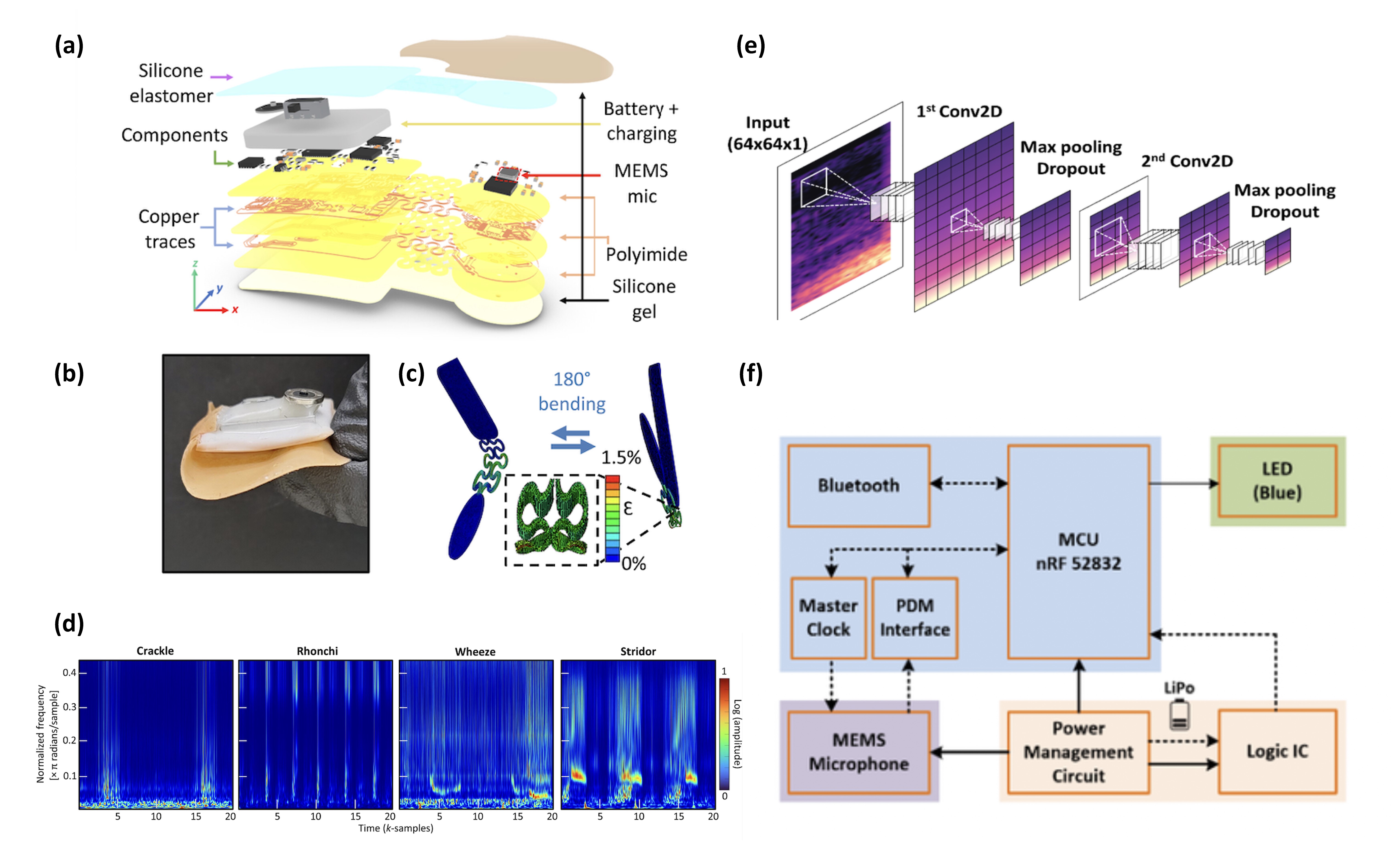}
    \captionsetup{justification=justified,font=footnotesize}
    \caption{ Examples of Intelligent Wearables: \textbf{(a)} Exploded view of soft wearable sensor by S.~Hoon Lee et al.~\cite{C71ref107}; 
             \textbf{(b)} Sensor under 180$^\circ$ bending~\cite{C71ref107}; 
             \textbf{(c)} Simulation results showing cyclic bending durability~\cite{C71ref107}; 
             \textbf{(d)} Spectrogram of cardiorespiratory sounds (crackle, rhonchi, wheeze, stridor)~\cite{C71ref107}; 
             \textbf{(e)} Embedded system architecture of sound monitoring patch by S.~Hyun Lee et al.~\cite{C71ref108}; 
             \textbf{(f)} Deep learning model architecture for respiratory sound classification~\cite{C71ref108}. 
             }
    \label{fig:ai_mems_wearable}
\end{figure}

\subsection{Conclusion}

MEMS-based acoustic biosensors are considered a pivotal technology in modern cardiorespiratory monitoring, offering significant advantages in sensitivity, miniaturization, and integration with wearable platforms. Over the past decade, innovations in MEMS transducer architecture, flexible substrates, and digital communication protocols have enabled high-fidelity acquisition of heart and lung sounds. Despite these advancements, several challenges remain, including motion-related interference, susceptibility to environmental noise, fabrication costs, and long-term biocompatibility for wearable use. The fusion of MEMS hardware with AI presents a transformative opportunity for real-time, automated diagnostics; however, fully autonomous clinical interpretation is yet to be realized. Future development will rely on interdisciplinary collaboration, combining advancements in microfabrication, digital signal processing, and generative AI to create intelligent auscultation devices that deliver reliable clinical insights beyond traditional examination methods.

\section{Quantum Biosensors on Chip}

\subsection{Introduction}
Quantum biosensors apply quantum mechanical principles to improve the detection of biochemical signals at low concentrations. Classical chip-based biosensing first advanced through electronic integrated circuits (EICs), which enabled early signal transduction using electrical measurements. However, issues such as electrical noise and thermal drift limited further sensitivity. The shift from EICs to photonic integrated circuits (PICs) introduced optical signal processing with higher stability and reduced interference. Quantum sensing builds upon this progression through quantum-based concepts, such as superposition and coherence. Quantum dots (QDs) utilize size-dependent optical emission for biological detection, while nitrogen-vacancy (NV) centers in diamond enable spin-state measurement for nanoscale magnetic and thermal sensing. PICs provide a suitable base for integrating these quantum elements into compact chip structures. The combination of quantum sensing with integrated photonics defines the direction toward integrated quantum photonics (IQPs). This chapter introduces fundamental quantum concepts, examines QD and NV biosensing, and outlines developments toward IQP-based biosensing systems.
\subsection{Fundamentals of Quantum Concepts}
In quantum information, a fundamental concept is the quantum bit (qubit). Unlike classical bits, which exist only in definite states $|0\rangle$ or $|1\rangle$, a qubit can exist in a linear combination of both states. This phenomenon is called superposition. Superposition allows quantum systems to encode more information than classical devices, which is essential in quantum-accelerated tasks such as protein folding \cite{C72ref20} and molecular structure search \cite{C72ref23}. Another cornerstone is entanglement, in which two or more particles share a quantum state that cannot be described independently. Such entangled states form the foundation for quantum communication protocols, which are increasingly being studied for biomedical applications requiring precise, correlated measurements \cite{C72ref25}. Another fundamental idea is the Heisenberg uncertainty principle, which imposes a lower bound on the simultaneous precision with which conjugate variables (e.g., position and momentum) can be known. This is fundamental to quantum measurement theory and directly influences the noise floor of biosensors and imaging systems \cite{C72ref27}. This intrinsic uncertainty enables quantum-enhanced metrology. For instance, by exploiting squeezed states, quantum sensors can detect tiny magnetic or electric field fluctuations within biological tissues \cite{C72ref29}. The time evolution of quantum systems is governed by the Schrödinger equation. Solutions to this equation form the basis for simulating quantum dynamics in biochemical systems, including reaction mechanisms and metabolic pathways \cite{C72ref32}.

One of the most well-established quantum effects in biology is quantum tunnelling, where subatomic particles cross energy barriers that are classically forbidden. In enzymatic catalysis, tunnelling enhances reaction rates far beyond what thermal activation can explain \cite{C72ref35}. Quantum models have also revisited proton tunnelling in DNA \cite{C72ref36,C72ref37}. This has implications for cancer development \cite{C72ref39} and genomic stability \cite{C72ref40}. Meanwhile, quantum coherence refers to the ability of a quantum system to maintain phase relationships between its states over time. This phenomenon is used in quantum sensors, particularly those based on NV centers, which exhibit extreme sensitivity to electromagnetic fields \cite{C72ref43}. These technologies are being explored for biomedical applications, such as magnetoencephalography (MEG) \cite{C72ref44}, brain activity monitoring \cite{C72ref45}, and single-cell spectroscopy \cite{C72ref46}. Figure~\ref{fig:quantum_neural_model} illustrates how quantum coherence may play a role in neuronal function. The diagram shows quantum interactions between protein-based qubits within a voltage-gated potassium channel. This suggests that quantum coherence within ion channels could influence neural signalling at the cellular level \cite{C72ref47}. 

\begin{figure}[H]
    \centering
    \includegraphics[width=0.75\linewidth]{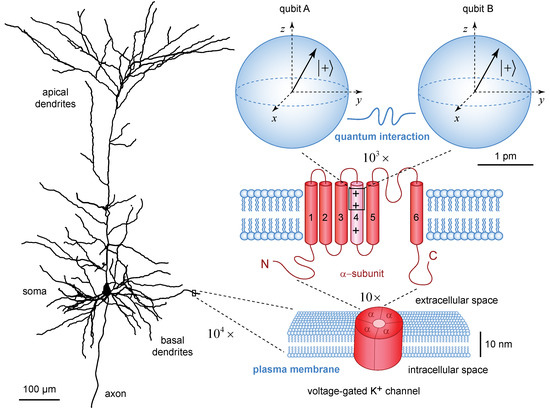}
    \captionsetup{justification=justified,font=footnotesize}
    \captionsetup{font=footnotesize}
    \caption{Quantum model of neural signaling~\cite{C72ref47}. 
    Quantum states inside channel proteins may act like qubits that can interact and remain in superposition. 
    These interactions can influence how ion channels open or close, thereby regulating the flow of potassium ions. 
    Since ion flow shapes neuronal firing and synaptic transmission, quantum effects at this nanoscale may potentially affect how neurons process and transmit information.
    }
    \label{fig:quantum_neural_model}
\end{figure}

\subsection{Quantum Dot and NV Center Biosensors}
Quantum sensing uses quantum properties to measure physical quantities with higher sensitivity than classical sensors. We focus on two quantum sensing approaches: quantum dot and NV center-based sensors. 

QD biosensors (Figure~\ref{fig:qd_basic_structure}) utilize semiconductor nanocrystals for the sensitive detection of biological targets. They emit size-tunable light when excited, and their surfaces can be modified to bind specific biomolecules. When a target molecule interacts with a QD, it changes the dot’s fluorescence, providing an optical readout. Thanks to their stable signal and ability to detect multiple targets at once, QD biosensors are widely used for biomarker detection and molecular imaging. 

\begin{figure}[H]
    \centering
    \includegraphics[width=0.6\linewidth]{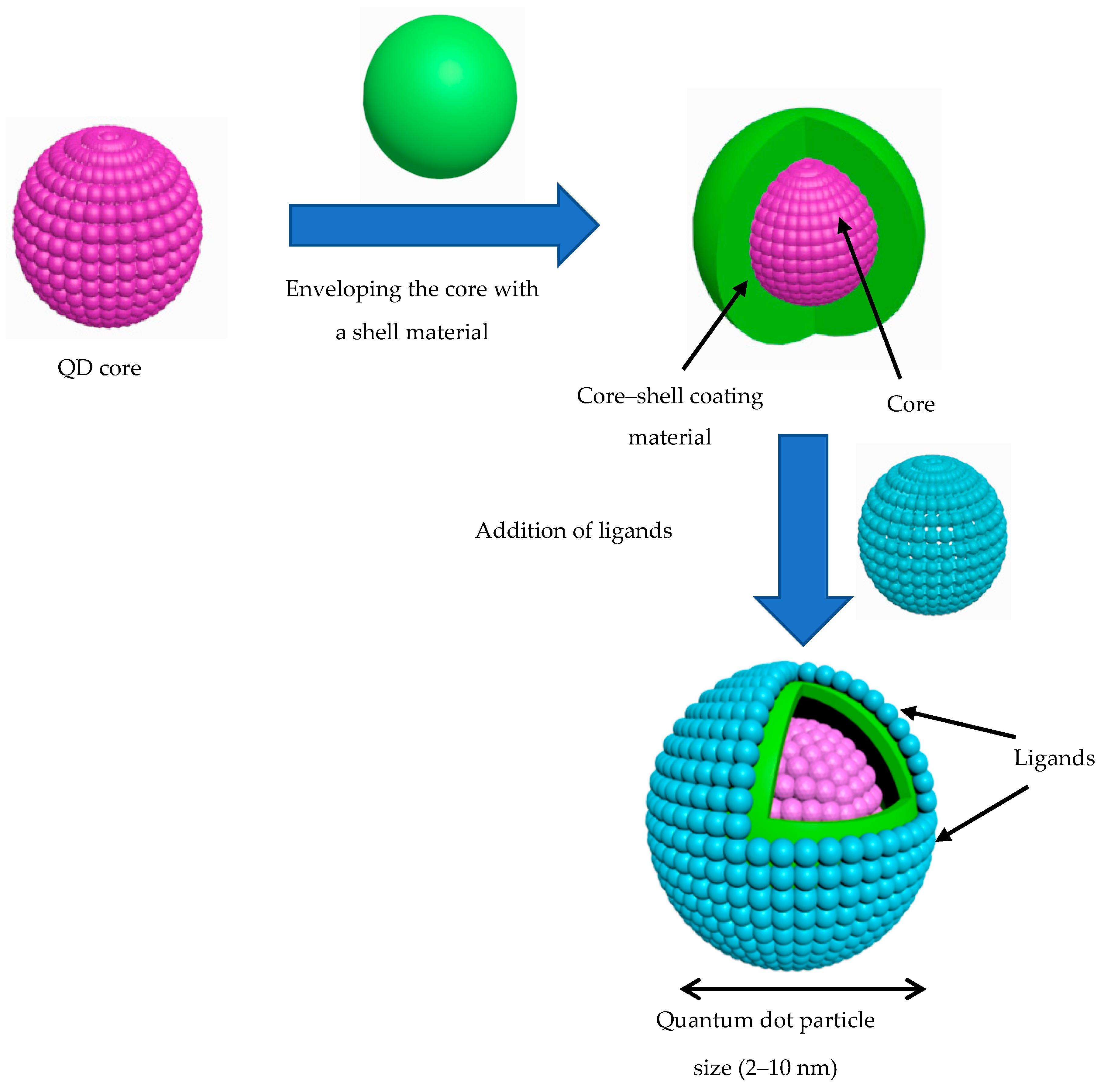}
    \captionsetup{font=footnotesize}
    \caption{Quantum dot basic structure consisting of core, shell, and surface ligands~\cite{C72ref90}.}
    \label{fig:qd_basic_structure}
\end{figure}

NV center biosensors use atomic-scale defects in diamond for highly sensitive detection of biological signals. Figure~\ref{fig:nv_center_schematic} illustrates a typical NV center formation within the diamond lattice. It can exist in either a neutral or a negatively charged state, with the \(\mathrm{NV}^{-}\) state being optically active and sensitive to external perturbations. In operation, a green laser excites the \(\mathrm{NV}^{-}\) center, and the resulting spin-dependent fluorescence is collected by a photodetector. Microwave radiation is applied to manipulate the spin state of the NV center. In practice, the microwaves drive transitions between the spin sublevels. Any external influence, such as a magnetic field or biomolecular interaction, shifts the spin resonance frequency~\cite{C72ref96}. By monitoring these resonance shifts, the NV center functions as a highly sensitive probe of magnetic, thermal, or chemical variations.

\begin{figure}[H]
    \centering
    \includegraphics[width=0.5\linewidth]{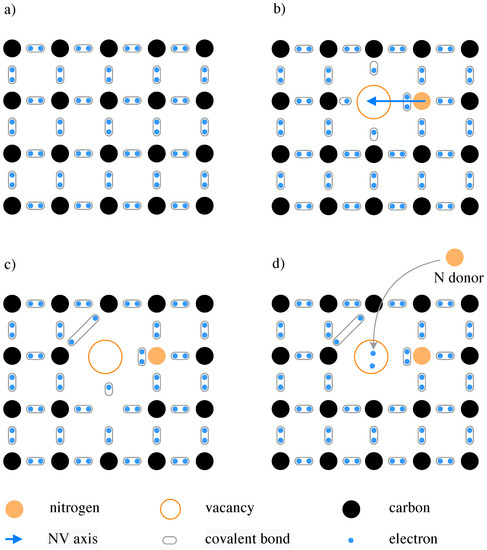}
    \captionsetup{justification=justified,font=footnotesize}
    \caption{NV center formation in diamond~\cite{C72ref96}: 
    \textbf{(a)} Nitrogen atom and nearby vacancy form the NV center; 
    \textbf{(b)} Nitrogen bonds to neighbouring atoms; 
    \textbf{(c)} NV center in the neutral state; 
    \textbf{(d)} Capture of an extra electron from a nitrogen donor produces the negatively charged state.}
    \label{fig:nv_center_schematic}
\end{figure}

The limit of detection (LOD) defines the lowest concentration of analyte that a sensor can detect. In terms of LOD, NV center-based sensors achieve higher sensitivity. Multiplexing capability refers to the sensor’s ability to simultaneously detect multiple targets in a single assay. QD biosensors generally offer superior multiplexing due to their size-tunable, spectrally distinct emissions. QD sensors enable multiplexing but encounter issues with toxicity and long-term stability. NV centers offer room-temperature operation and single-molecule resolution, but their performance is limited by photon collection efficiency and fabrication complexity. Recent studies (Table~\ref{tab:qd_nv_table}) demonstrate that both QD and NV center biosensors offer strong potential for practical applications, such as disease detection, molecular diagnostics, and single-molecule bioimaging.

\subsection{Photonic Integrated Circuits for Biosensing}

Electronic Integrated Circuits (EICs), particularly those based on CMOS technology, have long served as the foundation for biosensing devices. MEMS technologies have further advanced the performance of EIC biosensors, offering superior sensitivity, greater miniaturization, and higher power efficiency. Photonic Integrated Circuits (PICs), on the other hand, process information using light instead of electrical signals. Silicon photonic components are now considered among the most promising platforms for photonic integration (Figure~\figsubref{fig:pic_figures}{c})~\cite{C72ref109}. While EICs offer compact integration and cost-effectiveness, PICs provide greater noise immunity and speed. PIC-based biosensors are commonly realized using waveguide structures, such as ring resonators or interferometers, which monitor changes in refractive index with high sensitivity. A waveguide is a narrow optical channel that directs light along a defined path on the chip. Some of the light travelling in the waveguide extends just outside its surface, forming what is called an evanescent field. When target biomolecules attach to the surface of the waveguide, they interact with this evanescent field and cause a change in the local refractive index. This, in turn, alters the properties of the light travelling through the waveguide, such as its speed or phase. Among the most widely used waveguide-based biosensing structures are micro-ring resonators (MRRs) and Mach–Zehnder interferometers (MZIs)~\cite{C72ref111}. A ring resonator consists of a tiny closed-loop optical waveguide that traps light. It allows only specific wavelengths to resonate constructively. When biomolecules attach to the ring’s surface, they change the local refractive index, causing a shift in the resonance wavelength. An MZI splits light into two separate paths. One path is exposed to the sample, so that biomolecular binding alters its refractive index and optical path length. When the two light beams recombine, any difference in phase between the paths results in constructive or destructive interference, which can be measured as a change in output intensity. Various PIC biosensors based on MZIs and MRRs have been proposed for label-free biosensing (Table~\ref{tab:pic_comparison}). For example, Bryan et al.~\cite{C72ref10} proposed a multiplexed antibody detection MRR biosensor using SiN ring arrays (Figure~\figsubref{fig:pic_figures}{a}, and Figure~\figsubref{fig:pic_figures}{b}). Meanwhile, Vogelbacher et al.~\cite{C72ref113} developed an MZI biosensor with an integrated laser diode for bacterial protein detection (Figure~\figsubref{fig:pic_figures}{d}). In another work, Voronkov et al.~\cite{C72ref119} proposed a dual-ring resonator for liquid sensing (Figure~\figsubref{fig:pic_figures}{e}).

\begin{figure}[H]
    \centering
    \captionsetup{justification=justified,font=footnotesize}
    \includegraphics[width=0.8\linewidth]{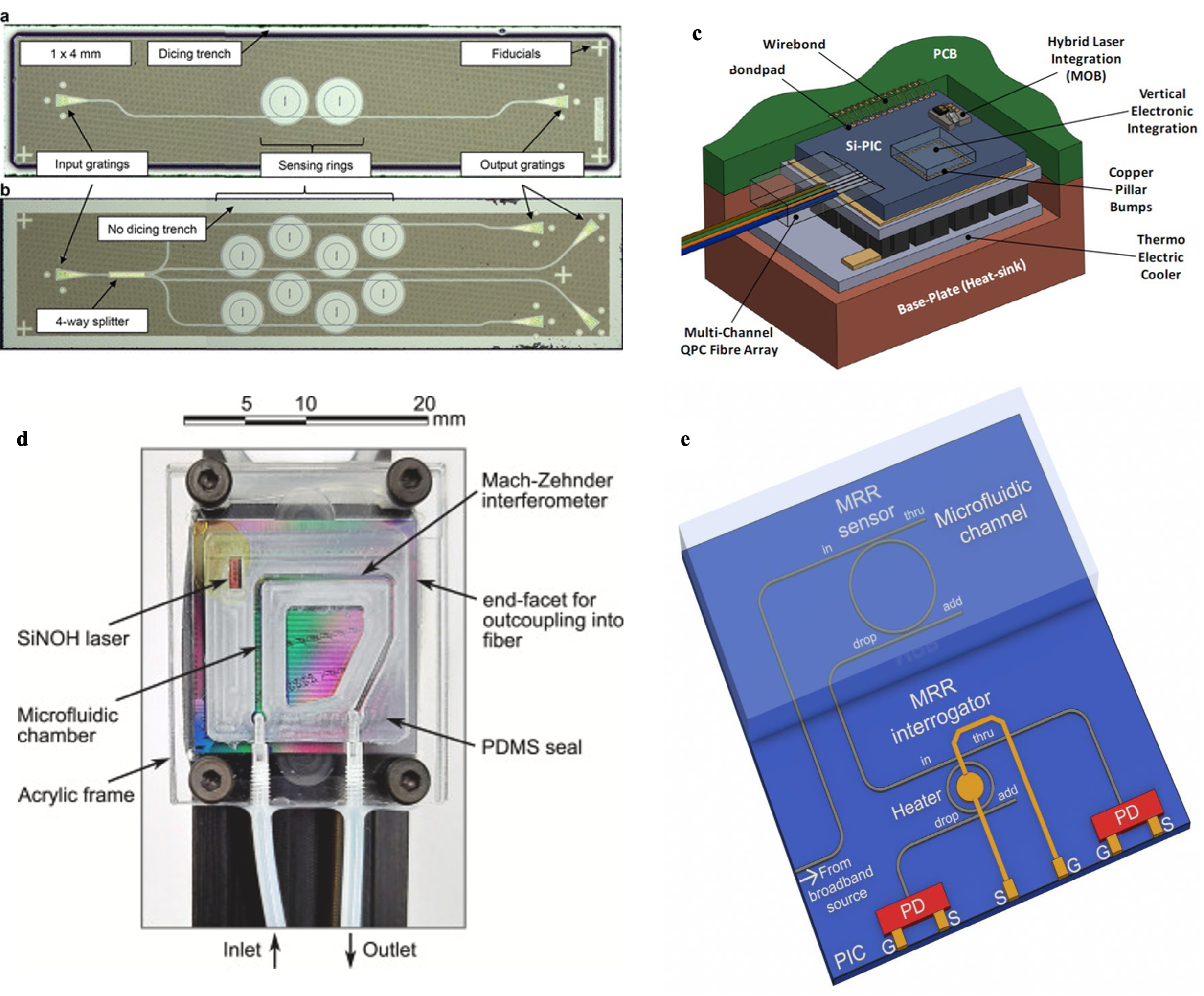}
    \caption{Examples of PIC-based biosensors: 
    \textbf{(a)} Singleplex PIC for antibody profiling, by Bryan et al.~\cite{C72ref10}; 
    \textbf{(b)} Multiplex PIC, by Bryan et al.~\cite{C72ref10}; 
    \textbf{(c)} Schematic of a Silicon PIC packaged with a multi-channel quasi-planar coupled (QPC) fibre-array, a hybrid-integrated laser source based on a micro-optic bench (MOB), and a thermo-electric cooler~\cite{C72ref120}; 
    \textbf{(d)} Implementation of the waveguide MZI photonic lab-on-a-chip biosensor, by Vogelbacher et al.~\cite{C72ref113}; 
    \textbf{(e)} Design draft of the integrated MRR-based photonic sensing system for liquid refractometry, by Voronkov et al.~\cite{C72ref119}.}
    \label{fig:pic_figures}
\end{figure}
\subsection{Toward Integrated Quantum Photonics (IQPs)}
Chip-scale quantum system architectures integrate quantum elements directly onto a single microelectronic chip using advanced semiconductor fabrication. To build compact integrated quantum photonic (IQP) chips, silicon-based approaches have been widely investigated. Despite progress on individual components, achieving full integration likely demands the use of multiple materials. So far, a chip-based IQP remains a scientific challenge. However, with the successful integration of on-chip light sources, a fully integrated silicon-based IQP holds strong potential for future microelectronics~\cite{C72ref121}. Lithium niobate (LiNbO$_3$) is a material used in integrated photonics, along with diamond, silicon carbide (SiC), and III--V semiconductors such as gallium arsenide (GaAs), indium phosphide (InP), and gallium nitride (GaN). Silicon nanostructures, in particular, are now recognized as leading materials for quantum photonic technology~\cite{C72ref122}. IQP devices, including quantum light sources, phase shifters, and single-photon detectors, form the essential building blocks for on-chip quantum information processing and sensing. Among integrated quantum photonic devices, quantum light sources and phase shifters are already compatible with microelectronic fabrication and can be integrated on silicon chips. Single-photon detectors (SPDs) can also be integrated, especially single-photon avalanche diodes (SPADs), though superconducting nanowire single-photon detectors (SNSPDs) often require cryogenic cooling, which limits their practical use. Despite these advances, many challenges remain for fully integrated quantum photonic sensors. For instance, an innovative design~\cite{C72ref133} using asymmetric waveguides for bidirectional transmission on printed circuit boards (PCBs) suggests a potential pathway to connect IQP with conventional microelectronics, indicating that significant work is still needed in this growing field~\cite{C72ref134}. In 2025, Kramnik et~al.~\cite{C72ref18} demonstrated the first commercial integrated quantum photonic chip that combines electronic, photonic, and quantum components on a single silicon platform, a major step toward IQP biosensors (Figure~\ref{fig:iqp_kramnik}).

\begin{figure}[H]
  \centering
  \includegraphics[width=\linewidth]{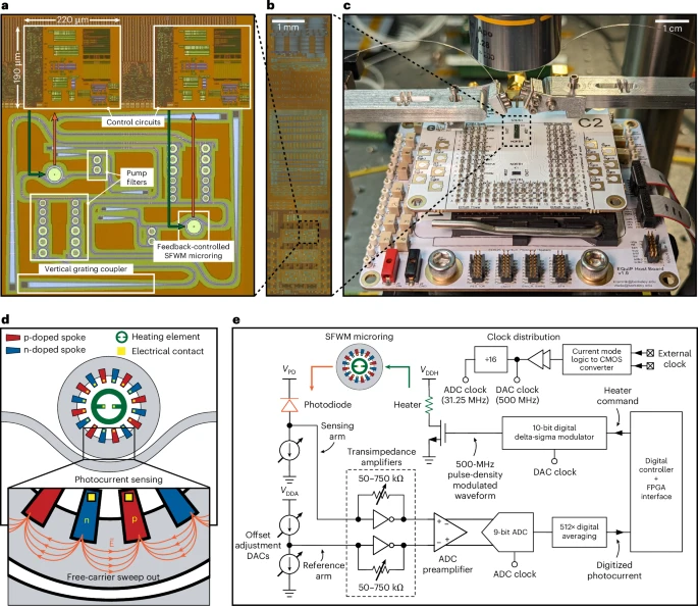}
\captionsetup{justification=justified,font=footnotesize}
  \caption{First-time electronic--photonic quantum chip, manufactured by Kramnik et~al.~\cite{C72ref18}:
  \textbf{(a)} Micrograph of electronic--photonic quantum circuit blocks with control circuits and microrings;
  \textbf{(b)} Entire 2\,mm$\times$9\,mm electronic--photonic CMOS chip;
  \textbf{(c)} CMOS chip bonded to a PCB providing power and interfaces to an FPGA;
  \textbf{(d)} Diagram of the microring with diodes for carrier sweep-out in reverse bias;
  \textbf{(e)} Schematic of the feedback-controlled microring pair source with on-chip thermal tuning and in-cavity monitoring via photodiodes.}
  \label{fig:iqp_kramnik}
\end{figure}

\subsection{Conclusion}

Quantum biosensors demonstrate clear advantages over classical sensing technologies by enabling ultra-high sensitivity, single-molecule resolution, and low-noise measurement. When evaluated alongside electronic integrated circuits (EICs) and photonic integrated circuits (PICs), each platform reveals distinct strengths and constraints. EIC biosensors offer cost efficiency and compactness but remain limited by electrical noise and multiplexing challenges. PIC-based sensors overcome many electrical constraints through optical immunity and scalable multiplexing, yet they face fabrication complexity, optical alignment issues, and nonlinear effects. Integrated quantum photonics (IQP) provides a promising pathway forward by combining discrete quantum states with photonic stability, although current IQP implementations still confront obstacles such as coherence loss, photon coupling inefficiencies, and reliance on laboratory-scale infrastructure.

Future progress toward practical IQP biosensors will depend on the co-integration of light sources, detectors, and microfluidic interfaces within CMOS-compatible platforms. Recent advances in room-temperature materials, such as diamond NV centers and silicon-based SPAD arrays, are beginning to reduce system complexity, while hybrid material strategies involving lithium niobate, silicon carbide, or two-dimensional quantum emitters may improve photon extraction and stability. Continued innovation in packaging, phase modulation, and three-dimensional integration will be critical for transitioning IQP biosensing from experimental prototypes to clinically deployable systems. With sustained development, integrated quantum photonics holds strong potential to redefine on-chip diagnostics.

        \setcounter{figure}{0}
        \setcounter{equation}{0}
        \setcounter{table}{0}
  \chapter{Heart and Lung Sounds Dataset Recorded from a Clinical Manikin}

This chapter is based on the following publication:

\vspace{5mm}
\noindent
\cite{10981596} Y. Torabi, S. Shirani and J. P. Reilly, ``Descriptor: Heart and Lung Sounds Dataset Recorded From a Clinical Manikin Using Digital Stethoscope (HLS-CMDS)," in \textit{IEEE Data Descriptions}, vol. 2, pp. 133-140, 2025, \href{https://doi.org/10.1109/IEEEDATA.2025.3566012}{doi:10.1109/IEEEDATA.2025.3566012}.

\newpage
Heart and lung sounds play a vital role in healthcare monitoring. Advances in stethoscope technology allow for the precise recording of these physiological signals. In this dataset, we used a digital stethoscope to capture heart and lung sounds from a clinical manikin. To our knowledge, this is the first dataset to include both separated and combined cardiorespiratory recordings. The recordings feature both normal and abnormal cases, and they were taken at multiple chest locations. This dataset supports research in disease prediction and deep learning–based audio signal analysis.

\section{Introduction}
Auscultation of heart and lung sounds provides clinical insight into cardiorespiratory health. However, the scarcity of high-quality datasets limits the performance of AI-based diagnostic systems. To address this gap, this work presents the \textit{Heart and Lung Sounds Clinical Manikin Dataset (HLS-CMDS)}, a collection of recordings acquired using a digital stethoscope from a clinical simulation manikin. The dataset’s design ensures controlled recording conditions and diverse pathological coverage, enabling machine-learning applications such as classification, clustering, and source separation.

\section{Collection Methods and Design}

\subsection{Setup and Environment}

Figure~\ref{fig:setup} illustrates the data acquisition process. We recorded sounds from various chest locations, depending on whether heart or lung sounds were being captured. The manikin reproduces heart and lung sounds through anatomically positioned speakers, ensuring that sound characteristics vary naturally across locations. Recording from multiple locations is essential for AI model generalization. Labelling recordings by site enables analysis of how auscultation position affects sound perception, supporting spatially aware AI models for robust classification and source separation. 

\begin{figure}[H]
    \centering
    \includegraphics[width=0.75\linewidth]{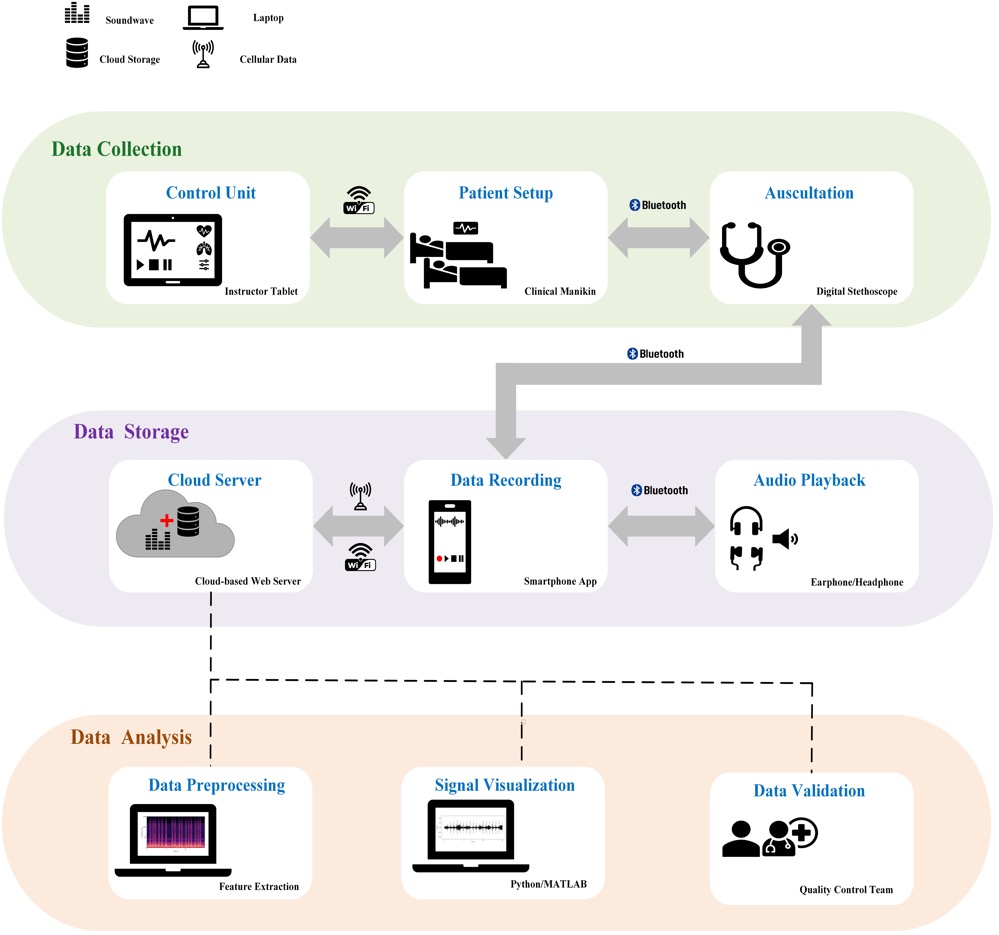}
    \captionsetup{font=footnotesize}
    \caption{Diagram of the data acquisition process.}
    \label{fig:setup}
\end{figure}

Auscultation was performed in a controlled, quiet, and noise-free environment (Figure~\figsubref{fig:setup_env}{a}). The patient simulator was positioned in a sitting posture to emulate realistic conditions for lung and heart sound recordings (Figure~\figsubref{fig:setup_env}{b}).

\begin{figure}[H]
    \centering
    \includegraphics[width=0.5\linewidth]{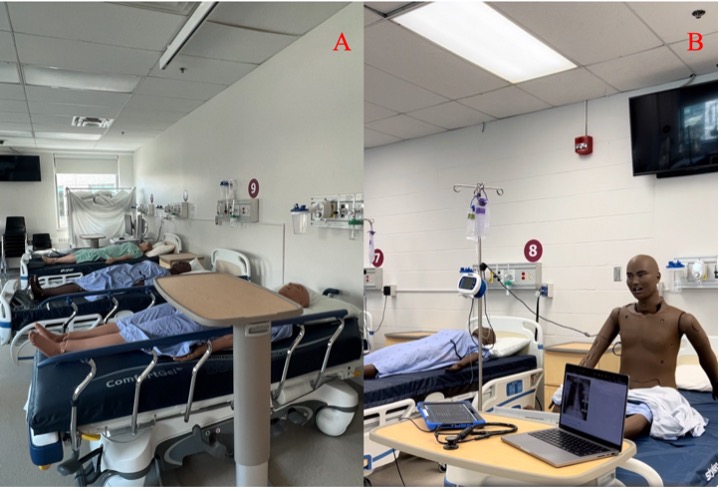}
    \captionsetup{justification=justified,font=footnotesize}
    \caption{Setup and environment: \textbf{(a)} Recording unit; \textbf{(b)} Manikin in a sitting position alongside the recording setup.}
    \label{fig:setup_env}
\end{figure}

\subsection{Clinical Skills Manikin and Control System}

We utilized the CAE Juno\textsuperscript{TM} nursing skills manikin, a mid-fidelity patient simulator designed to enhance clinical nursing skills. We used the CAE Maestro software, installed on a tablet, to control and monitor the manikin in real-time. The tablet was connected to the manikin via a secure Wi-Fi connection.

For mixture recordings, both heart and lung sounds were enabled and recorded simultaneously. For each mixture, individual heart and lung sounds were also recorded separately without altering the stethoscope position. Because the manikin repeatedly plays the same pre-recorded sounds, this approach is effectively equivalent to true simultaneous recording. Hence, our dataset provides a reliable ground truth for source separation algorithms, as each mixture contains the same heart and lung components captured individually. For single-type recordings, we first enabled only the heart sounds and recorded across standard auscultation landmarks, capturing a variety of normal and abnormal cases, including: normal, late diastolic murmur, mid-systolic murmur, late systolic murmur, atrial fibrillation, fourth heart sound (S4), early systolic murmur, third heart sound (S3), tachycardia, and atrioventricular block. Next, we enabled only lung sounds and recorded normal, wheezing, fine crackle, rhonchi, pleural rub, and coarse crackle sounds. Each sound type was collected from multiple locations.

\subsection{Digital Stethoscope and Recording}

We used the 3M\textsuperscript{TM} Littmann\textsuperscript{\textregistered} CORE Digital Stethoscope. Each recording lasted 15\,s. The stethoscope includes built-in frequency filters for selective recording of heart, lung, or mixed sounds. The amplification feature provided up to $\times$40 sound enhancement, while active noise cancellation reduced ambient interference. The stethoscope was connected to the Eko mobile application via Bluetooth for real-time monitoring, visualization, and cloud-based storage of the recordings in \texttt{.wav} format. These files were later downloaded for detailed analysis on a personal computer or laptop.

\subsection{Noise Handling}

In real-world auscultation, recordings are often affected by multiple noise sources, including ambient noise, stethoscope friction noise, and motion artifacts from patient movement. Biological noise from muscle contractions and respiration can also obscure cardiorespiratory sounds. Unlike human subjects, the manikin eliminates biological noise and motion artifacts. Recordings were performed in a quiet, controlled environment to suppress ambient interference. The stethoscope was held steadily to minimize motion and friction noise. Furthermore, the built-in frequency filters reduced irrelevant noise during acquisition.

\section{Validation and Quality}

The pre-recorded sounds used in this study were curated by CAE Healthcare, originating from real patients and synchronized with physiological parameters. In addition to this built-in validation, we conducted post-recording verification by inspecting the waveforms to confirm the absence of signal artifacts. The nursing team assisted in identifying correct auscultation landmarks and ensured adherence to clinical protocols during recording sessions. Recordings were carried out in a quiet, noise-free environment with the manikin positioned upright to minimize environmental interference. Before each recording, the stethoscope’s diaphragm was carefully placed at standard auscultation points to avoid motion artifacts. After data collection, all audio files were qualitatively inspected by listening to and comparing them against the expected patterns of heart and lung sounds. We selected a digital stethoscope with high amplification gain, active noise cancellation, and built-in frequency filtering to achieve minimal distortion. Bluetooth-based data transfer ensured reliable communication and secure storage without packet loss or corruption~\cite{c3ref28}. The selected manikin closely replicates real clinical scenarios and is approved by the U.S. National Council of State Boards of Nursing~\cite {c3ref29}. Routine maintenance and annual calibration performed by the simulation lab staff ensure the accuracy and reliability of the manikin’s sound system over time. The reproducibility of the dataset is further supported by the standardized and repeatable recording procedures used in this study. 

\section{Records and Storage}

The dataset comprises a total of 535 audio recordings (268 female, 267 male) in \texttt{.wav} format, categorized into three primary groups: 50 heart sound recordings (\texttt{HS.zip}), 50 lung sound recordings (\texttt{LS.zip}), and 3$\times$145 recordings of mixed heart and lung sounds, along with their individual heart and lung source sounds (\texttt{Mix.zip}). Each category includes a corresponding metadata file (\texttt{HS.csv}, \texttt{LS.csv}, \texttt{Mix.csv}) describing the associated audio recordings, including file name, gender, sound type, and anatomical location. The dataset includes ten heart sound types and six lung sound types, recorded from twelve anatomical locations. 

\section{Conclusion}

This dataset is suitable for a wide range of research applications, particularly for tasks involving the separation and analysis of heart and lung sounds in machine learning pipelines. Simulator-based datasets play a crucial complementary role in AI-driven auscultation. Our dataset provides a structured foundation for developing and benchmarking algorithms while underscoring the need for hybrid datasets and clinical validation. Beyond research, the dataset supports educational and clinical applications, such as feature extraction, acoustic modelling, and stethoscope sound analysis training. It enriches the limited pool of publicly available cardiorespiratory datasets and offers a high-quality resource for machine learning, signal processing, and healthcare education.

        \setcounter{figure}{0}
        \setcounter{equation}{0}
        \setcounter{table}{0}    
  \chapter{Generative AI Algorithms for Blind Source Separation}

This chapter is currently under review in the following publication:

\vspace{5mm}
Y. Torabi, S. Shirani, and J. P. Reilly, 
``Large Language Model-Enhanced Non-negative Matrix Factorization for Cardiorespiratory Sound Separation,'' 
in \textit{Biomedical Signal Processing and Control}, Elsevier.

\newpage
Generative artificial intelligence (GenAI) is becoming an important approach in biomedical signal processing, offering new ways to separate mixed physiological sounds. This chapter presents two GenAI approaches for blind source separation of heart and lung sounds. The first method, called LingoNMF, combines large language models with non-negative matrix factorization to improve unsupervised separation through prompting. The second method, called XVAE-WMT, is based on explainable artificial intelligence (XAI), which aims to make internal representations interpretable. In this approach, we analyze the latent space of a variational autoencoder to understand how the model distinguishes cardiorespiratory features, using temporal-consistency loss and wavelet inputs to preserve the structure of the signals. Together, these algorithms demonstrate how generative algorithms can support interpretable analysis of cardiorespiratory recordings.
\section{LLM-based Non-Negative Matrix Factorization}
\subsection{Introduction}
In this study, a modified version of the affine non-negative matrix factorization (NMF) approach, termed LingoNMF, is proposed to blindly separate lung and heart sounds recorded by a digital stethoscope. This method incorporates a parallel structure of multilayer units applied to the input signal and utilizes the periodic property of the signals to improve accuracy. Furthermore, LLaMA (Large Language Model Meta AI) is employed in two ways: 1) it provides an adjunct diagnostic facility of potential clinical value, and 2) it is used to optimize a fundamental frequency estimate.

\subsection{Methodology}
\subsubsection{Theoretical Background}
The standard non-negative matrix factorization (NMF) problem~\cite{Cichocki2009} states that given a non-negative data matrix $\mathbf{Y} \in \mathbb{R}_+^{I \times T}$, find two non-negative matrices  
$\mathbf{A} \in \mathbb{R}_+^{I \times J}$ and $\mathbf{X} \in \mathbb{R}_+^{J \times T}$ such that:
\begin{equation}
    \mathbf{Y} = \mathbf{A}\mathbf{X} + \mathbf{E},
\end{equation}
where $\mathbf{E} \in \mathbb{R}^{I \times T}$ represents the approximation error,  
$\mathbf{A}$ is the mixing matrix, and $\mathbf{X}$ is the source signal.
In multi-layer NMF, the basis matrix $\mathbf{A}$ is replaced by a set of cascaded matrices. 
Initially, the input matrix $\mathbf{Y}$ is approximated as $\mathbf{Y} \approx \mathbf{A}_1 \mathbf{X}_1$. 
Subsequently, the intermediate output $\mathbf{X}_1$ is treated as a new input and further decomposed as 
$\mathbf{X}_1 \approx \mathbf{A}_2 \mathbf{X}_2$. 
This iterative process continues, sequentially factoring each output into new components until a predefined stopping criterion is met. 
The final model can thus be expressed as Equation~\ref{eq:12}, where $L$ denotes the number of layers:
\begin{equation}
    \mathbf{Y} \approx \mathbf{A}_1 \mathbf{A}_2 \cdots \mathbf{A}_L \mathbf{X}.
    \label{eq:12}
\end{equation}

\subsubsection{PL-NMF Algorithm}
Figure~\figsubref{fig:plnmf}{a} shows an overview of our periodicity-based multilayer NMF algorithm (PL-NMF). The hierarchical multi-stage procedure reduces the risk of converging to local minima (Figure~\figsubref{fig:plnmf}{b}).  
\begin{figure}[H]
    \centering
    \includegraphics[width=0.9\linewidth, trim=500 200 0 200, clip]{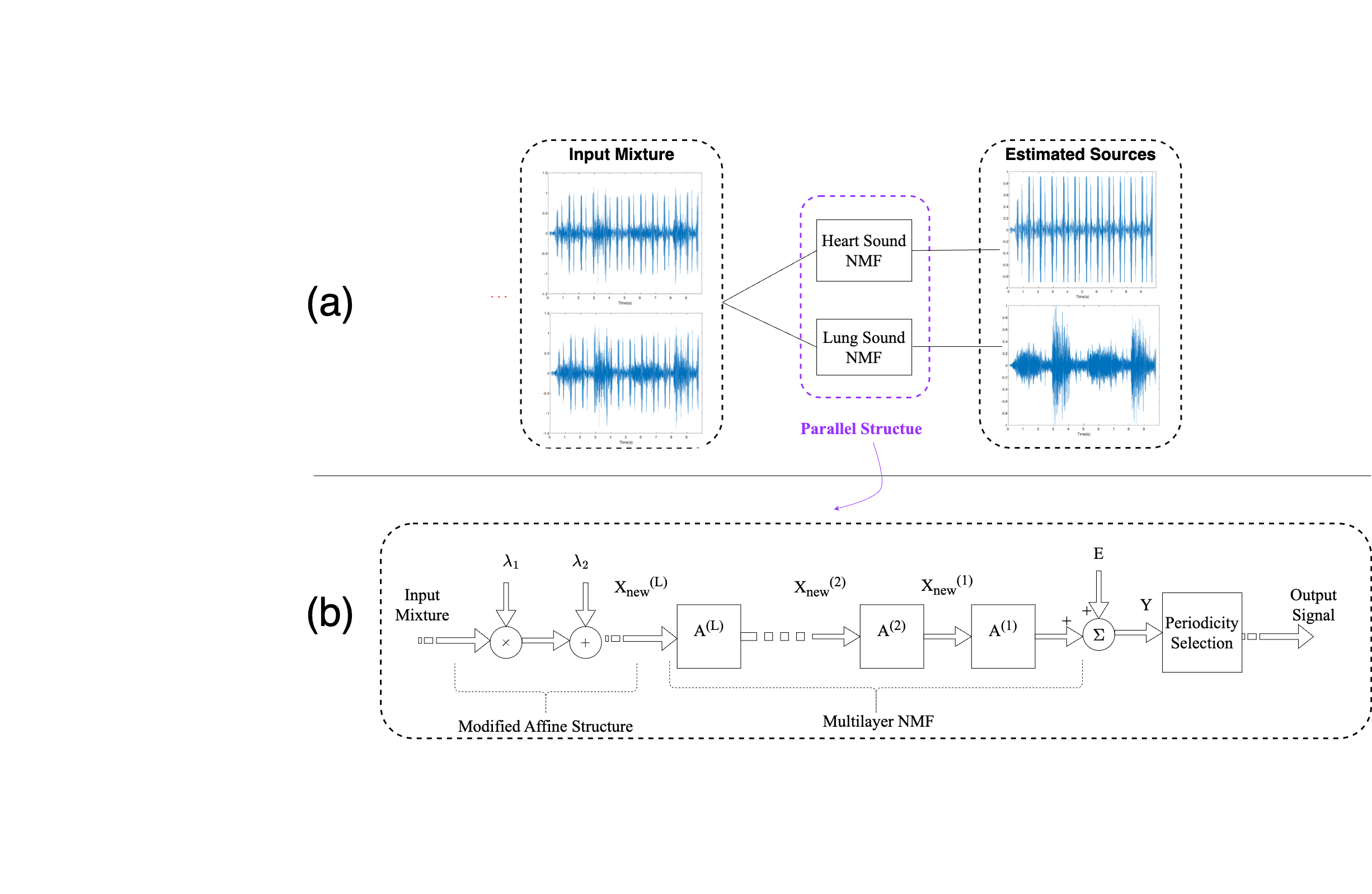}
    \captionsetup{justification=justified,font=footnotesize}
    \caption{PL-NMF overview: \textbf{(a)} Proposed algorithm; \textbf{(b)} Internal structure of the “Heart Sound NMF” and “Lung Sound NMF” blocks.}
    \label{fig:plnmf}
\end{figure}
Both the “Heart Sound NMF” and “Lung Sound NMF” blocks use the same internal multilayer structure. The distinction between the two blocks lies in their parameter settings, such as the scaling factor, offset, and number of layers. Our proposed method employs two NMF modules to separate the sources independently. Each module generates two signals, and one signal from each module is discarded based on the periodic nature of heart and lung sounds. Specifically, one module selects the output with a shorter period, corresponding to the heart sound, while the other module selects the output with a longer period, corresponding to the lung sound. This dual-module structure provides additional degrees of freedom, allowing independent tuning of parameters for each block. Algorithm~\ref{alg:plnmf} (see Appendix) summarizes the PL-NMF method. 

\textbf{Modified Affine NMF}:
The recorded heart and lung sounds may have negative values. To address this, an affine NMF approach introduces an offset to transform the input mixtures into non-negative values. Here, we propose a modified transformation by introducing an additional scaling factor $\lambda_1$ to enhance performance. 
Incorporating biological prior knowledge can further guide parameter selection. 
Since lung sounds vary more slowly and are generally more dominant than heart sounds, 
we assign two distinct scaling factors: $\lambda_{1,\text{lung}} \in (0,1)$ for the lung-detection block 
and $\lambda_{1,\text{heart}} \ge 1$ for the heart-detection block. 
Additionally, an offset term $\lambda_2$ is introduced to ensure non-negativity of the final signal:
\begin{equation}
    \lambda_{1,p} \times \min(x) + \lambda_{2,p} \ge 0, \qquad 
    p \in \{\text{heart}, \text{lung}\}.
    \label{eq:11}
\end{equation}

\textbf{Periodicity-based Parallel Analysis}: We estimate the period using the autocorrelation function. First, we compute the autocorrelation of the signal of length $T$ for all possible lag shifts $P$ (Equation~\ref{eq:13}). Next, we estimate the period $p$ by calculating the average distance between consecutive peaks $\Delta P_k$ (Equation~\ref{eq:14}):
\begin{equation}
    ACF(P) = \frac{1}{T} \sum_{t=1}^{T-P} x_t \cdot x_{t+P}, 
    \quad \forall P \in \{1,2,\ldots,T\}.
    \label{eq:13}
\end{equation}
\begin{equation}
    p = \frac{1}{N} \sum_{k=1}^{N-1} \Delta P_k.
    \label{eq:14}
\end{equation}

\subsubsection{LingoNMF Algorithm}
LingoNMF (Figure~\ref{fig:lingo}) incorporates a feedback loop driven by LLaMA. It is pretrained on a diverse collection of writings to effectively generate human-like texts~\cite{Touvron2023}.

\begin{figure}[H]
    \centering
    \includegraphics[width=0.9\linewidth, trim= 0 10 0 10, clip]{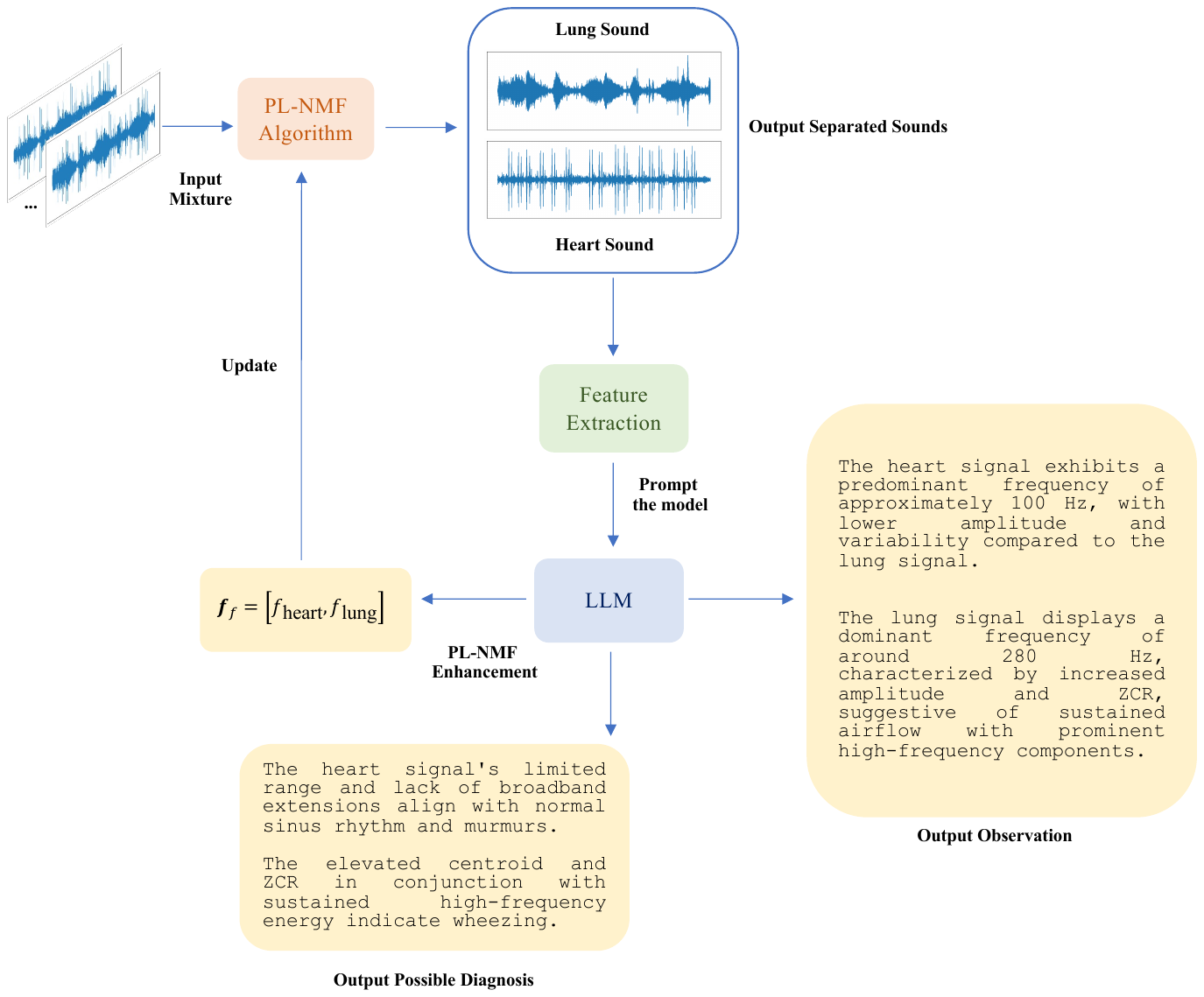}
    \captionsetup{justification=justified,font=footnotesize}
    \caption{LingoNMF overview. The input mixtures are decomposed by the PL-NMF. Extracted features are used to prompt the LLM, which outputs suggested fundamental frequency $\mathbf{f}_f$ to update the penalty term, feature-based observations, and possible diagnosis based on the observations.}
    \label{fig:lingo}
\end{figure}

The LLM plays two key roles in this framework. First, it analyzes the extracted features of the separated outputs, providing detailed insights for tasks such as anomaly detection and diagnostic interpretation. Second, the LLM feedback assists in fine-tuning a frequency-based penalty term. Algorithm~\ref{alg:lingonmf} represents LingoNMF (see Appendix). 
The fundamental frequency $f_f$ is the primary frequency at which a signal oscillates. It is estimated as the frequency $\widehat{f}_f$ (Equation~\ref{eq:16}), where $P_x(f)$ is the power spectral density (PSD) of signal $x$ at frequency $f$:
\begin{equation}
    \widehat{f}_f = \arg\max_f P_x(f).
    \label{eq:16}
\end{equation}
We define the modified cost function $D_f(\mathbf{Y}\,||\,\mathbf{A}\mathbf{X})$ as Equation~\ref{eq:15}, 
where $D(\mathbf{Y}\,||\,\mathbf{A}\mathbf{X})$ is the original temporal cost function, 
$\lambda_f$ is the scaling factor, 
$\widehat{\mathbf{f}}_f=[\widehat{f_{f,1}},\ldots,\widehat{f_{f,J}}]^\top$ denotes the vector of fundamental frequencies for each estimated signal 
(i.e., rows of $\mathbf{X}$), 
and $\mathbf{f}_f=[f_{f,1},\ldots,f_{f,J}]^\top$ is the vector of fundamental frequencies for each source. 
\begin{equation}
    D_f(\mathbf{Y}\,||\,\mathbf{A}\mathbf{X}) 
    = D_\alpha(\mathbf{Y}\,||\,\mathbf{A}\mathbf{X}) 
    + \lambda_f \|\widehat{\mathbf{f}}_f - \mathbf{f}_f\|^2 ,
    \label{eq:15}
\end{equation}

where $D_\alpha(\mathbf{Y}\,||\,\mathbf{A}\mathbf{X})$ is the $\alpha$-divergence, which generalizes common NMF cost
functions including the Euclidean, Kullback--Leibler, and Itakura--Saito divergences.
The parameter $\alpha$ controls the trade-off between robustness to noise and data fidelity, with smaller $\alpha$ being more robust and larger $\alpha$ enforcing
sharper reconstruction. Compared to the Frobenius norm, the
$\alpha$-divergence offers more flexible convergence behaviour by adapting to the
statistical structure of non-negative data. In practice, this flexibility improves
separation robustness and stability, particularly under noise and spectral mismatch,
whereas the Frobenius norm is more sensitive to outliers.

We decompose the input mixtures using the PL-NMF algorithm. We extract features, including spectral centroid, root-mean-square (RMS) energy, zero-crossing rate, variance, mean frequency, and maximum amplitude. We prompt the LLM with features, which outputs (i) suggested fundamental frequency $\mathbf{f}_f$, 
(ii) feature-based observations, and (iii) possible diagnoses based on these observations. 
We calculate $\widehat{\mathbf{f}}_f$ from the separated signals $\mathbf{X}$. We initially assume that the fundamental frequencies $\mathbf{f}_f$ correspond to normal heart and lung sounds, 
and subsequently refine these estimates using LLM feedback. We design the prompts to constrain the LLM output into numeric fields (corresponding to the estimated fundamental frequencies of heart and lung) and text fields (for feature-based observations and possible diagnoses). To further reduce hallucinations, we employ a lightweight retrieval-augmented generation (RAG) pipeline grounded in cardiorespiratory acoustics literature (Figure~\ref{fig:rag}). 

\begin{figure}[H]
    \centering
    \includegraphics[width=0.9\linewidth, trim=0 100 0 100, clip]{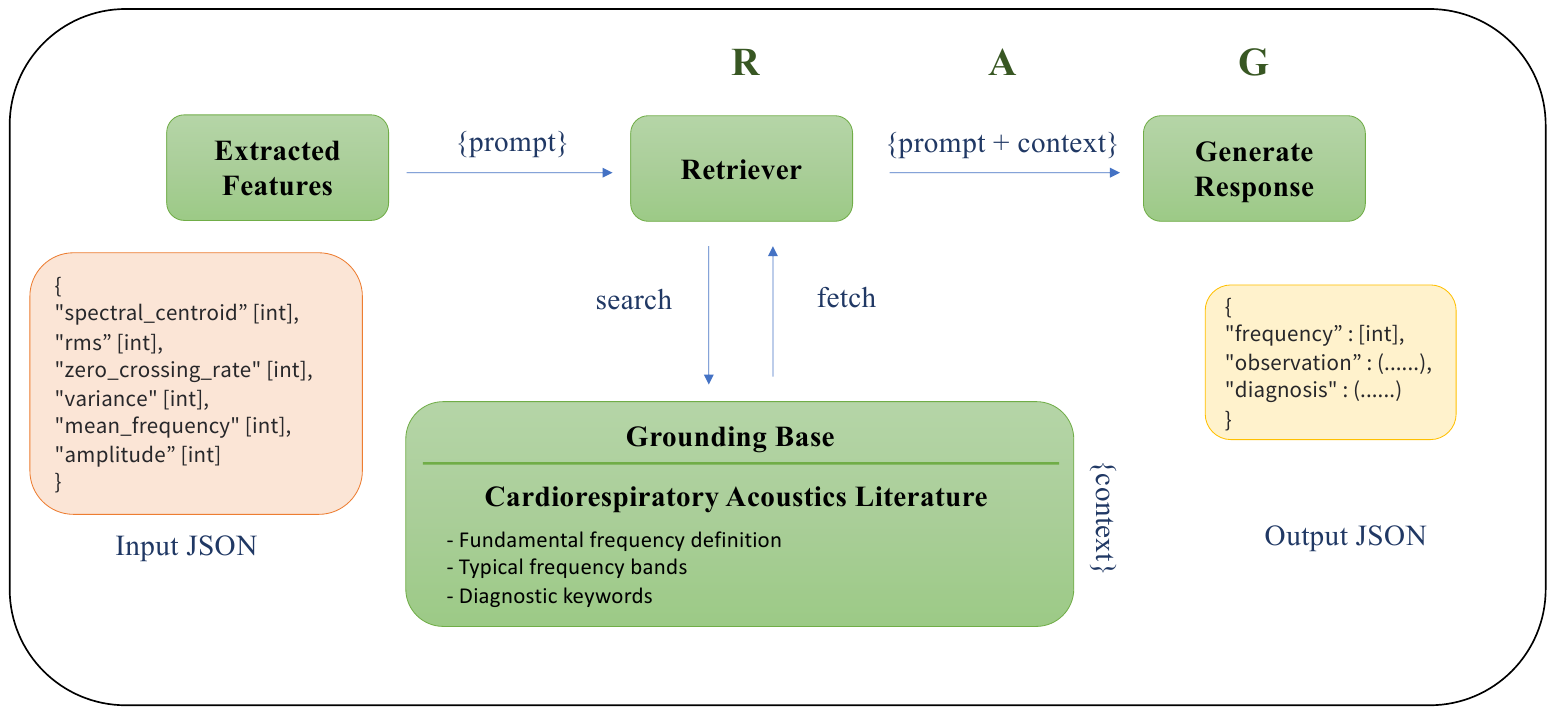}
    \captionsetup{justification=justified,font=footnotesize}
    \caption{Lightweight RAG pipeline. The retriever (R) looks up external cardiorespiratory acoustics knowledge to fetch relevant information. The augmenter (A) adds this retrieved information to the structured prompt. The generator (G) uses the LLM to produce context-aware outputs, returning physiologically plausible fundamental frequencies and diagnostic observations.}
    \label{fig:rag}
\end{figure}
\subsection{Experimental Results}
We use two datasets collected with the digital stethoscope. Dataset One contains 100 cases of two-mixture normal heart and lung sounds. Dataset Two includes 210 clinical manikin recordings, covering normal and abnormal sounds. We use the pure sounds as a reference to compute the separation performance~\cite{Torabi2024Manikin}.
The separated signal $\mathbf{s}$ can be expressed as Equation~\ref{eq:17}:
\begin{equation}
    \mathbf{s} = \mathbf{s}_{\mathrm{target}} + \mathbf{e}_{\mathrm{interference}} 
    + \mathbf{e}_{\mathrm{noise}} + \mathbf{e}_{\mathrm{artifact}} .
    \label{eq:17}
\end{equation}
We measure performance using the signal-to-distortion ratio (SDR), signal-to-interference ratio (SIR), and signal-to-artifact ratio (SAR), defined in Equation~\ref{eq:19}, Equation~\ref{eq:20}, and Equation~\ref{eq:21}, respectively:
\begin{equation}
    \mathrm{SDR} = 10 \log_{10} 
    \frac{\|\mathbf{s}_{\mathrm{target}}\|^2}
         {\|\mathbf{e}_{\mathrm{interference}} + \mathbf{e}_{\mathrm{noise}} + \mathbf{e}_{\mathrm{artifact}}\|^2}.
    \label{eq:19}
\end{equation}
\begin{equation}
    \mathrm{SIR} = 10 \log_{10} 
    \frac{\|\mathbf{s}_{\mathrm{target}}\|^2}
         {\|\mathbf{e}_{\mathrm{interference}}\|^2}.
    \label{eq:20}
\end{equation}
\begin{equation}
    \mathrm{SAR} = 10 \log_{10} 
    \frac{\|\mathbf{s}_{\mathrm{target}} + \mathbf{e}_{\mathrm{interference}} + \mathbf{e}_{\mathrm{noise}}\|^2}
         {\|\mathbf{e}_{\mathrm{artifact}}\|^2}.
    \label{eq:21}
\end{equation}
Figure~\ref{fig:stat} compares the performance of different NMF-based separation methods across two datasets. We conduct paired \textit{t}-tests for each separation metric. The null hypothesis ($H_0$) assumed no difference in mean performance between methods, 
whereas the alternative hypothesis ($H_1$) assumed that LingoNMF achieved higher mean scores. 
We set statistical significance thresholds at $p < 0.05$ and $p < 0.01$, 
corresponding to moderate and strong evidence against $H_0$, respectively. 
LingoNMF outperforms conventional approaches (i.e. Standard NMF and $\alpha$-NMF), 
showing overall gains in SAR ($p < 0.05$) and significant improvements in heart sound SDR and SIR ($p < 0.01$). 
Lung SDR and SIR also improve significantly in Dataset~Two ($p < 0.01$), 
while Dataset~One shows more modest gains ($p < 0.05$). 
The only difference between PL-NMF and LingoNMF is the presence or absence of LLM feedback. 
Therefore, we compare PL-NMF with LingoNMF for the ablation study. 
LingoNMF significantly improves heart sound SIR and SDR in both datasets ($p < 0.01$), 
and also improves lung sound separation in Dataset~Two ($p < 0.05$). 
We also observe SAR improvements for both heart and lung sounds in Dataset~One ($p < 0.01$). 
These results indicate that the LLM feedback primarily benefits heart sound separation, 
likely because heart sounds are more structured and periodic, 
making them more sensitive to accurate frequency-penalty tuning, 
whereas lung sounds are more broadband.

\begin{figure}[H]
    \centering
    \includegraphics[width=0.8\linewidth]{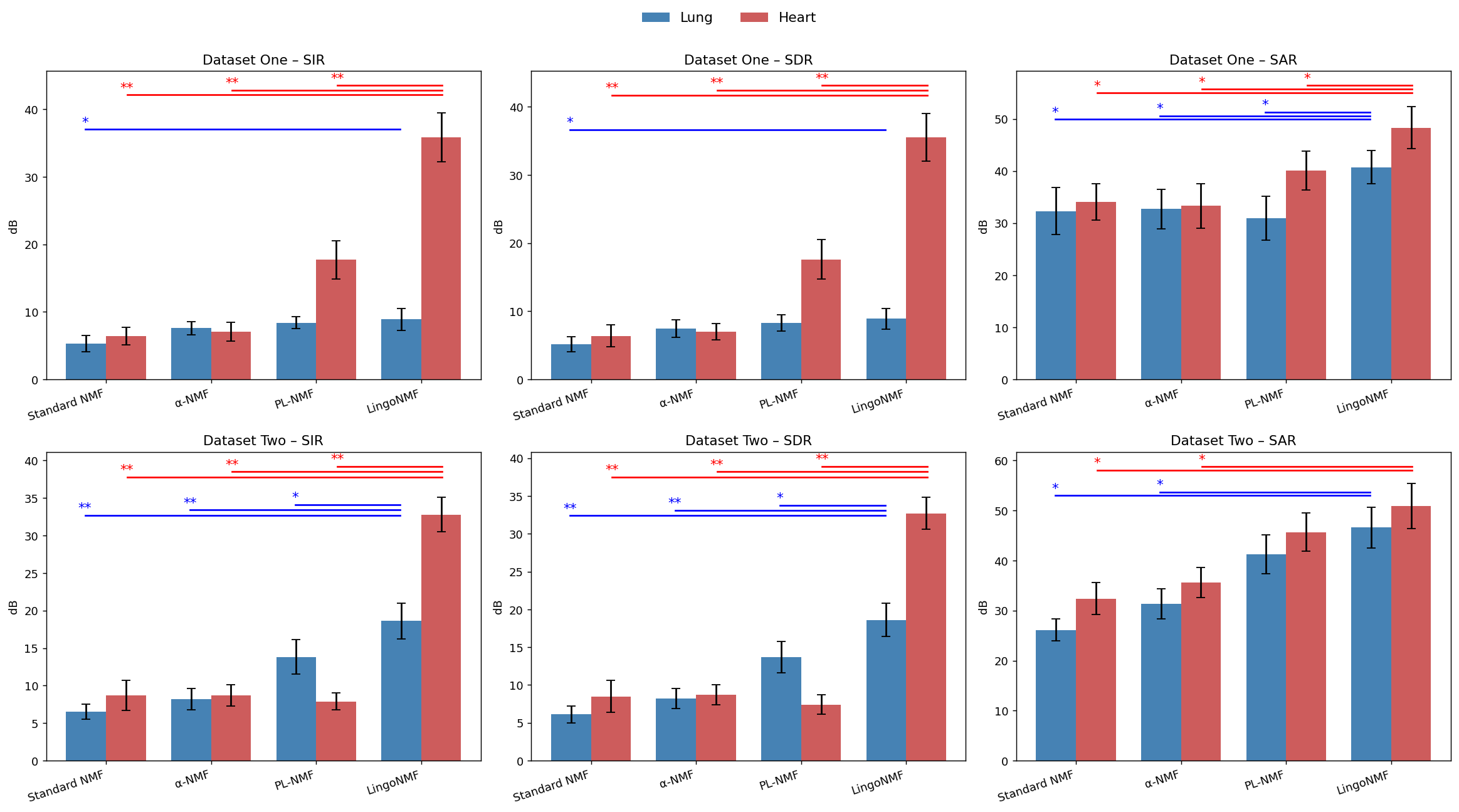}
    \captionsetup{justification=justified,font=footnotesize}
    \caption{Comparison of NMF-based separation methods across two datasets (Standard NMF, $\alpha$-NMF, PL-NMF, and LingoNMF). Bars show mean performance with 95\% confidence intervals for lung (blue) and heart (red) sounds in terms of SIR, SDR, and SAR. Horizontal lines above bar pairs denote results of paired t-tests comparing each baseline with LingoNMF. Single asterisks * represent $p<0.05$ and double asterisks ** represent $p<0.01$, indicating statistically significant improvements in separation performance.}
    \label{fig:stat}
\end{figure}
Table~\ref{tab:lingonmf_comparison} compares LingoNMF with other separation methods. It achieves higher SDR and SAR than all compared methods, while its SIR remains competitive. 

\subsection{Conclusion}
We propose a novel algorithm based on NMF for blind source separation of heart and lung sounds, combining parallel blocks and multi-layer structures. Additionally, the proposed method incorporates the periodicity of body signals. Moreover, the integration of LLMs enhances the separation results in two unique ways: by providing detailed diagnostics for analyzing the output characteristics and by dynamically optimizing a fundamental frequency penalty term in the NMF cost function through a feedback loop. Figure~\ref{fig:conc} illustrates the waveform and spectrogram of the mixed signal and the separated sounds using different NMF methods. The experimental results demonstrated that the proposed method outperformed previous methods, which highlights the potential of the natural language processing (NLP) tools to enhance medical sound analysis.

\begin{figure}[H]
    \centering
    \includegraphics[width=0.8\linewidth, trim= 0 100 0 100, clip]{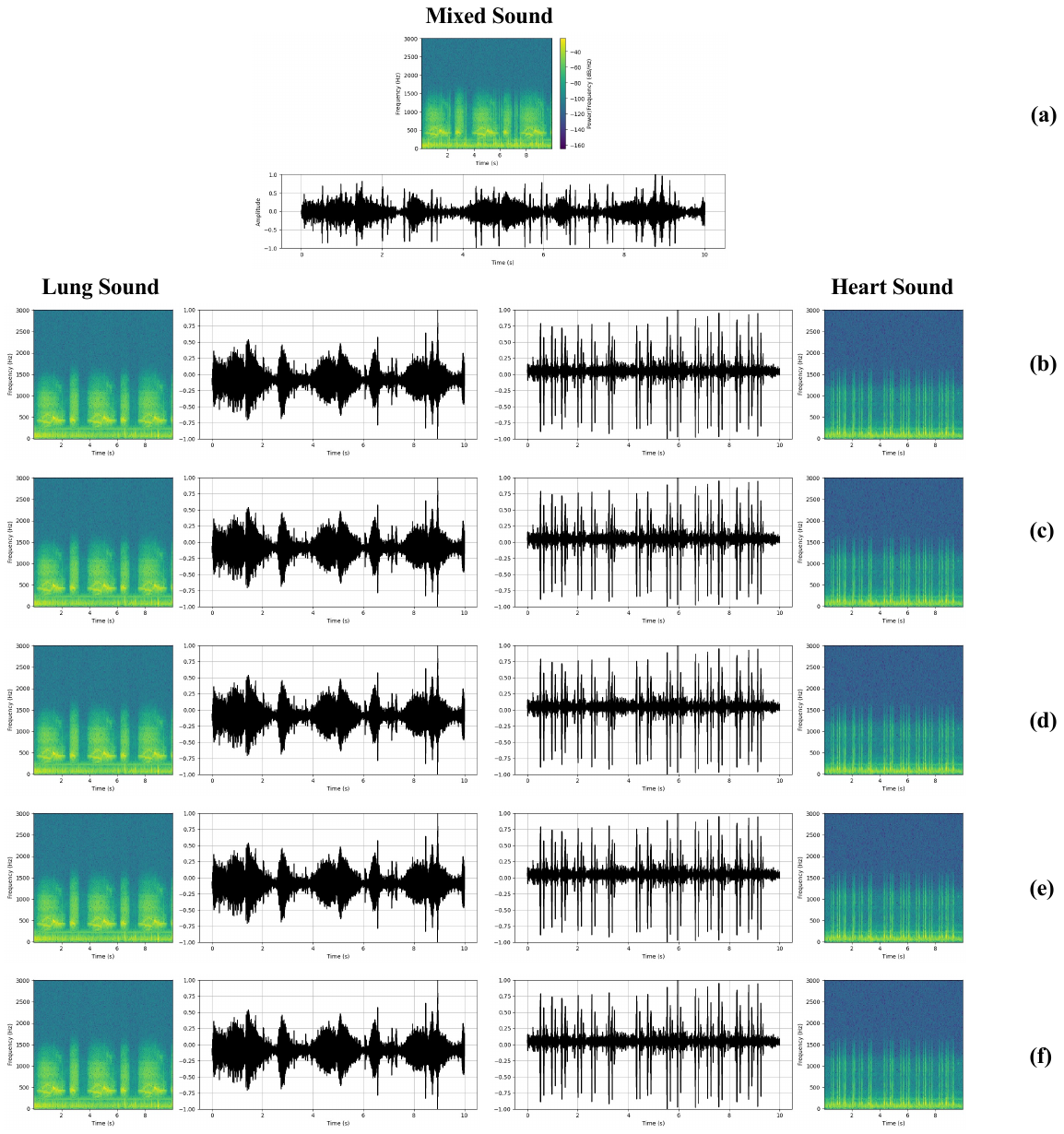}
    \captionsetup{justification=justified,font=footnotesize}
    \caption{Cardiorespiratory sound separation results: \textbf{(a)} Mixed sound with its corresponding spectrogram; \textbf{(b)} Expected source sounds; \textbf{(c)-(f)} Separated sounds using Standard NMF, $\alpha$-NMF, PL-NMF, and LingoNMF, each represented by a spectrogram and waveform. }
    \label{fig:conc}
\end{figure}
\section{Explainable Variational Autoencoder}
\subsection{Introduction}
Autoencoders are unsupervised neural models that learn to reconstruct their input through a latent bottleneck representation, and consist of an encoder and a decoder. Figure~\ref{fig:vae_arch} illustrates the overall architecture of a variational autoencoder (VAE). In this section, we introduce XVAE‑WMT, a generative AI algorithm that combines a variational autoencoder (VAE) with explainable AI (XAI), incorporating temporal consistency (TC) loss, wavelet-based inputs, and a final masked output stage. For explainability, we analyze the latent space and further apply SHAP (SHapley Additive exPlanations) to quantify the contribution of individual latent features.

\begin{figure}[H]
    \centering
    \includegraphics[width=0.8\linewidth]{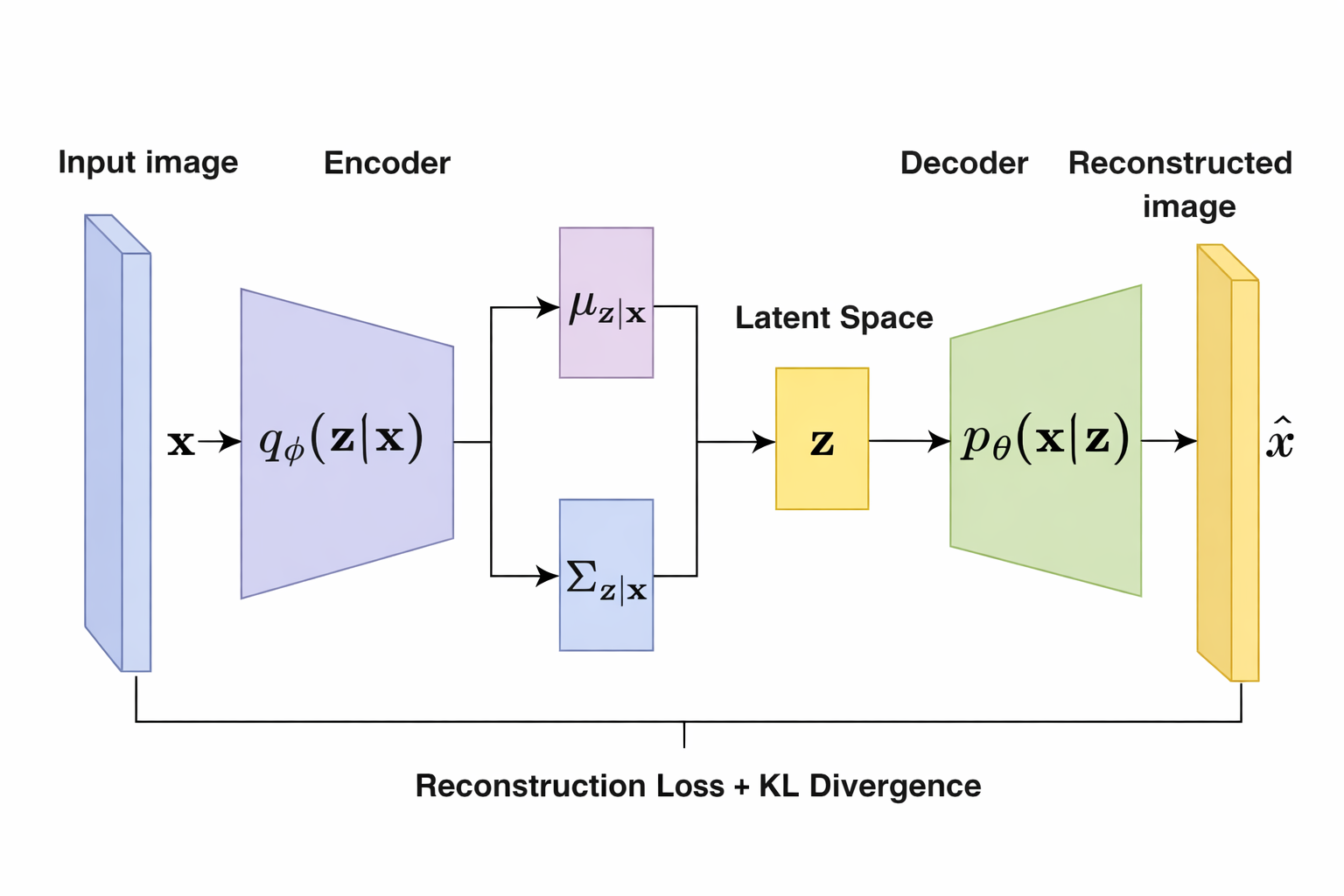}
    \captionsetup{justification=justified,font=footnotesize}
    \caption{Schematic overview of the VAE architecture. The encoder $q_\phi(z \mid x)$ maps the input mixture $x$ into a probabilistic latent representation parameterized by the mean $\mu_{z\mid x}$ and covariance $\Sigma_{z\mid x}$. 
    A latent variable $z$ is sampled from this distribution and passed to the decoder $p_\theta(x \mid z)$ to reconstruct the input signal $\hat{x}$. 
    Training is performed by minimizing the reconstruction loss together with the Kullback--Leibler (KL). Both the encoder $q(\cdot)$ and decoder $p(\cdot)$ are implemented as neural
networks trained jointly via variational inference.
}
    \label{fig:vae_arch}
\end{figure}

The VAE-based separation operates in a single-microphone setting, whereas
most conventional source separation methods rely on multiple microphones and spatial
diversity. This is enabled by learning a structured latent representation that captures
source-specific spectral and temporal patterns directly from single-channel mixtures.

\subsection{Methodology}
\subsubsection{Theoretical Background}
In a standard autoencoder, the encoder maps an input vector \( \mathbf{x} \in \mathbb{R}^n \) to a latent vector \( \mathbf{h} \in \mathbb{R}^k \) (Equation~\ref{eq:encoder}), and the decoder reconstructs the input from this latent representation (Equation~\ref{eq:decoder}). The goal of training is to minimize the reconstruction loss, given in Equation~\ref{eq:ae_loss}. Here, $\Delta$ denotes a reconstruction loss function measuring the discrepancy between the input and its reconstruction. In this work, $\Delta$ is chosen as the mean squared error (MSE).

\begin{equation}
\mathbf{h} = g(\mathbf{x}).
\label{eq:encoder}
\end{equation}
\begin{equation}
\tilde{\mathbf{x}} = f(\mathbf{h}) = f(g(\mathbf{x})).
\label{eq:decoder}
\end{equation}
\begin{equation}
\arg\min_{f,g} \big[ \Delta(\tilde{\mathbf{x}}, f(g(\mathbf{x}))) \big].
\label{eq:ae_loss}
\end{equation}
Here, the latent dimension satisfies $k \leq n$, enforcing a compressed representation of the input. The encoder $g$ and decoder $f$ are implemented as parameterized neural networks, where the encoder maps the input signal to a low-dimensional latent space and the decoder reconstructs the input from this latent representation. Although the autoencoder learns a compact latent representation, it cannot generate new data because it learns a deterministic mapping. Variational autoencoders (VAEs) address this limitation through a probabilistic framework that enables both reconstruction and generation. The encoder maps the input \( \mathbf{x} \) not to a fixed point but to a distribution in the latent space. The decoder reconstructs data by sampling from this distribution. We formulate the VAE optimization problem as a maximum likelihood estimation (MLE) task, where the goal is to maximize the log-likelihood of the observed data under the model parameters:
\begin{equation}
\max_{\theta}\,\log p_{\theta}(\mathbf{x})
\label{eq:vae_mle_objective}
\end{equation}
We start from the marginal data likelihood:
\begin{equation}
\log p_{\theta}(\mathbf{x})=\log \int p_{\theta}(\mathbf{x},\mathbf{z})\,d\mathbf{z}.
\label{eq:vae_loglike}
\end{equation}
Here, $z \in \mathbb{R}^k$ denotes the latent variable associated with the bottleneck of the autoencoder, representing the low-dimensional latent representation of the input signal. Since the integral in Equation~\ref{eq:vae_loglike} is generally intractable, we introduce an approximate posterior \(q_{\phi}(\mathbf{z}\,|\,\mathbf{x})\) and rewrite the expression by multiplying and dividing inside the integral:
\begin{align}
\log p_{\theta}(\mathbf{x})
&=\log \int \frac{p_{\theta}(\mathbf{x},\mathbf{z})}{q_{\phi}(\mathbf{z}\,|\,\mathbf{x})}\,q_{\phi}(\mathbf{z}\,|\,\mathbf{x})\,d\mathbf{z}
=\log \mathbb{E}_{\mathbf{z}\sim q_{\phi}}\!\left[\frac{p_{\theta}(\mathbf{x},\mathbf{z})}{q_{\phi}(\mathbf{z}\,|\,\mathbf{x})}\right].
\label{eq:vae_expectation}
\end{align}
Applying Jensen’s inequality yields a lower bound on the log-likelihood:
\begin{equation}
\log p_{\theta}(\mathbf{x})
\geq \mathbb{E}_{\mathbf{z}\sim q_{\phi}}\!\left[\log \frac{p_{\theta}(\mathbf{x},\mathbf{z})}{q_{\phi}(\mathbf{z}\,|\,\mathbf{x})}\right].
\label{eq:elbo_inequality}
\end{equation}
Expanding the joint distribution \(p_{\theta}(\mathbf{x},\mathbf{z})=p_{\theta}(\mathbf{x}\,|\,\mathbf{z})\,p_{\theta}(\mathbf{z})\) and separating the terms inside the logarithm, we obtain
\begin{align}
\mathbb{E}_{\mathbf{z}\sim q_{\phi}}\!\left[\log \frac{p_{\theta}(\mathbf{x},\mathbf{z})}{q_{\phi}(\mathbf{z}\,|\,\mathbf{x})}\right]
&=\mathbb{E}_{\mathbf{z}\sim q_{\phi}}\!\left[\log \frac{p_{\theta}(\mathbf{x}\,|\,\mathbf{z})\,p_{\theta}(\mathbf{z})}{q_{\phi}(\mathbf{z}\,|\,\mathbf{x})}\right] \nonumber\\
&=\mathbb{E}_{\mathbf{z}\sim q_{\phi}}\![\log p_{\theta}(\mathbf{x}\,|\,\mathbf{z})]
+\mathbb{E}_{\mathbf{z}\sim q_{\phi}}\!\left[\log\!\frac{p_{\theta}(\mathbf{z})}{q_{\phi}(\mathbf{z}\,|\,\mathbf{x})}\right].
\label{eq:elbo_expand_raw}
\end{align}
For two continuous probability distributions \(a(\mathbf{x})\) and \(b(\mathbf{x})\), the Kullback–Leibler (KL) divergence is defined as
\begin{equation}
\mathcal{D}_{\mathrm{KL}}(a\parallel b)
=\mathbb{E}_{\mathbf{x}\sim a}\!\left[\log\!\frac{a(\mathbf{x})}{b(\mathbf{x})}\right]
=\int a(\mathbf{x})\log\!\frac{a(\mathbf{x})}{b(\mathbf{x})}\,d\mathbf{x}.
\label{eq:kl_def}
\end{equation}
The second expectation term in Equation~\ref{eq:elbo_expand_raw} corresponds to the negative KL divergence between the approximate posterior \(q_{\phi}(\mathbf{z}\,|\,\mathbf{x})\) and the prior \(p_{\theta}(\mathbf{z})\):
\begin{equation}
\mathbb{E}_{\mathbf{z}\sim q_{\phi}}\!\left[\log\!\frac{p_{\theta}(\mathbf{z})}{q_{\phi}(\mathbf{z}\,|\,\mathbf{x})}\right]
=-\,\mathcal{D}_{\mathrm{KL}}\!\big(q_{\phi}(\mathbf{z}\,|\,\mathbf{x})\,\|\,p_{\theta}(\mathbf{z})\big).
\label{eq:elbo_klterm}
\end{equation}
Substituting Equation~\ref{eq:elbo_klterm} into Equation~\ref{eq:elbo_expand_raw} yields the \emph{Evidence Lower Bound} (ELBO), which serves as the objective function optimized during VAE training:
\begin{equation}
\mathrm{ELBO}
=
-\mathcal{D}_{\mathrm{KL}}\!\big(q_{\phi}(\mathbf{z}\,|\,\mathbf{x})\,\|\,p_{\theta}(\mathbf{z})\big)
+\mathbb{E}_{\mathbf{z}\sim q_{\phi}}[\log p_{\theta}(\mathbf{x}\,|\,\mathbf{z})].
\label{eq:elbo}
\end{equation}
The ELBO is maximized with respect to the model parameters $\theta$ and $\phi$, where $\theta$ denotes the parameters of the decoder network $p_\theta(x|z)$ and $\phi$ denotes the parameters of the encoder network $q_\phi(z|x)$. Maximizing this bound regularizes the latent space to match the prior in the encoder through the first KL regularization term, and it encourages accurate reconstruction of the input in the decoder through the second likelihood term. Given the assumption that the prior distribution follows a standard multivariate 
normal $p(\mathbf{z}) \equiv \mathcal{N}(\mathbf{0}, \mathbf{I})$, and that the 
approximate posterior 
$q_{\phi}(\mathbf{z}\,|\,\mathbf{x}) = \mathcal{N}(\boldsymbol{\mu}, 
\mathrm{diag}(\boldsymbol{\sigma}^2))$ 
is Gaussian with mean vector $\boldsymbol{\mu} \in \mathbb{R}^k$ and standard 
deviation vector $\boldsymbol{\sigma} \in \mathbb{R}^k$, 
the ELBO in Equation~\ref{eq:elbo} simplifies to a closed-form expression. 
The KL divergence term becomes
\begin{equation}
\mathcal{D}_{\mathrm{KL}}\!\big(q_{\phi}(\mathbf{z}\,|\,\mathbf{x})\,\|\,p_{\theta}(\mathbf{z})\big)
=\tfrac{1}{2}\sum_{i=1}^{k}\!\left(\sigma_{i}^{2}+\mu_{i}^{2}-1-\log\sigma_{i}^{2}\right),
\label{eq:kl_gauss_simplified}
\end{equation}
which regularizes the latent representation to remain close to the unit Gaussian prior. In addition, the reconstruction term $\mathbb{E}_{\mathbf{z}\sim q_{\phi}}[\log p_{\theta}(\mathbf{x}\,|\,\mathbf{z})]$ can be 
simplified under the same Gaussian assumption and  implemented as the mean squared error 
(MSE) loss between the input mixture and its reconstruction.

\subsubsection{Proposed Method}

The proposed XVAE-WMT method integrates wavelet transform preprocessing, a deep encoder–decoder architecture, and a temporal consistency (TC) regularization term to enforce smooth transitions between consecutive time steps. We further apply explainable AI (XAI) analysis to visualize and interpret the learned latent space and identify the most influential latent dimensions.

\textbf{Temporal Consistency Regularization:} We enforce smooth temporal evolution in the reconstructed spectrograms through a regularization term. This loss minimizes abrupt variations between adjacent frames in the separated outputs, defined as
\begin{equation}
\mathcal{L}_{TC} = \| \hat{\mathbf{s}}_{t+1} - \hat{\mathbf{s}}_{t} \|_2^2,
\end{equation}
where $\hat{\mathbf{s}}_t$ and $\hat{\mathbf{s}}_{t+1}$ denote consecutive reconstructed frames. We add this term to the ELBO loss function in Equation~\ref{eq:elbo}. The temporal consistency loss is incorporated into the ELBO objective with a weighting parameter $\lambda_{TC}$ to control the strength of the regularization.

\textbf{Network Architecture:} The VAE model takes a spectrogram $\mathbf{X}(t,f)$ of size $1000 \times 128$ as input and separates it into source components. It employs two parallel encoders and two decoders, one for each source. Each encoder includes three convolutional layers. The first layer uses 128 filters with a $1 \times 128$ kernel to capture spectral patterns. The second layer applies 128 filters with a $4 \times 1$ kernel to learn temporal dependencies, and the third layer employs 256 filters with a $4 \times 1$ kernel to extract higher-level correlations. The output of each encoder has a shape of $994 \times 1 \times 256$. The model flattens this feature map into a vector and passes it through two fully connected layers to estimate the latent distribution parameters: the mean $\boldsymbol{\mu} \in \mathbb{R}^{128}$ and the log-variance $\log\boldsymbol{\sigma}^2 \in \mathbb{R}^{128}$. Using the reparameterization trick, the latent vector $\mathbf{z}$ is sampled as
\begin{equation}
\mathbf{z} = \boldsymbol{\mu} + \boldsymbol{\sigma} \odot \boldsymbol{\varepsilon}, \qquad \boldsymbol{\varepsilon} \sim \mathcal{N}(\mathbf{0}, \mathbf{I}).
\end{equation}
This sampling operation uses the \emph{reparameterization trick} \cite{kingma2013auto}, which expresses the random variable $z$ as a deterministic function of the encoder outputs and an auxiliary noise variable $\epsilon \sim \mathcal{N}(0,I)$, enabling backpropagation through the stochastic sampling step during training. Each encoder therefore produces one latent vector $\mathbf{z} \in \mathbb{R}^{128}$ per input. The resulting latent representations, $\mathbf{z}_1$ and $\mathbf{z}_2$, correspond to the heart and lung sources, respectively. Each decoder receives its associated $\mathbf{z}$ and reconstructs its target spectrogram.

\textbf{Output Masking:} 
Each latent vector $\mathbf{z}_i$ obtained from the encoder is passed through 
its corresponding decoder network to generate a soft time--frequency mask 
$\mathbf{M}_i(t,f)$. The decoder consists of fully connected and deconvolution layers, 
and its final activation is a Sigmoid function that constrains mask values 
to the range $[0,1]$. The mask is then applied element-wise to the input 
mixture spectrogram $\mathbf{X}(t,f)$ as
\begin{equation}
\hat{\mathbf{S}}_i(t,f) = \mathbf{M}_i(t,f) \odot \mathbf{X}(t,f).
\end{equation}

Spectral masking is applied so that the VAE estimates \emph{how much} of the input mixture belongs to each source at every time--frequency point, rather than generating the source signals directly. This masking operation can be interpreted as a soft assignment that partitions the mixture energy across sources at each time--frequency bin. By multiplying the learned mask with the mixture spectrogram, the separated outputs are guaranteed to sum to the original mixture and remain consistent with the observed data. This design improves separation stability and avoids unrealistic signal reconstruction.

\textbf{Explainable Latent-Space Analysis:} 
We applied XAI analysis to interpret the latent representations learned by the VAE. The XAI method analyzes the learned latent space by identifying which latent dimensions most strongly influence the reconstruction and separation outputs.
Rather than assigning explicit semantic meaning, the analysis highlights dimensions associated with dominant time--frequency activity patterns in the data.
First, we visualized the latent vectors of both sources using t-SNE. Next, we performed a SHAP (SHapley Additive exPlanations) analysis to rank the importance of each latent dimension. Using these SHAP scores, we iteratively reduced the latent space.

\paragraph{Training Process and Motivation}
The VAE is trained in an unsupervised manner using sound mixtures without access to ground-truth source signals. The objective is to learn a structured latent representation that captures the underlying characteristics of each source. Wavelet-based inputs are used to preserve time--frequency structure, while output masking enforces mixture consistency. The temporal consistency regularization encourages smooth evolution of the reconstructed signals over time, reflecting the physiological continuity of cardiorespiratory sounds.

\subsection{Experimental Results}
\textbf{Dataset:} We trained the model on two datasets. We created the first dataset by mixing three databases: the Kaggle Respiratory Sound Database, the CirCor DigiScope Heart Sound Database, and the Chest Wall Recording Dataset. All recordings were resampled to 16 kHz, normalized, and segmented into 1-second frames, producing 24,383 heart and 18,144 lung segments across training iterations. Random pairs were combined to generate the mixtures. We collected the second dataset (HLS-CMDS) from a clinical manikin. After resampling, normalization, and segmentation, we randomly paired heart and lung segments, forming 25,000 mixtures. Each sound frame was transformed into a wavelet spectrogram using the Morlet mother wavelet.\\ 
\textbf{Separation Results:}
Table~\ref{tab:results_all} presents the separation performance of different VAE variants across the two datasets, demonstrating the performance gains achieved by incorporating output masking, temporal consistency loss, and wavelet-based representations. Table~\ref{tab:vae_vs_nonNMF} further compares the proposed VAE-WMT method with other baseline blind source separation approaches, showing that VAE-WMT achieves higher SIR, SAR, and SDR, highlighting the effectiveness of the proposed architecture.

\noindent
\textbf{Clustering Results:} We analyzed the clustering quality using four unsupervised metrics—Silhouette coefficient, Davies--Bouldin (DB) index, Calinski--Harabasz (CH) index, and variance (Table~\ref{tab:latent_metrics}). We applied a SHAP-based explainability analysis on the latent embeddings to estimate feature importance and rank dimensions according to their contribution to clustering separability. We evaluated clustering performance for different top-k\% of the latent dimension (Figure~\ref{fig:latent}). We retained the top 75\% of latent dimensions for final visualization and quantitative analysis, which yielded the most interpretable structure without degrading performance. 

\begin{figure}[H]
    \centering
    \includegraphics[width=\linewidth]{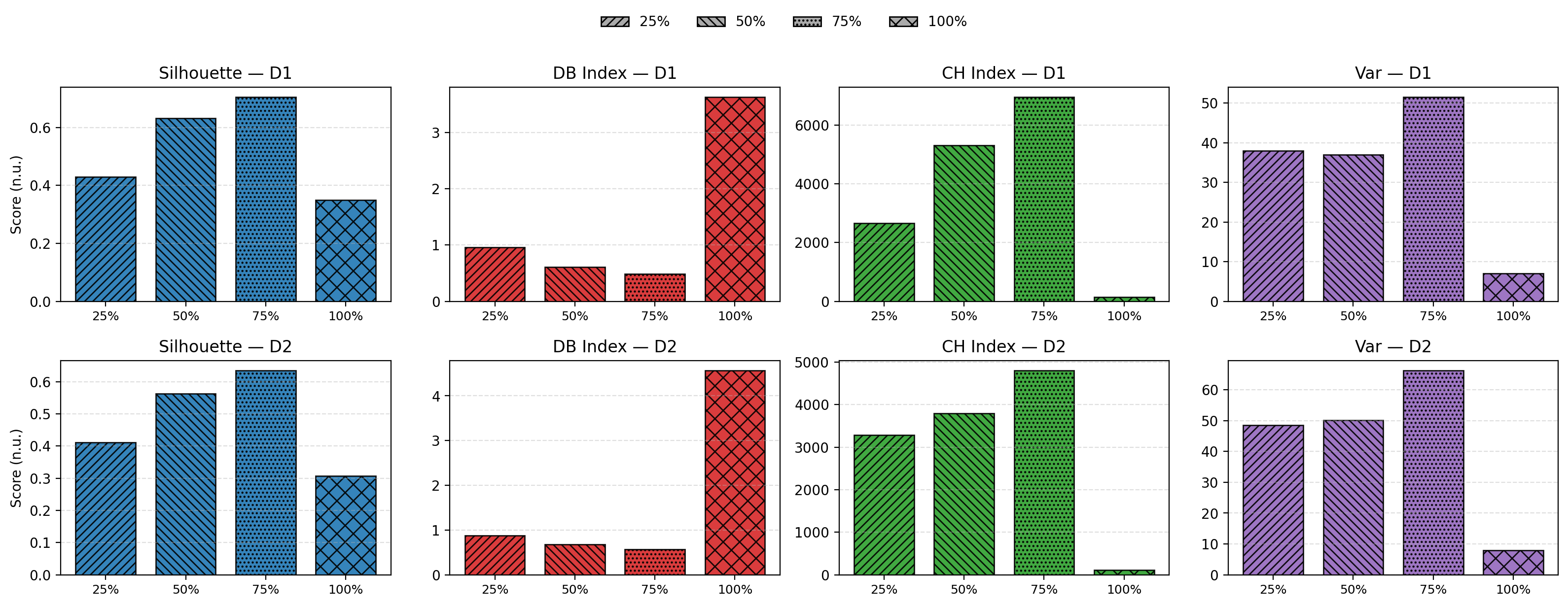}
    \captionsetup{justification=justified,font=footnotesize}
\caption{VAE clustering performance across different top-$k$\% subsets of the latent dimensions. Each bar represents the mean score for Silhouette, Davies–Bouldin (DB), Calinski–Harabasz (CH), and variance metrics 
on datasets D1 and D2. Higher Silhouette and CH, and lower DB and variance values indicate better clustering quality.}

    \label{fig:latent}
\end{figure}

Figure~\ref{fig:vae} visualizes the latent spaces of different VAE variants using t-SNE projections. Models incorporating wavelet, mixture, and temporal regularization components yield more compact and separable manifolds, 
    while simpler configurations exhibit diffuse or overlapping structures.
\begin{figure}[H]
    \centering
    \includegraphics[width=\linewidth, trim= 0 150 0 150, clip]{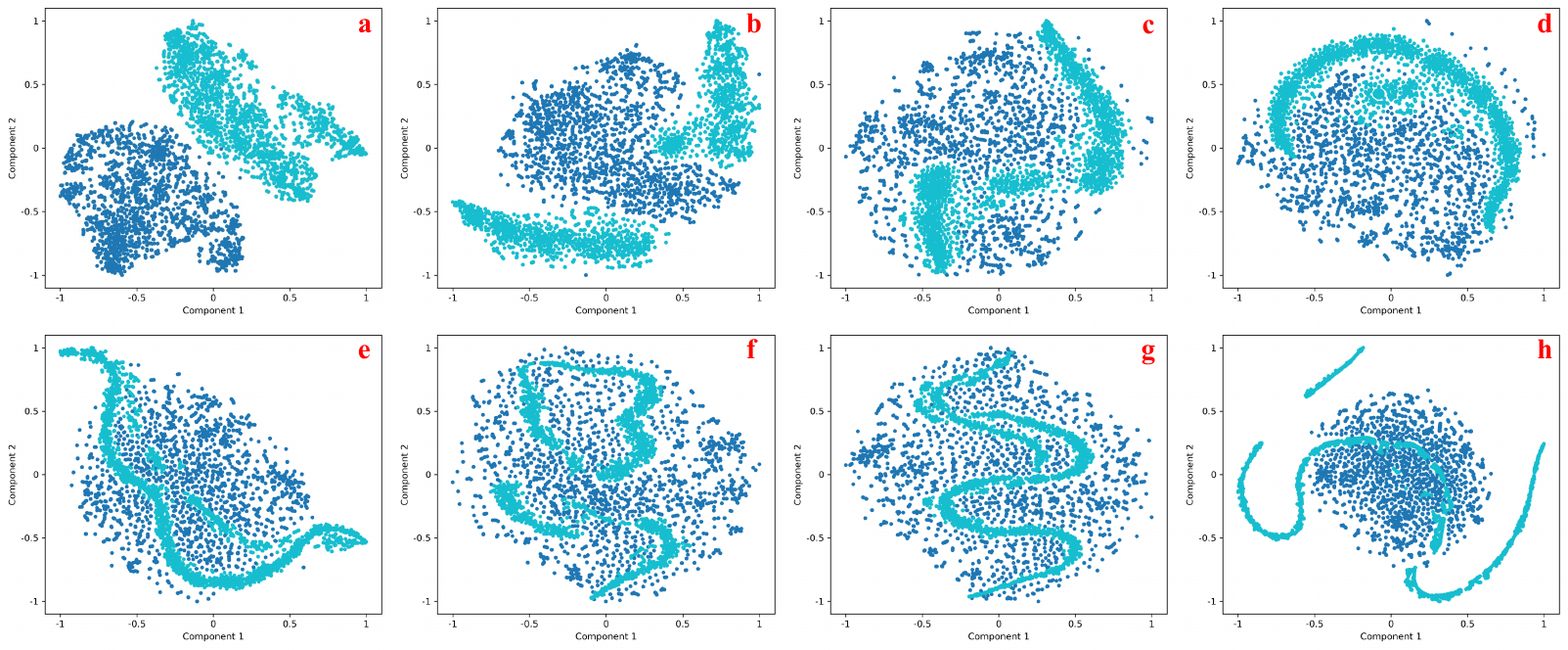}
    \captionsetup{justification=justified,font=footnotesize}
\caption{Latent space visualization of different VAE variants using t-SNE: 
    \textbf{(a)} XVAE-WMT, \textbf{(b)} XVAE-WT, \textbf{(c)} XVAE-WM, \textbf{(d)} XVAE-W, 
    \textbf{(e)} XVAE-MT, \textbf{(f)} XVAE-T, \textbf{(g)} XVAE-M, and \textbf{(h)} XVAE. 
    Each subplot shows the latent embeddings color-coded by cluster assignment.
    Two color groups (two shades of blue) correspond to the two sources being separated, 
    and the scatter patterns illustrate how each VAE variant organizes these sources in 
    latent space.}

    \label{fig:vae}
\end{figure}

\subsection{Conclusion}
We proposed an explainable variational autoencoder (VAE) for blind source separation in scenarios where access to clean source signals is restricted. By combining deep neural networks with probabilistic latent modelling, the proposed approach separates mixed cardiorespiratory recordings into low-dimensional latent representations corresponding to heart and lung sounds. The incorporation of wavelet-based inputs, output masking, and temporal consistency regularization promotes physiologically meaningful and stable separations, while XAI-based analysis enables interpretation and dimensionality reduction of the learned latent space. Experimental results on real heart and lung sound recordings demonstrate that the proposed method outperforms existing baseline approaches. In practical operation, an unseen mixture is encoded into latent variables and decoded to produce time--frequency masks that reconstruct the separated heart and lung signals without retraining.

        \setcounter{figure}{0}
        \setcounter{equation}{0}
        \setcounter{table}{0}
  \chapter{Physical Chemistry-Inspired Matrix Factorization for Clustering}
\newpage
Non-negative matrix factorization (NMF) is an unsupervised learning method that offers low-rank representations across various domains, including audio processing, biomedical signal analysis, and image recognition. The incorporation of $\alpha$-divergence in NMF formulations enhances flexibility in optimization, yet extending these methods to multi-layer architectures presents challenges in ensuring convergence. To address this, we introduce a novel approach inspired by the Boltzmann probability of the energy barriers in chemical reactions to perform convergence analysis theoretically. We introduce a novel method, called Chem-NMF, with a bounding factor which stabilizes convergence. To our knowledge, this is the first study to apply a physical chemistry perspective to analyze the convergence behaviour of the NMF algorithm.

\section{Introduction}
Data clustering plays a critical role in computer vision and pattern recognition, as it enables unsupervised organization of large-scale datasets into meaningful groups. In this context, NMF is an interpretable representation learning tool for data clustering, which can generate low-dimensional features. NMF is a widely used technique for decomposing high-dimensional data into interpretable low-rank components, and it has found applications in audio processing, biomedical signal analysis, image recognition, text mining, and blind source separation. Among divergence-based NMF approaches, the $\alpha$-divergence formulation provides a flexible framework that generalizes traditional cost functions and enhances model adaptability in different applications~\cite{Cichocki2008}. However, extending these formulations to multi-layer architectures~\cite{Cichocki2010} introduces additional mathematical complexities, requiring a deeper understanding of their theoretical properties, such as convergence~\cite{Torabi2025ChemNMF}. Several studies have investigated the convergence properties of NMF algorithms, often focusing on different divergence measures and optimization techniques, but they primarily address single-layer architectures, leaving the convergence behaviour of multi-layer NMF largely unexplored. To address these challenges, our work draws inspiration from physical chemistry, particularly the concepts of energy barriers and Boltzmann probability \cite{PIETRUCCI201732}, where systems must overcome thresholds to transition between states.

\section{Physical Chemistry Background}

In order to motivate the analogy between chemical reactions and the convergence of 
multi-layer $\alpha$-NMF, we review several basic chemical concepts \cite{Atkins2022}. In chemical reactions, the initial molecules that change are called \textit{reactants}, while the final stable molecules formed after completion are referred to as \textit{products}. The driving force behind these transformations is the \textit{Gibbs free energy}, which combines a system’s enthalpy $H$ and entropy $S$ at temperature $T$:
\begin{equation}
\Delta G = \Delta H - T \Delta S .
\end{equation}

Many reactions proceed in \textit{multi-stage reactions}, each with its own transition state and energy barrier (Figure~\figsubref{fig:reaction}{a}). A free energy diagram shows reactants moving through several intermediates before reaching a stable product state. Each stage resembles an energy basin separated by barriers. The \textit{transition state} itself is a high-energy, unstable configuration that represents the maximum energy barrier between reactants and products. Between two such barriers, a temporary species known as an \textit{intermediate} can form. The energy needed to cross the transition state is called the \textit{activation energy}. The Gibbs free energy difference between the reactants and the transition state defines the activation barrier, which controls the reaction rate:
\begin{equation}
\Delta G^{\ddagger} = G_{\text{TS}} - G_{\text{reactants}}.
\end{equation}
Catalysts lower $\Delta G^{\ddagger}$ by stabilizing the transition state (Figure~\figsubref{fig:reaction}{b}). In catalyzed reactions, the pathway is rerouted to reduce the activation barrier.
The likelihood of crossing these barriers is governed by the \textit{Boltzmann distribution}, which describes the probability of a system occupying a state with energy $E$ and thereby determines how easily the system can overcome energy barriers to reach more stable states. 
\begin{figure}[H]
    \centering
    \includegraphics[width=\textwidth]{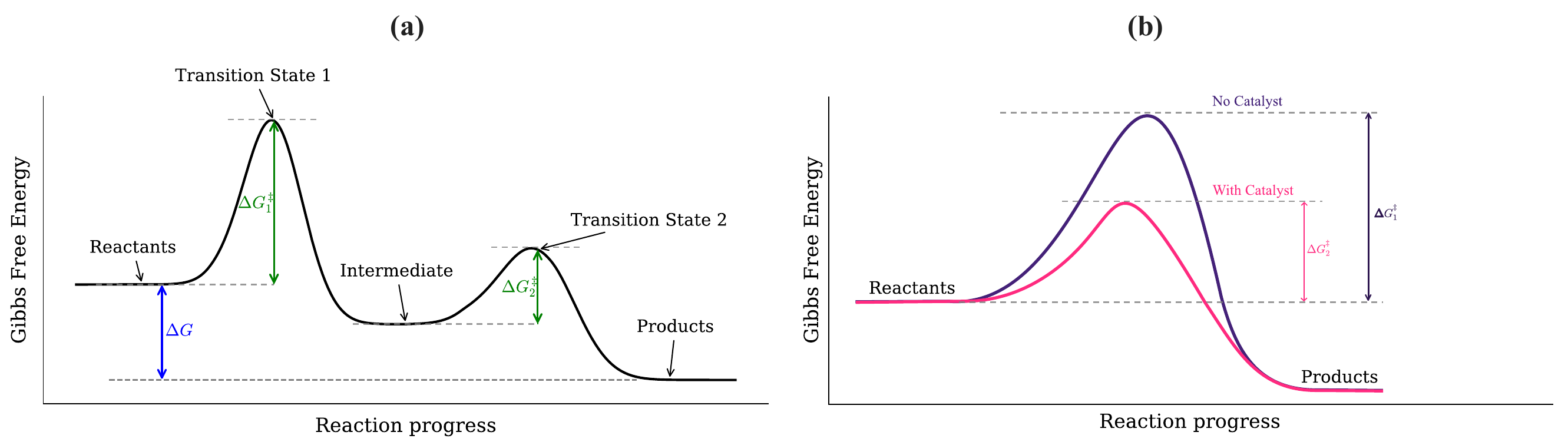}
    \captionsetup{justification=justified,font=footnotesize}
    \caption{Energy profile of a reaction progress: \textbf{(a)} Reactants, intermediate, and products. 
    The two transition states correspond to the energy maxima, with activation 
    free energies $\Delta G^{\ddagger}_{1}$ and $\Delta G^{\ddagger}_{2}$ indicated by vertical 
    arrows. The overall free energy change $\Delta G$ is shown between reactants and products; \textbf{(b)} The effect of catalyst on lowering the activation barrier.}
    \label{fig:reaction}
\end{figure}
The above chemical phenomena provide a natural analogy to the optimization process of the multi-layer $\alpha$-NMF algorithm. In this analogy, just as chemical systems move through intermediates before reaching the most stable state, multi-layer $\alpha$-NMF traverses successive layers to escape shallow minima and converge to better solutions.

\section{Proposed Method}
As shown in Figure~\ref{fig:method}, we factorize the input data, and we perform clustering on the activation maps. Next, we reconstruct the maps by multiplying the feature basis matrices, and we evaluate the clustering results. Chem-NMF is a multi-layer $\alpha$-divergence NMF algorithm that introduces a bounding factor ($BF$) as a novel mechanism to improve convergence (Algorithm~\ref{alg:chemnmf} ). Inspired by the way catalysts reduce activation energy in chemical reactions, the bounding factor is applied during the random initialization step to stabilize the search space in order to reduce the risk of overshooting and getting trapped in local minima.

\begin{figure}[H]
    \centering
    \includegraphics[width=0.7\textwidth]{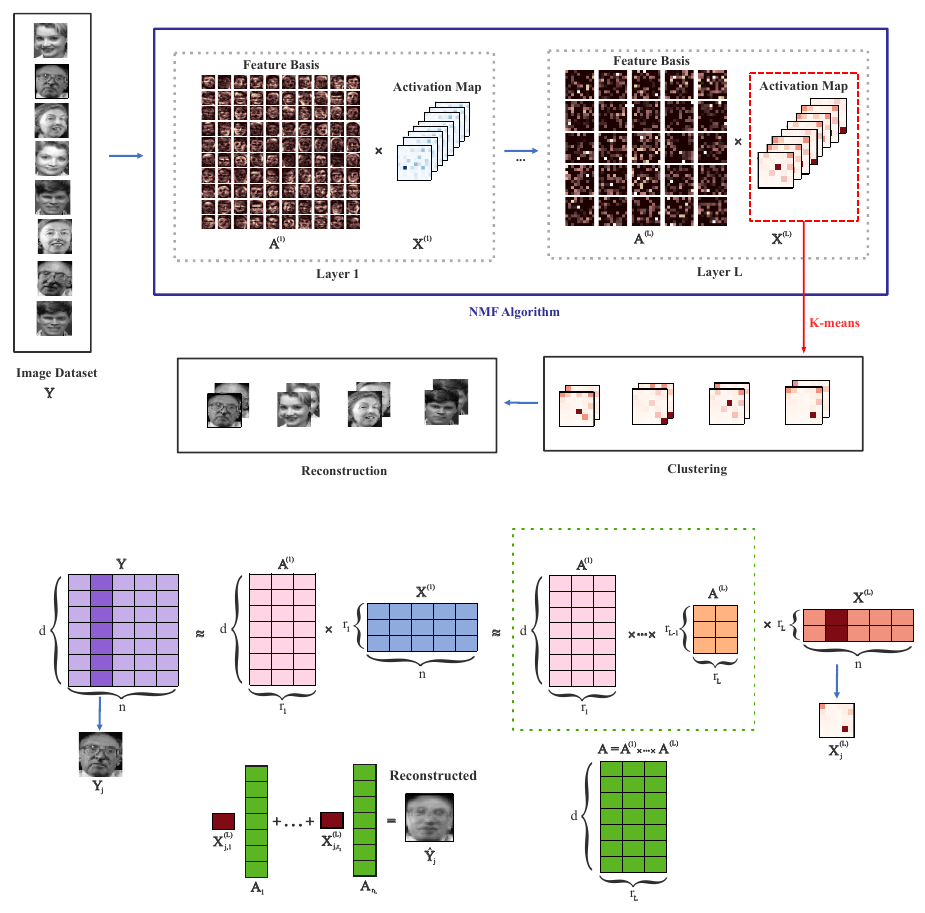}
    \captionsetup{justification=justified, font=footnotesize}
    \caption{Overview of the clustering procedure. 
    The input dataset $\mathbf{Y}$ is factorized into a feature basis $\mathbf{A}$ and activation maps $\mathbf{X}$ across multiple layers using NMF. 
    The activation maps are clustered with $k$-means, and images are reconstructed using feature basis matrices.}
    \label{fig:method}
\end{figure}

\section{Rigorous Convergence Analysis}
Let $D_{\alpha}(\mathbf{Y} \parallel \mathbf{A}\mathbf{X})$ denote the objective function based on the $\alpha$-divergence between $\mathbf{Y}$ and $\mathbf{A}\mathbf{X}$ defined as Equation~\ref{eq:div}. We show that the algorithm converges subject to its multiplicative update rule \cite{Cichocki2009}. The convergence analysis assumes that the optimization variables follow a
Boltzmann-type distribution, which enables controlled energy minimization and
stable convergence behaviour. In practice, the method is not highly sensitive to
moderate deviations from this assumption, as convergence is primarily driven by
the bounded update dynamics. As future work, this framework could be extended beyond NMF to more general non-convex optimization problems.

\begin{equation}
D_{\alpha}(\mathbf{Y} \parallel \mathbf{A}\mathbf{X}) = \frac{1}{\alpha(\alpha - 1)} 
\sum_{it} \left( y_{it}^\alpha [\mathbf{A}\mathbf{X}]_{it}^{1 - \alpha} - \alpha y_{it} + (\alpha - 1) [\mathbf{A}\mathbf{X}]_{it} \right).
\label{eq:div}
\end{equation}

\begin{theorem} 
The NMF algorithm follows the multiplicative update rules: 

\begin{equation}
    x_{jt} \leftarrow x_{jt} 
    \left( \frac{\sum\limits_{i} a_{ij} \left( \frac{y_{it}}{[\mathbf{A}\mathbf{X}]_{it}} \right)^{\alpha}}
    {\sum\limits_{i} a_{ij}} \right)^{\frac{1}{\alpha}}.
\end{equation}

\begin{equation}
    a_{ij} \leftarrow a_{ij} 
    \left( \frac{\sum\limits_{t} x_{jt} \left( \frac{y_{it}}{[\mathbf{A}\mathbf{X}]_{it}} \right)^{\alpha}}
    {\sum\limits_{t} x_{jt}} \right)^{\frac{1}{\alpha}}.
\end{equation}
\end{theorem}
\textit{Proof} Appendix~\ref{sec:thm441}.

\begin{definition}[Auxiliary Function]
A function $G(\mathbf{X}, \mathbf{X}')$ is an auxiliary function for $F(\mathbf{X})$ if it satisfies the following conditions:
\begin{enumerate}[label=\roman*]
    \item $G(\mathbf{X}, \mathbf{X}) = F(\mathbf{X})$,
    \item $G(\mathbf{X}, \mathbf{X}') \geq F(\mathbf{X})$, for all $\mathbf{X}'$.
\end{enumerate}
\end{definition}

\begin{lemma}
The function
\begin{equation}
G(\mathbf{X}, \mathbf{X}') = \frac{1}{\alpha(\alpha - 1)} 
\sum_{ijt} y_{it} \zeta_{itj} 
\left[ \left( \frac{a_{ij}x_{jt}} {y_{it} \zeta_{itj}}  \right)^{(1-\alpha)} 
+ (\alpha - 1)  \frac{a_{ij} x_{jt}} {y_{it}\zeta_{itj}}  - \alpha \right],
\end{equation}
where 
\begin{equation}
\zeta_{itj} = \frac{a_{ij} x'_{jt}}{\sum_{j=1}^{J} a_{ij} x'_{jt}},
\label{eq:zeta}
\end{equation}
is an auxiliary function for
\begin{equation}
F(\mathbf{X}) = \frac{1}{\alpha(\alpha - 1)} 
\sum_{it} \left( y_{it}^{\alpha} [\mathbf{A}\mathbf{X}]_{it}^{1-\alpha} - \alpha y_{it} + (\alpha - 1) [\mathbf{A}\mathbf{X}]_{it} \right).
\end{equation}
\end{lemma}
\textit{Proof} Appendix~\ref{sec:lm441}.

\begin{theorem}
$F(\mathbf{X})$ is non-increasing such that:
\begin{equation}  
    F(\mathbf{X}^{(t+1)}) \leq G(\mathbf{X}^{(t+1)}, \mathbf{X}^{(t)}) 
    \leq G(\mathbf{X}^{(t)}, \mathbf{X}^{(t)}) = F(\mathbf{X}^{(t)}).
    \label{eq:Gfunc}
\end{equation}  
\end{theorem}
\textit{Proof} Appendix~\ref{sec:thm442}.

\noindent The convergence analysis of the update rule for $a_{ij}$ is similar. 

 Although we proved that the algorithm has a non-increasing cost function that guarantees convergence, it may still get trapped in local minima. We show that multi-layer NMF with bounded initialization reduces the probability of converging to local minima.

\begin{definition}[Energy Barrier]
Let $D_{\alpha}(\mathbf{Y} \parallel \mathbf{A}\mathbf{X})$ be the cost function associated with the NMF algorithm. 
The energy barrier $\xi$ is defined as the difference between the highest cost encountered along an 
optimization path $\gamma$ and the cost at the global minimum:
\begin{equation}
\xi = \max_{\gamma} D_{\alpha}(\mathbf{Y} \parallel \mathbf{A}\mathbf{X}) 
      - D_{\alpha}(\mathbf{Y} \parallel \mathbf{A}^*\mathbf{X}^*),
\end{equation}
where $(\mathbf{A}^*, \mathbf{X}^*)$ is the global minimum solution, 
and $\gamma$ is a transition path in the optimization landscape.
\end{definition}

\begin{definition}[Boltzmann Probability]
The probability of escaping from a local minimum is given by:
\begin{equation}
P = \frac{1}{Z} e^{-\beta \xi},
\end{equation}
where $Z > 0$ is a normalization constant, $\beta > 0$ is an inverse temperature parameter 
that controls stochastic exploration, and $\xi$ is the energy barrier that must be overcome 
to escape local minima.
\end{definition}

\begin{lemma} 
Let $D_l$ represent the $\alpha$-divergence at layer $l$:
\begin{equation}
    D_l = D_{\alpha}(\mathbf{X}^{(l-1)} \parallel \mathbf{A}^{(l)} \mathbf{X}^{(l)}).
\end{equation}
Then, for all $l > 1$, we have:
\begin{equation}
    D_l \leq D_{l-1}.
\end{equation}
\end{lemma}
\textit{Proof} Appendix~\ref{sec:lm442}.

\begin{theorem} 
Let $P_l$ represent the probability of escaping from a local minimum
at layer $l$. Then, for all sufficiently large $l$ we have: 
\begin{equation}
    P_l \;\ge\; P_{l-1}.
\end{equation}
\end{theorem}
\textit{Proof} Appendix~\ref{sec:thm443}.
\begin{lemma} 
The escape probability $P_l$ converges to a finite value:
\begin{equation}
\lim_{l\to\infty} P_l \;=\; P_\infty.
\end{equation}
\end{lemma}
\textit{proof} Appendix~\ref{sec:lm443}.

\begin{theorem} 
Across multiple attempts, the multi-layer NMF algorithm has a smaller probability of remaining trapped in a local minimum compared to the single-layer NMF algorithm.
\end{theorem}

\begin{proof}
Let $L_e$ denote the number of attempts at which the algorithm escapes a local minimum. 
For each layer $l \in \mathbb{N}$, define the survival event as:
\begin{equation}
S_l = \{L_e > l\}, 
\end{equation}
which means the process has not yet escaped any local minimum by layer $l$.
\begin{remark}
We interpret each attempt as one run of the algorithm at a given layer. 
Thus, the $l$-th attempt corresponds to applying the algorithm at layer $l$. 
In the multi-layer setting, the algorithm proceeds through successive layers, 
while in the single-layer setting, all attempts are confined to the same layer. 
The total number of attempts is denoted by $n$, meaning the algorithm has been 
applied up to layer $n$.
\end{remark}

\begin{lemma}  
For all $n \geq 1$, the survival probability $\mathbb{P}(S_n)$ satisfies:
\begin{equation}
   \mathbb{P}(S_n) = \prod_{l=1}^{n} (1-P_l). 
\end{equation}
\end{lemma}
\textit{proof} Appendix~\ref{sec:lm444}. \\
By Lemma 4.4.3, $\lim\limits_{l\to\infty} P_l \;=\; P_\infty.$ 
The formal definition of the limit implies:
\begin{equation}
\forall \, \varepsilon > 0, \; \exists \, l_\varepsilon \in \mathbb{N} 
\;\; \text{such that} \;\; \forall\, l \ge l_\varepsilon 
\;\Longrightarrow\; P_\infty - P_l \le \varepsilon.
\end{equation}
Recall Lemma~4.4.4 and split the product at $l_\varepsilon$. Then, for any $n \ge l_\varepsilon$ we have:
\begin{align}
\mathbb{P}(S_n)
&= \prod_{l=1}^{n}(1-P_l) \notag = \underbrace{\Big(\prod_{l<l_\varepsilon}(1-P_l)\Big)}_{C_\varepsilon}\,
   \prod_{l=l_\varepsilon}^{n}(1-P_l) \notag \\
&\le C_\varepsilon \prod_{l=l_\varepsilon}^{n}\bigl(1-(P_\infty-\varepsilon)\bigr)
\qquad\text{(since $P_l \ge P_\infty - \varepsilon$)} \notag \\
&= C_\varepsilon\,\bigl(1-(P_\infty-\varepsilon)\bigr)^{\,n-l_\varepsilon+1}.
\label{eq:bound}
\end{align}
Hence,
\begin{equation}
\mathbb{P}(S_n) \;\le\; C_\varepsilon \,\bigl(1-(P_\infty-\varepsilon)\bigr)^{\,n-l_\varepsilon+1}.
\end{equation}
Let $\widehat{S}_n$ denote the survival event in the single-layer case. Then:
\begin{equation}
\mathbb{P}(\widehat{S}_n) = (1-P_1)^n.
\end{equation}
Since $P_\infty > P_1$, for any $\varepsilon \in (0,\, P_\infty - P_1)$ we have:
\begin{equation}
    1-(P_\infty-\varepsilon) < 1-P_1.
\end{equation}
Thus, for sufficiently large $n$ we obtain:
\begin{equation}
\mathbb{P}(S_n) \;\le\; 
C_\varepsilon \,\bigl(1-(P_\infty-\varepsilon)\bigr)^{\,n-l_\varepsilon+1}
\;<\; (1-P_1)^n \;=\; \mathbb{P}(\widehat{S}_n).
\end{equation}

This implies that across multiple attempts, the multi-layer NMF algorithm has a smaller probability of remaining trapped in a local minimum compared to the single-layer NMF algorithm. 
\qedhere \text{ QED.}
\end{proof}

\section{Experimental Results}

For image recognition, we employ the ORL face \cite{orl} and the MNIST digit \cite{mnist} datasets. The ORL dataset consists of 400 grayscale facial images, all resized to $32 \times 32$ pixels. The MNIST dataset contains 70{,}000 grayscale images of handwritten digits. For our experiments, we construct a balanced subset of 400 samples, each normalized to $28 \times 28$ pixels. In addition, we cluster time-frequency spectrograms of heart and lung sounds using the HLS-CMDS dataset. The lung subset consists of 50 recordings, divided into 6 classes. The heart subset contains 50 recordings of cardiac sounds, categorized into 10 classes. Figure~\ref{fig:bf-alpha} illustrates the sensitivity of pattern recognition with respect to the boundary factor ($BF$) and the divergence parameter $\alpha$. At $BF=0$, Chem-NMF reduces to the baseline $\alpha$-NMF. Adding $BF$ improves performance, particularly at intermediate values of $\alpha$. 

\begin{figure}[H]
    \centering
    \includegraphics[width=0.5\textwidth]{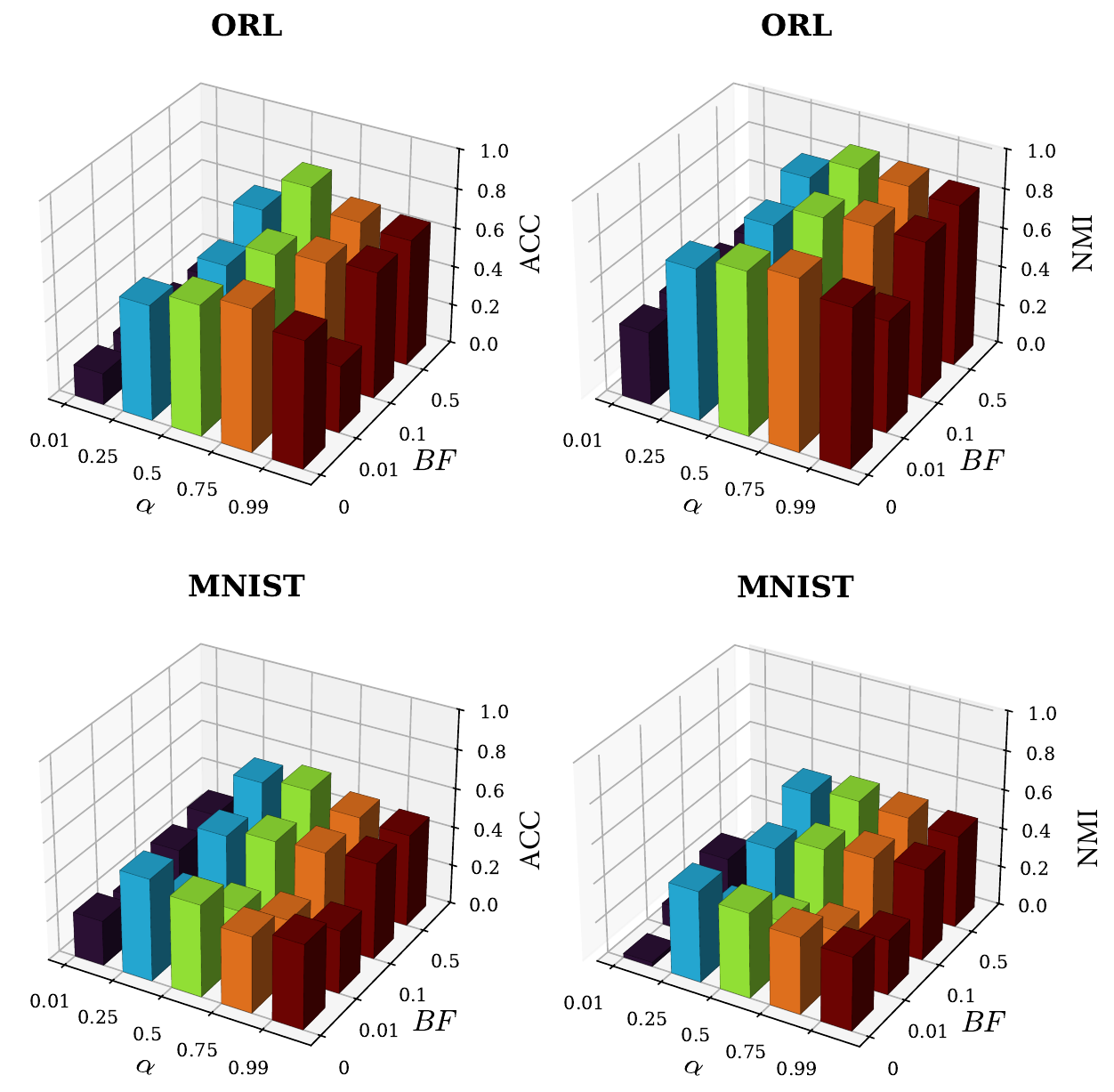}
    \captionsetup{justification=justified, font=footnotesize}
\caption{
Effect of the bounding factor ($BF$) and divergence parameter $\alpha$ on clustering performance for the ORL and MNIST datasets. 
The four subplots show how accuracy (ACC) and normalized mutual information (NMI) change as $BF$ and $\alpha$ vary. 
The top row corresponds to the ORL dataset (ACC left, NMI right), and the bottom row corresponds to the MNIST dataset. 
}

    \label{fig:bf-alpha}
\end{figure}

In addition, we perform unsupervised clustering on lung and heart sound datasets. We transform the recordings into time–frequency spectrograms, factorize the data into a low-rank representation, and cluster data using K-means, Gaussian mixture models (GMM), agglomerative clustering, and spectral clustering. We compare clustering performance without and with NMF feature extraction (Figure~\ref{fig:ablation}). 

\begin{figure}[H]
    \centering
    \includegraphics[width=0.8\textwidth]{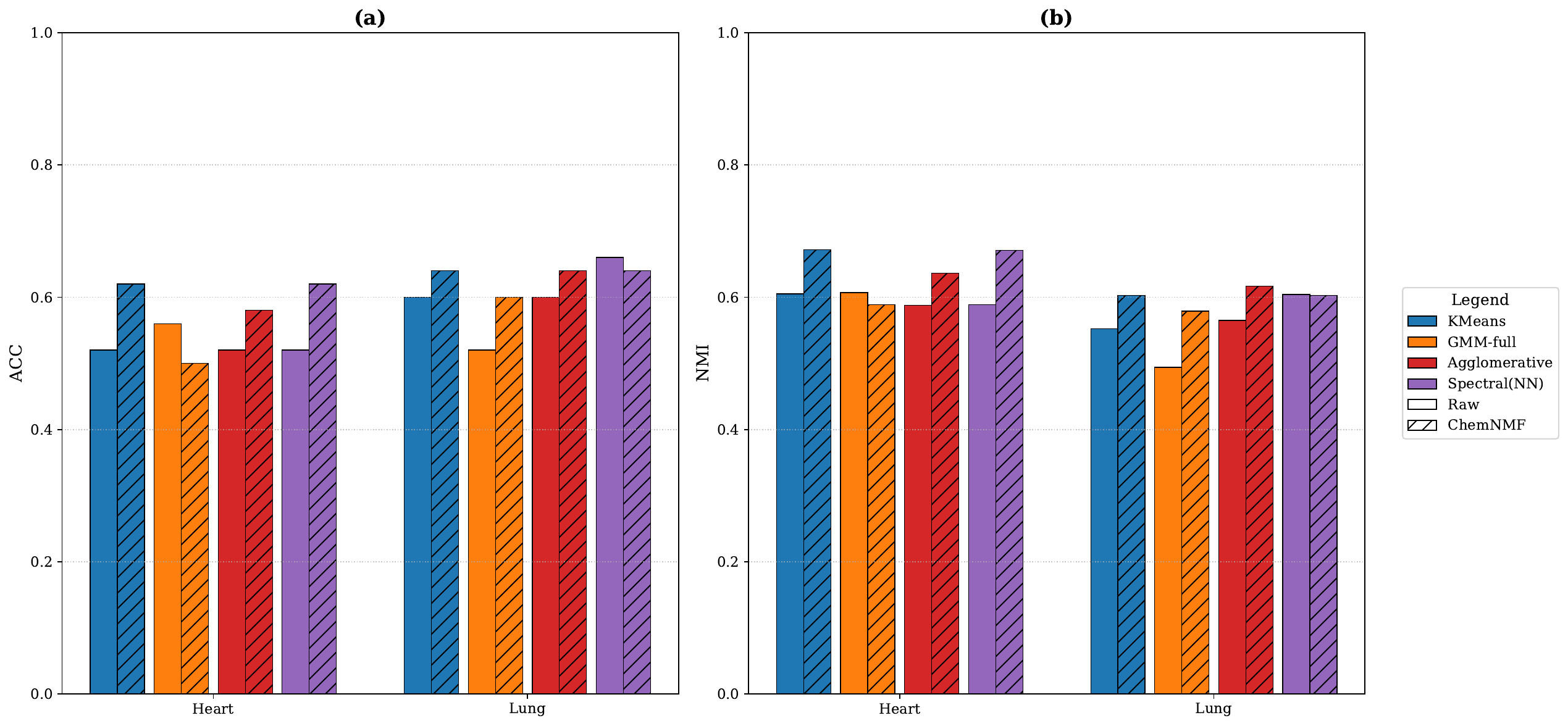}
    \captionsetup{justification=justified, font=footnotesize}
\caption{
Ablation study on the effect of feature extraction on clustering performance of cardiovascular sounds based on: \textbf{(a)} ACC and \textbf{(b)} NMI measures. 
The plots compare clustering performance \emph{with} and \emph{without} feature extraction across four clustering algorithms.}

    \label{fig:ablation}
\end{figure}

\section{Conclusion}
In this chapter, we introduced Chem-NMF, a multi-layer $\alpha$-divergence NMF framework inspired by energy barriers in chemical reactions. By incorporating a bounding factor analogous to a chemical catalyst, the method stabilizes convergence, reduces overshoot, and improves clustering performance compared to Regular NMF and $\alpha$-NMF. Theoretical analysis confirmed a lower probability of staying in local minima, while experiments on image and biomedical datasets demonstrated reasonable clustering accuracy. 
        \setcounter{figure}{0}
        \setcounter{equation}{0}
        \setcounter{table}{0} 
  \chapter{Quantum Convolutional Neural Network for Detecting Abnormal Patterns}

\newpage

This work introduces a hybrid quantum–classical convolutional neural network (QCNN) designed to classify the third heart sound (S3) and murmur abnormalities in heart sound signals. The approach transforms one-dimensional phonocardiogram (PCG) signals into compact two-dimensional images through a combination of wavelet feature extraction and adaptive threshold compression methods. We compress the cardiac-sound patterns into an 8-pixel image so that only 8 qubits are needed for the quantum stage. Preliminary results on the HLS-CMDS dataset demonstrate 93.33$\%$ classification accuracy on the test set, and 97.14$\%$ on the train set, suggesting that quantum models can efficiently capture temporal–spectral correlations in biomedical signals. 

\section{Introduction}

Heart sounds recorded as phonocardiograms (PCGs) contain temporal and spectral patterns linked to valve activity and cardiac rhythm. Classical signal-processing techniques, such as Fourier and wavelet transforms, have been used to extract features from PCGs \cite{Kannan2025}. Deep learning methods, such as convolutional neural networks (CNNs)  \cite{Patwa2025} and recurrent neural networks (RNNs) \cite{Hsieh2025} can detect murmurs with high precision. Transfer learning and attention mechanisms have further improved recognition rates on small datasets \cite{Alotaibi2025}. However, such networks still require strong computational hardware. These requirements restrict their real-time use on portable diagnostic devices \cite{Gholizade2025}. Quantum computing provides a new way to process data through qubits that can represent several states simultaneously. Quantum machine learning (QML) combines this property with classical optimization to improve computational efficiency. In hybrid networks, quantum circuits extract correlations in feature space while classical layers update the parameters. Variational quantum circuits (VQCs) \cite{Chen2024} and hybrid quantum neural networks have been applied to pattern classification, clustering, and feature selection \cite{Wang2024}. Quantum convolutional neural networks (QCNNs) extend this idea by using layers of quantum convolution and pooling gates that mimic hierarchical feature extraction \cite{Long2025}. These models have been used for biomedical classifications, such as breast cancer diagnosis \cite{Xiang2024}. Despite this progress, no earlier study has applied QCNNs to heart sounds. PCG signals differ from images and ECG traces because they contain short acoustic pulses with complex frequency variation. A quantum model must therefore compress and encode the information efficiently before processing. This study aims to demonstrate that a compact hybrid model can detect cardiac abnormalities with limited computational cost. Figure~\ref{fig:qcnn} illustrates our proposed method. The main contributions are:

1. Transforming sounds into images, which enables image processing methods for pattern recognition.

2. Introducing a compression pipeline using Wavelet transform, downsampling, and binarization that compresses images into eight pixels suitable for 8-qubit quantum encoding.

3. Designing a quantum convolutional neural network that learns hierarchical features on Qiskit simulators.

\begin{figure}[H]
    \centering
    \includegraphics[width=0.8\textwidth]{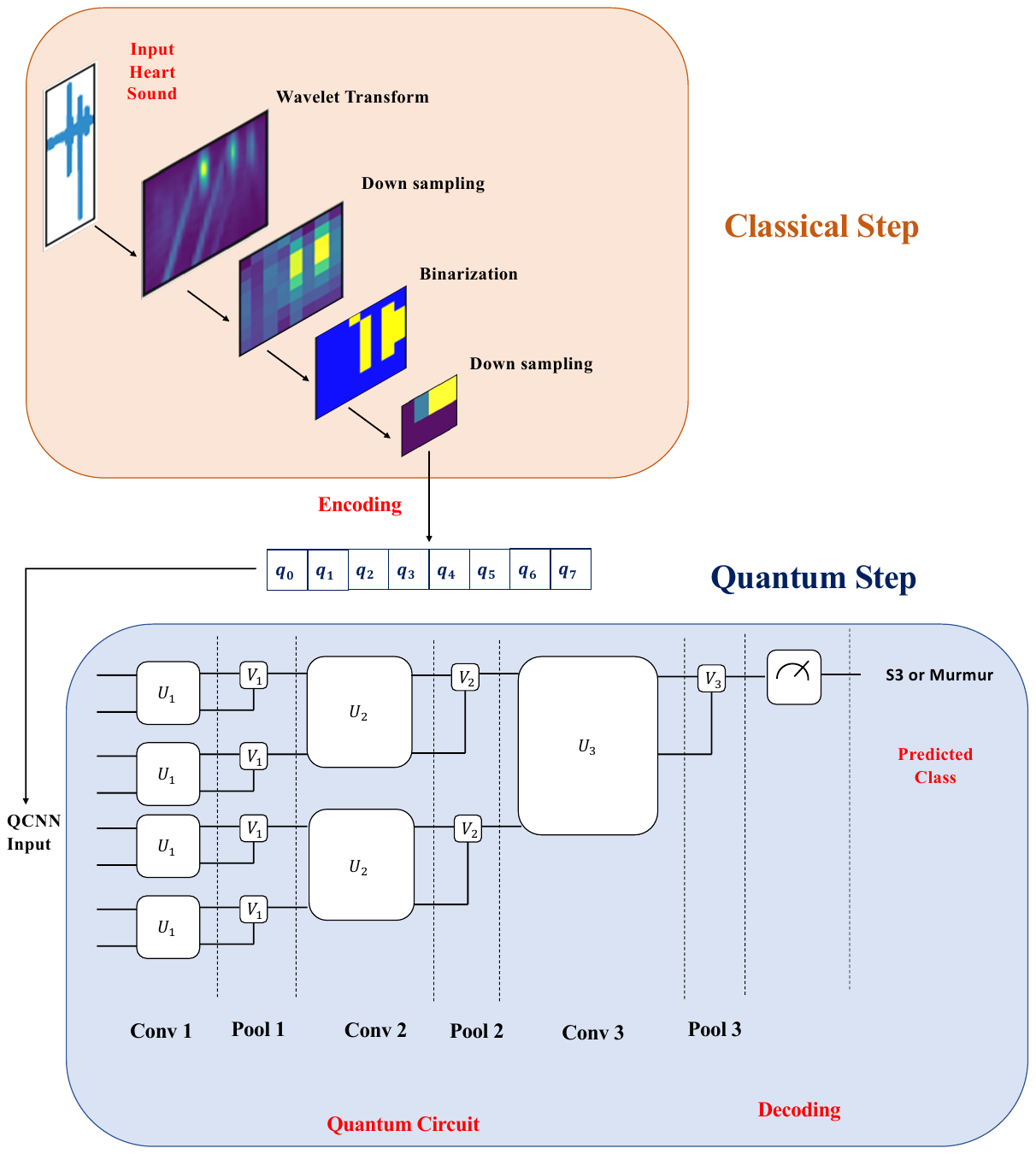}
    \captionsetup{justification=justified, font=footnotesize}
    \caption{Overview of the proposed hybrid quantum–classical model for heart sound classification (QuPCG). The classical step applies wavelet transform, downsampling, and binarization to generate compact quantum-ready 8-qubit feature maps. The quantum step encodes these features into qubits and processes them through three stages of successive quantum convolution and pooling layers, with the final measurement used to predict the class (S3 or murmur).}
    \label{fig:qcnn}
\end{figure}

\section{Methodology}
\subsection{Theoretical Background}
Heart-sound recordings are non-stationary signals whose frequency content changes with time. The wavelet transform produces a two-dimensional representation of the heart-sound signals suitable for image-based analysis. Convolutional neural networks (CNNs) extract patterns from these maps through local filtering and pooling. Each convolution layer computes the correlation between an input map and a kernel. Quantum machine learning (QML) extends these operations to quantum space. The information is encoded into qubits that exist in superposition (Equation~\ref{eq2}). The information is then processed through parameterized unitary gates $U(\theta)$ that evolve the state (Equation~\ref{eq3}). Measurement of an observable $O$ gives an expected value (Equation~\ref{eq4}), which serves as the network output. These expectation values feed the optimizers that adjust the rotation vector ${\theta}$ to minimize loss. This mechanism allows compact circuits to model complex nonlinear relationships beyond classical feature spaces.

\begin{equation}
\ket{\psi} = \alpha\ket{0} + \beta\ket{1}, \qquad |\alpha|^2 + |\beta|^2 = 1.
\label{eq2}
\end{equation}
\begin{equation}
\ket{\psi'} = U(\theta) \, \ket{\psi}.
\label{eq3}
\end{equation}
\begin{equation}
\braket{O} = \langle \psi' | \, O \, | \psi' \rangle.
\label{eq4}
\end{equation}

\subsection{Dataset Segmentation}
We used the heart sounds dataset (HLS-CMDS) collected from clinical manikins~\cite{10981596}. In a phonocardiogram (PCG), the first and second heart sounds (S1, S2) correspond to the closure of the heart valves, marking the start and end of each cardiac cycle. The third heart sound (S3) appears as a low-frequency vibration after S2, while murmurs present as prolonged oscillations. We focused on the classification of S3 versus murmur abnormal sounds. We segmented the recordings based on cardiac cycles (Figure~\ref{fig:pcg}). Each signal was resampled to 4 kHz and analyzed to identify the heart-sound peaks.

\begin{figure}[H]
    \centering
    \includegraphics[width=0.8\textwidth,trim= 0 40 0 20,clip]{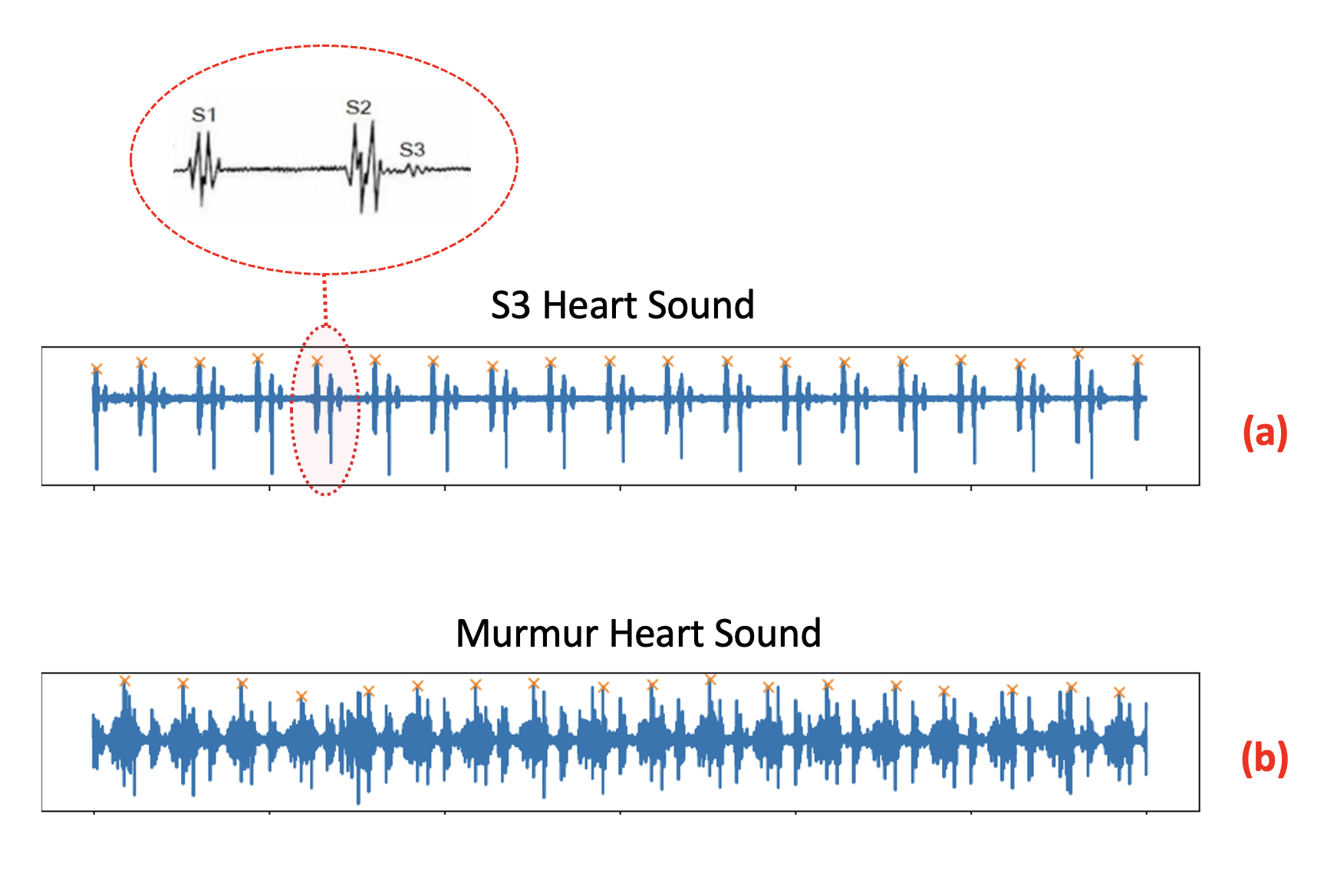}
    \captionsetup{justification=justified, font=footnotesize}
\caption{PCG signal segments: \textbf{(a)} Heart sound with additional S3; \textbf{(b)} Murmur heart sound. Both panels show pathological examples used as the two abnormal classes in the classification task.}

    \label{fig:pcg}
\end{figure}

\subsection{From Time-Series Signals to Energy Map Images}

We transformed each segment into a two-dimensional time–frequency representation using the wavelet transform with a complex morlet mother wavelet. We applied 128 scales to compute the scalogram and visualized the magnitude of the complex coefficients as an image. Unlike conventional approaches that process long time-series data, our method converts heart sounds into compact wavelet images, allowing us to apply image-based pattern recognition techniques directly. This transformation enables effective compression, which is an essential advantage when operating with a limited number of qubits. The wavelet-based representation also preserves transient events such as murmur patterns more effectively. The scalograms were resized to 32 × 32 pixels to provide a consistent input format for the next compression stage.

\subsection{Feature Compression}
As shown in Figure~\ref{fig:FC}, we compress each 32 × 32 scalogram into a compact quantum-ready format. Max-pooling with 4 × 4 kernel downsamples the image to 8 × 8. We binarize the map into high and low-energy regions. The 8 × 8 maps are reduced to eight representative values to align with the 8-qubit architecture.

\begin{figure}[H]
    \centering
    \includegraphics[width=0.8\textwidth, trim=0 120 0 120, clip]{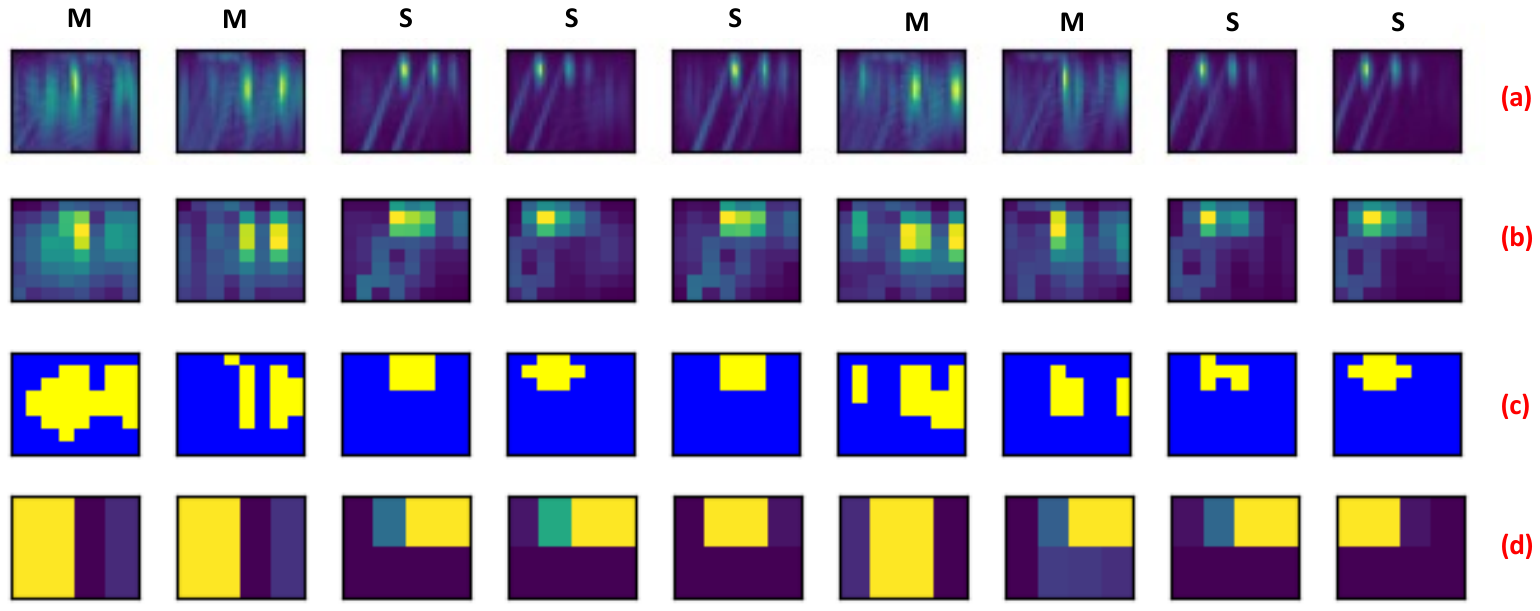}
    \captionsetup{justification=justified, font=footnotesize}
\caption{Progressive compression of wavelet scalograms for quantum encoding: \textbf{(a)} Original 32 × 32 time–frequency images of murmur and S3 sounds; \textbf{(b)} 8 × 8 max-pooled maps; \textbf{(c)} Binarized energy patterns highlighting dominant regions; \textbf{(d)} Final 8-value representations matched to the 8-qubit QCNN input format. Here, M denotes murmur segments and S denotes S3 segments, and each column corresponds to different signal segments extracted from different patients.}

    \label{fig:FC}
\end{figure}

\subsection{Quantum Encoding and QCNN Design}
We mapped each normalized pixel intensity to the rotation angle of a single-qubit gate. Figure~\figsubref{fig:QCNNs}{a} shows the feature-mapping circuit implemented, where each qubit undergoes Hadamard and phase-rotation gates to embed classical features into the quantum state space. The QCNN contained three sets of alternating convolutional and pooling layers. Figure~\figsubref{fig:QCNNs}{b} illustrates parametrized unitary gates $U(\theta)$ used in the convolution circuit. We apply a pooling layer after the convolutional layer to reduce the dimensions of the quantum circuit. Figure~\figsubref{fig:QCNNs}{c} shows the two-qubit pooling circuit $V(\theta)$. This layer merges the information of two qubits into one. The final measured qubit provides the classification feature. After each quantum execution, the expectation value of the Pauli-$Z$ observable was measured, and the loss between predicted and actual labels was computed to update the gate parameters iteratively.

\begin{figure}[H]
    \centering
    \includegraphics[width=\textwidth, trim=0 100 20 95, clip]{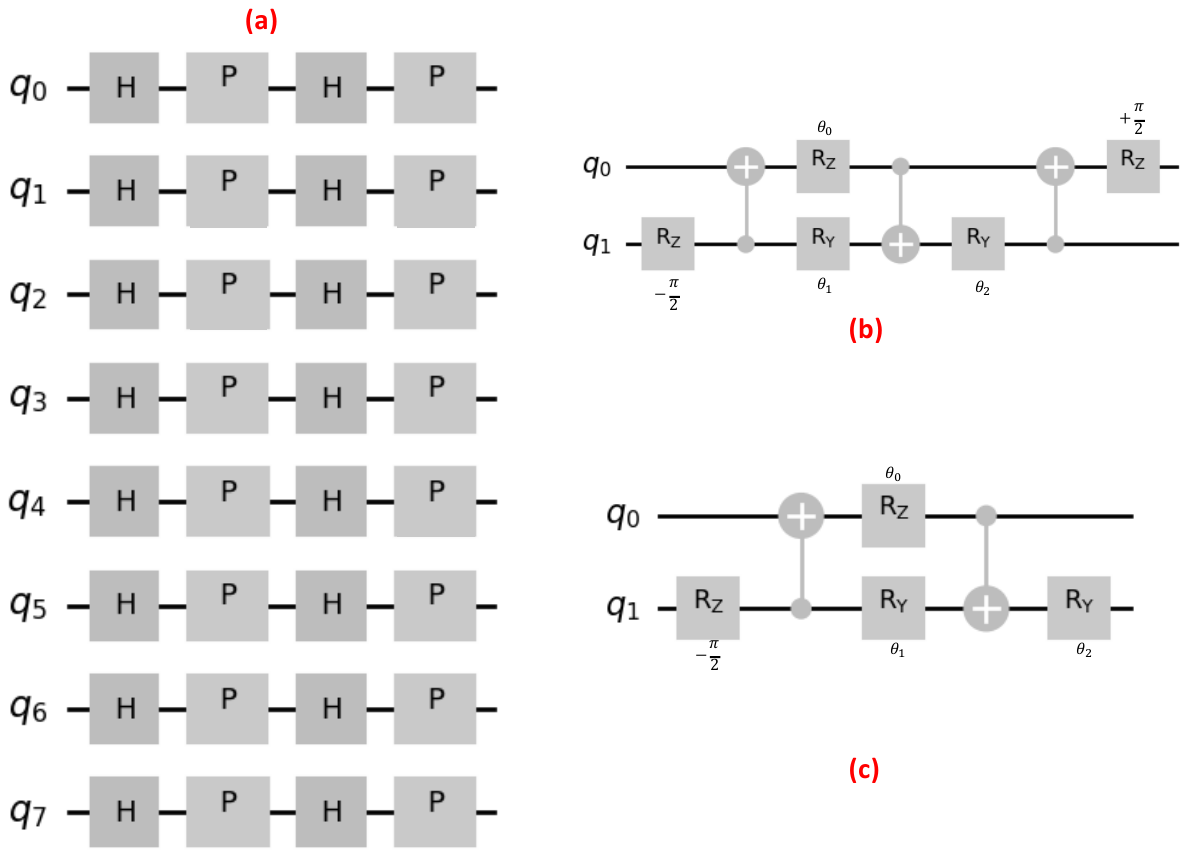}
    \captionsetup{justification=justified, font=footnotesize}
    \caption{Quantum convolutional neural network components: \textbf{(a)} Feature-mapping circuit; \textbf{(b)} Parametrized two-qubit unitary circuit used in the convolutional layer; \textbf{(c)} Parametrized two-qubit unitary circuit used in the pooling layer.}
    \label{fig:QCNNs}
\end{figure}

\section{Experimental Settings and Results}
We implemented the quantum circuit in Qiskit version 0.45. The network employed 8 qubits with a circuit depth of 3 layers. We experimented with three approaches for converting PCG time-series signals into two-dimensional representations suitable for the QCNN. Table~\ref{tab:QCNN_Performance} compares the performance of different QCNNs under various signal-to-image preprocessing strategies. In the first method (I-QuPCG), the raw time-series signals were directly converted into grayscale images; however, this approach yielded poor performance, as the temporal information was largely lost during compression. The second approach (M-QuPCG) used the Short-Time Fourier Transform (STFT) to generate Mel spectrograms, which captured more meaningful frequency patterns and slightly improved accuracy. Finally, the W-QuPCG applied a wavelet transform to obtain time–frequency maps that preserved the features of biomedical signals more effectively. In addition, W-QuPCG demonstrates the fastest convergence and achieves the highest accuracy with the lowest loss, confirming its superior learning efficiency, as shown in Figure~\ref{fig:QCNN_Performance}. Building upon this, we further increased the number of training epochs to 1000, resulting in the enhanced W-QuPCG$^{+}$, which achieved the best overall performance and demonstrated the optimal configuration of the proposed model.

\begin{table}[H]
\centering
\footnotesize
\captionsetup{font=footnotesize}
\caption{Performance comparison of different implemented QCNNs (mean~$\pm$~std).}
\label{tab:QCNN_Performance}
\begin{tabular}{lcccc}
\hline
\textbf{Method} & \textbf{Accuracy (Train)} & \textbf{Accuracy (Test)} & \textbf{Loss (Train)} & \textbf{Loss (Test)} \\
\hline
I-QuPCG   & 51.06~$\pm$~2.3\% & 47.62~$\pm$~2.8\% & 1.00~$\pm$~0.25 & 1.10~$\pm$~0.54 \\
M-QuPCG   & 74.29~$\pm$~5.2\% & 53.33~$\pm$~4.0\% & 0.80~$\pm$~0.13 & 0.91~$\pm$~0.25 \\
W-QuPCG   & 91.43~$\pm$~2.9\% & 80.00~$\pm$~6.1\% & 0.69~$\pm$~0.12 & 0.73~$\pm$~0.03 \\
\textbf{W-QuPCG$^{+}$} & 97.14~$\pm$~4.6\% & 93.33~$\pm$~2.9\% & 0.42~$\pm$~0.62 & 0.45~$\pm$~0.12 \\
\hline
\end{tabular}
\end{table}

\begin{figure}[H]
    \centering
    \includegraphics[width=\textwidth, trim=0 20 0 20, clip]{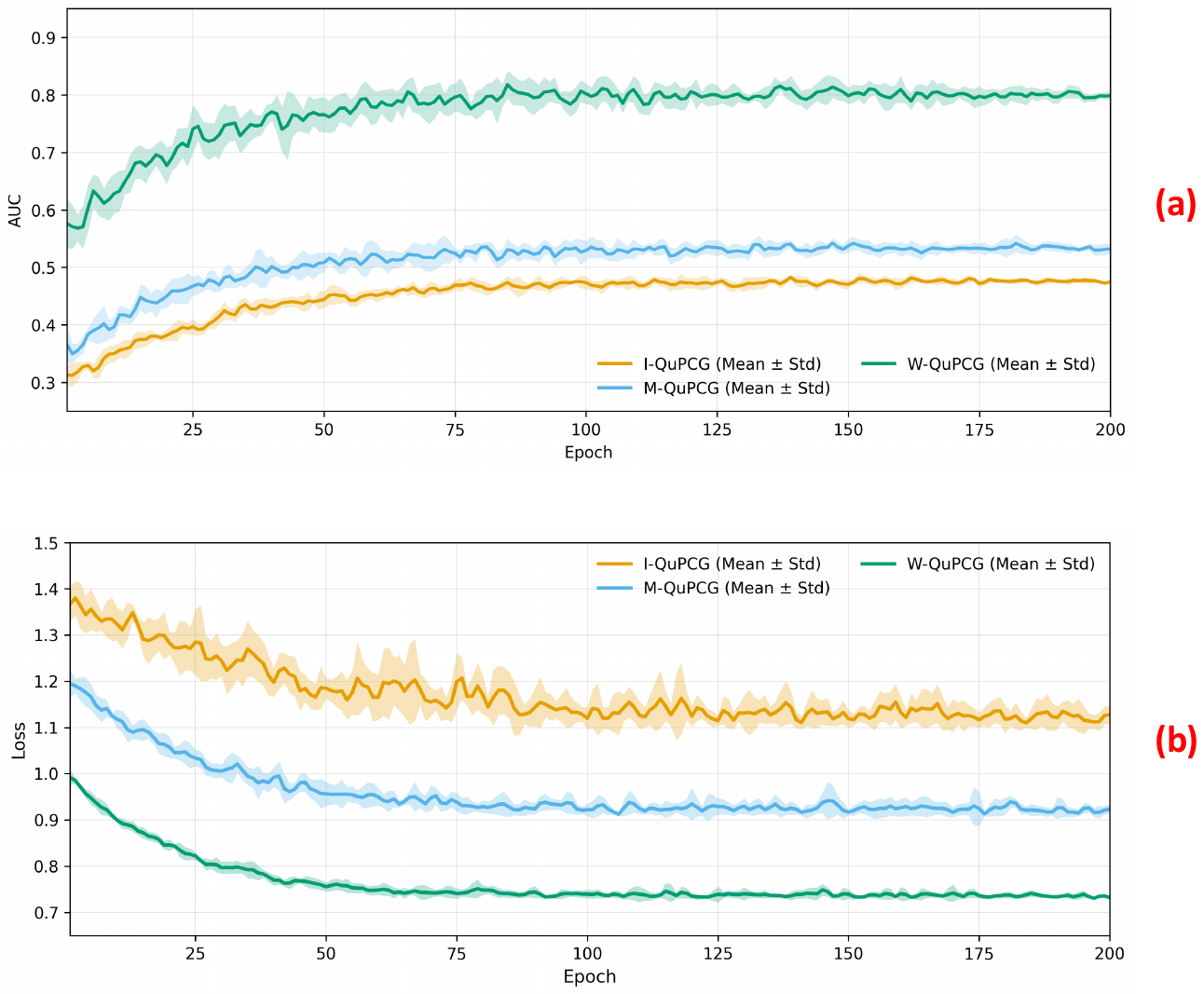}
    \captionsetup{justification=justified, font=footnotesize}
    \caption{Performance analysis of different QuPCG variants: \textbf{(a)} Accuracy; \textbf{(b)} Loss function value.}
    \label{fig:QCNN_Performance}
\end{figure}

\section{Discussion}
The results indicate that the hybrid quantum–classical model can distinguish abnormal heart sound patterns using very limited pixels. The wavelet transformation converts signals into energy maps, which allow the quantum circuit to analyze correlations that are difficult to capture in the time domain. The 8-qubit QCNN performed an acceptable separation between murmur and normal sounds. This supports the idea that small-scale quantum models can perform pattern recognition tasks efficiently when guided by well-engineered preprocessing and compression stages. Although the study demonstrates promising performance, there are practical challenges that must be addressed. The use of simulated backends cannot capture all physical noise sources that would occur on real quantum hardware. Future work should extend this research to real devices and compare the effect of decoherence on model stability. The dataset size is another limitation since quantum models benefit from diverse and balanced samples. Expanding the dataset to include multiple cardiac conditions and patient populations will improve generalization. In future studies, additional encoders, adaptive pooling circuits, and integration with explainable modules may help to clarify how quantum layers learn medical features.
\section{Conclusion}
This work presented a hybrid quantum–classical convolutional network for the classification of abnormal heart sound patterns. The method combines wavelet-based feature extraction with adaptive compression and quantum encoding to create a compact learning model. The results show that even with a small number of qubits, quantum circuits can process biomedical data effectively when designed with domain-specific constraints. These findings support the potential of quantum learning for diagnostic applications and motivate further investigation using real quantum processors and larger clinical datasets.
        \setcounter{figure}{0}
        \setcounter{equation}{0}
        \setcounter{table}{0} 
  \chapter{Conclusion}

This thesis presented a study on artificial intelligence and quantum methods for cardiopulmonary signal analysis and sensing. The work combined classical signal processing and next-generation intelligent sensing to enhance the separation, classification, and acquisition of heart and lung sounds. The studies covered both algorithmic development and sensor-level design, from electronic and photonic systems to integrated quantum biosensors. At the algorithmic level, non-negative matrix factorization (NMF) was revisited in combination with large language models (LLMs) to improve blind source separation of heart and lung sounds. The results confirmed that adding contextual feedback from LLMs improved the separation of overlapping signals. However, the model still depends on the interpretability of the language model feedback. While the LLaMA-based module showed reliable responses, its general-purpose nature limits precision for clinical sound data. Future work will include domain adaptation of medical LLMs, prompt optimization, and robustness tests using real patient recordings. Extending validation to clinical datasets will also reduce the gap between controlled experiments and real diagnostic environments.

The second contribution of this thesis introduced a generative AI model based on the variational autoencoder (VAE) architecture. The proposed model separated mixtures of heart and lung sounds in an unsupervised manner by mapping them into a probabilistic latent space. The addition of a temporal-consistency loss helped the model maintain smooth signal transitions and capture cyclic patterns in cardiopulmonary sounds. Using wavelet-based spectrograms as input allowed the network to represent transient frequency components effectively. The current encoder–decoder network may be improved by alternative latent sampling or quantum-inspired encoding schemes to increase efficiency. In addition, validation on large, clinically diverse datasets is required to test robustness across patient age groups and disease conditions. Future research can also explore hybrid models that combine probabilistic VAEs with explainable AI (XAI) tools to improve interpretability.

The third stream of this research examined the use of quantum artificial intelligence for biomedical signal processing. Quantum computing offers a new computational model that can process multiple states in parallel using qubits. Quantum machine learning (QML) unites this principle with classical optimization to extract correlations from high-dimensional biomedical data. In this work, hybrid quantum–classical neural networks, including quantum convolutional neural networks (QCNNs), were explored for phonocardiogram (PCG) classification. The design employed quantum convolution and pooling gates to analyze temporal and spectral patterns within compact qubit spaces. Compared to classical CNNs, QCNNs can reduce the number of trainable parameters while maintaining accuracy. This property is valuable for edge devices with limited resources.
Nevertheless, practical challenges remain. Current quantum simulators are restricted by qubit count and noise, which limit model scalability. The encoding of analog biomedical signals into quantum states also requires efficient normalization and compression. Future work should focus on hardware–software co-design, where quantum circuits and classical preprocessing are optimized jointly. In addition, improved quantum feature maps and parameterized circuits could enhance performance on real-world datasets. With further optimization, QCNNs can play a key role in compact, energy-efficient diagnostic systems.

The final part of this thesis focused on biosensing devices that enable data collection for AI-based analysis. The review covered the progression from classical electronic integrated circuits (EICs) to photonic integrated circuits (PICs) and quantum biosensors. EIC-based sensors provided the first on-chip biomedical platforms but face issues such as electrical noise and limited multiplexing. PIC-based sensors solved some of these constraints by using light–matter interaction, offering higher sensitivity and immunity to electromagnetic interference. Quantum biosensors, including quantum dots and nitrogen-vacancy (NV) centers, extend sensing to the quantum regime by detecting single molecules and weak magnetic fields. Despite their sensitivity, quantum biosensors still require advances in material quality, and coherence control. The future direction toward integrated quantum photonics (IQP) combines EIC, PIC, and quantum sensing elements on a single chip. Such integration could allow room-temperature operation, CMOS compatibility, and scalable manufacturing for next-generation diagnostic systems.

In summary, this thesis demonstrated how classical and quantum AI algorithms can enhance biomedical signal analysis and how future biosensing devices can support these advances through on-chip integration. The findings suggest several key directions for future research. First, algorithmic robustness must be improved through real clinical validation, domain-specific language models, and explainable generative architectures. Second, sensing platforms must move toward hybrid integration of electronic, photonic, and quantum materials. Finally, the convergence of AI and quantum hardware can form the foundation for intelligent diagnostic systems capable of real-time, noninvasive, and personalized healthcare monitoring. The progress in this field points to a future where intelligent quantum biosensors and AI-driven algorithms work together on portable platforms to provide accurate and accessible medical assessment for all patients.
        \setcounter{figure}{0}
        \setcounter{equation}{0}
        \setcounter{table}{0}

\begin{appendix}
		
	\label{appendix_Tables}
\chapter{Tables}

\begin{table}[H]
\centering
\scriptsize
\captionsetup{font=footnotesize}
\caption{Comparison of LingoNMF and other separation methods.}
\hspace*{-3cm} 
\label{tab:lingonmf_comparison}
\begin{tabular}{c p{5 cm} c c c p{7 cm}}
\toprule
\textbf{Ref.} & \textbf{Method} & \textbf{SIR [dB]} & \textbf{SAR [dB]} & \textbf{SDR [dB]} & \textbf{Description} \\
\midrule
\cite{CanadasQuesada2017} & Spectro-Temporal Clustering NMF & 24.1 & 17.9 & 17.3 & NMF variant using spectro-temporal clustering \\
\cite{Wang2023_SNMF_DCNN_MCUNet} & SNMF-DCNN-MCUNet & 28.9 & 14.9 & 14.4 & Supervised NMF with multi-channel U-Net deep CNN \\
\cite{Xie2019_UBSS_HOS_SR} & BSS + HOS & 13.3 & 14.5 & 9.9 & Blind source separation using higher-order statistics \\
\cite{Lei2018_LSTM} & LSTM & 26.7 & 13.9 & 13.3 & Recurrent neural network (long short-term memory) \\
\cite{UNet_ref_if_any} & U-Net & 21.2 & 11.7 & 10.9 & Convolutional encoder–decoder model \\
\cite{Wang2023_ECNet} & ECNet & 31.1 & 15.1 & 14.8 & Hybrid NMF–deep learning framework \\
\cite{Tsai2020_PCDAE} & PC-DAE & 14.9 & 16.7 & 12.5 & Periodic-coded deep autoencoder\\
\cite{Xie2021_AuxFuncNMF} & NMF + Auxiliary Function & 14.7 & 8.1 & 7.1 & NMF variant for reverberant recordings \\
\midrule
\textbf{This Work} & \textbf{LingoNMF} & \textbf{22.4} & \textbf{25.2} & \textbf{22.3} & NMF with large language model (LLM) guidance \\
\bottomrule
\end{tabular}
\vspace{2mm}
\end{table}

\begin{table}[H]
\centering
\scriptsize
\setlength{\tabcolsep}{4.5pt}
\renewcommand{\arraystretch}{1.1}
\captionsetup{font=footnotesize}
\caption{Comparison of VAE sound separation performance.}
\label{tab:results_all}
\begin{tabular}{lcccccccc}
\toprule
\textbf{Method} & \multicolumn{4}{c}{\textbf{Dataset One}} & \multicolumn{4}{c}{\textbf{Dataset Two}} \\
\cmidrule(lr){2-5}\cmidrule(lr){6-9}
& SDR \,$\uparrow$& SIR\,$\uparrow$ & SAR \,$\uparrow$& Execution Time\,$\downarrow$ & SDR\,$\uparrow$ & SIR\,$\uparrow$ & SAR\,$\uparrow$ & Execution Time\,$\downarrow$ \\
\midrule
VAE        & 8.7  & 7.9  & 17.2 & 66.3 & 6.3  & 6.0  & 21.0 & 44.3 \\
VAE-M      & 10.8 & 9.9  & 18.2 & 68.6 & 12.4 & 12.0 & 24.5 & 45.0 \\
VAE-T     & 15.9 & 14.9 & 20.0 & 68.3 & 8.2  & 8.0  & 22.0 & 44.7 \\
VAE-MT   & 17.1 & 19.2 & 20.3 & 69.8 & 14.0 & 13.5 & 26.2 & 45.1 \\
VAE-W      & 18.1 & 14.9 & 19.1 & 44.2 & 11.7 & 11.1 & 23.1 & 43.6 \\
VAE-WM    & 22.2 & 26.1 & 23.8 & 46.3 & 15.5 & 16.8 & 28.7 & 43.9 \\
VAE-WT   & 25.8 & 30.0 & 28.1 & 45.8 & 16.8 & 17.2 & 30.2 & 44.1 \\
\textbf{VAE-WMT}
           & \textbf{26.8} & \textbf{32.8} & \textbf{28.6} & \textbf{46.8}
           & \textbf{15.1} & \textbf{20.7} & \textbf{17.6} & \textbf{44.3} \\
\bottomrule
\end{tabular}
\end{table}

\begin{table}[H]
\centering
\scriptsize
\setlength{\tabcolsep}{6pt}
\renewcommand{\arraystretch}{1.1}
\captionsetup{font=footnotesize}
\caption{Comparison of the proposed VAE method with other baselines.}
\label{tab:vae_vs_nonNMF}
\begin{tabular}{lccc}
\toprule
\textbf{Method} & \textbf{SIR (dB)} & \textbf{SAR (dB)} & \textbf{SDR (dB)} \\
\midrule
UBSS (HOS + SR) \cite{Xie2019_UBSS_HOS_SR} & 13.3 & 14.5 & 9.9 \\
LSTM \cite{Lei2018_LSTM}                  & 26.7 & 13.9 & 13.3 \\
U-Net \cite{UNet_ref_if_any}              & 21.2 & 11.7 & 10.9 \\
ECNet \cite{Wang2023_ECNet}               & 31.1 & 15.1 & 14.8 \\
PC-DAE \cite{Tsai2020_PCDAE}              & 14.9 & 16.7 & 12.5 \\
\midrule
\textbf{VAE-WMT}                          & \textbf{32.8} & \textbf{28.6} & \textbf{26.8} \\
\bottomrule
\end{tabular}

\vspace{2pt}
\begin{flushleft}
\scriptsize
\justifying
\textit{Note:} UBSS (HOS + SR): underdetermined blind source separation using higher-order statistics and sparse representation; LSTM: long short-term memory network; U-Net: U-shaped convolutional neural network; ECNet: embedding centroid network; PC-DAE: periodicity-coded deep autoencoder; VAE-WMT: variational autoencoder with wavelet transform, output masking, and temporal consistency.
\end{flushleft}
\end{table}

\begin{table}[H]
\centering
\scriptsize
\setlength{\tabcolsep}{4pt}
\renewcommand{\arraystretch}{1.1}
\captionsetup{font=footnotesize}
\caption{Comparison of VAE Latent space clustering performance.}
\label{tab:latent_metrics}
\begin{tabular}{lcccccccc}
\toprule
\textbf{Algorithm} & \multicolumn{4}{c}{\textbf{Dataset One}} & \multicolumn{4}{c}{\textbf{Dataset Two}} \\
\cmidrule(lr){2-5}\cmidrule(lr){6-9}
& Silh.\,$\uparrow$ & DB\,$\downarrow$ & CH\,$\uparrow$ & Var\,$\downarrow$
& Silh.\,$\uparrow$ & DB\,$\downarrow$ & CH\,$\uparrow$ & Var\,$\downarrow$ \\
\midrule
VAE        & 0.324 & 4.433 & 114.7 & 0.604 & 0.259 & 4.473 & 126.2 & 0.799 \\
VAE-M      & 0.274 & 4.868 & 107.9 & 0.623 & 0.345 & 4.117 & 128.8 & 0.659 \\
VAE-T     & 0.280 & 4.469 & 122.7 & 0.676 & 0.311 & 4.743 & 116.5 & 0.616 \\
VAE-MT   & 0.310 & 4.511 & 119.2 & 0.649 & 0.322 & 4.257 & 126.1 & 0.648 \\
VAE-W      & 0.271 & 4.437 & 127.3 & 0.713 & 0.237 & 3.084 & 229.9 & 1.075 \\
VAE-WM    & 0.311 & 4.743 & 116.5 & 0.616 & 0.310 & 4.511 & 119.2 & 0.649 \\
VAE-WT   & 0.322 & 4.257 & 126.1 & 0.648 & 0.271 & 4.437 & 127.3 & 0.713 \\
\textbf{VAE-WMT}
           & \textbf{0.345} & \textbf{4.117} & \textbf{128.8} & \textbf{0.659}
           & \textbf{0.324} & \textbf{4.433} & \textbf{114.7} & \textbf{0.604} \\
\bottomrule
\end{tabular}
\end{table}

\begin{table}[H]
\centering
\scriptsize 
\captionsetup{font=footnotesize}
\caption{Examples of Quantum Dot and NV Center Biosensors}

\setlength{\fboxrule}{0.8pt}

\begin{adjustbox}{angle=90}
\fbox{%
    \begin{tabular}{p{1cm} p{4.5cm} p{4.2cm} p{6.2cm} p{1cm}}
    \textbf{Type} & \textbf{Application} & \textbf{Specification} & \textbf{Working Principle} & \textbf{Ref} \\
    \midrule
    \multirow{6}{*}{\centering \uprightcell{\bfseries Quantum Dot Biosensors}}
    & Streptomycin antibiotics detection 
    & LOD: 0.0033 pg/mL; Linear range: 0.01--812.21 pg/mL 
    & QDs increase electrode conductivity and active area, enabling sensitive detection via electron transfer changes upon target binding. 
    & \cite{C72ref92} \\
    
    & Anticancer (flutamide) drug detection in urine 
    & LOD: 0.0169 $\mu$M; Linear range: 0.05--590 $\mu$M 
    & NCQDs enhance electron transfer and surface area; quantum confinement improves electrocatalytic sensing. 
    & \cite{C72ref93} \\
    
    & Ovarian cancer biomarker detection 
    & LLOQ: 0.001 U/mL; Linear range: 0.001--400 U/mL 
    & Antibody-modified GQD nano-ink enables electrochemical detection via antigen--antibody binding. 
    & \cite{C72ref94} \\
    
    & Cancer metastasis biomarkers 
    & LOD: 0.03 ng/mL
    & GQDs enhance signal via peroxidase-like activity; hybrid nanocarriers amplify current from sandwich immunoassay. 
    & \cite{C72ref95} \\
    
    & Rheumatoid arthritis detection
    & Linear range: $1\times10^{4}$--$10$ cells/mL; LOD: 2 cells/mL 
    & QDs and perovskite enhance charge separation, boosting photocurrent for ultrasensitive detection. 
    & \cite{C72ref91} \\
    
    & ClO$^{-}$ ion detection; imaging in cells 
    & LOD: 13 nM; Linear range: 0.05--1.8 $\mu$M 
    & Sulfhydryl groups enable selective and rapid fluorescence quenching by ClO$^{-}$. 
    & \cite{C72ref14} \\
    \midrule
    
    \multirow{4}{*}{\centering \uprightcell{\bfseries NV Center Biosensors}}
    & Magnetic field detection 
    & Sensitivity: 1 nT/$\sqrt{\mathrm{Hz}}$ 
    & Optical spin-state readout with pulsed noise filtering. 
    & \cite{C72ref15} \\
    
    & HIV-1 RNA detection 
    & LOD: $8.2\times10^{-19}$ M 
    & Microwave-modulated NV fluorescence separates signal from background, enabling record-low detection limits. 
    & \cite{C72ref97} \\
    
    & SARS-CoV-2 RNA detection 
    & LOD: few hundred RNA copies; $<1$\% false negatives 
    & DNA-generated magnetic tags detected via NV magnetometry. 
    & \cite{C72ref99} \\
    
    & Free radical detection in living cells 
    & Single-cell resolution; real-time monitoring 
    & NV spin relaxation changes with local magnetic noise from free radicals. 
    & \cite{C72ref98} \\
    \end{tabular}%
}
\end{adjustbox}

\label{tab:qd_nv_table}
\end{table}

\begin{table}[H]
\centering
\scriptsize 
\captionsetup{font=footnotesize}
\caption{Examples of MZI and MRR Photonic Integrated Circuit (PIC) Biosensors}

\setlength{\fboxrule}{0.8pt}

\begin{adjustbox}{angle=90}
\fbox{%
    \begin{tabular}{p{1cm} p{5.5cm} p{4.8cm} p{4.5cm} p{1cm}}
    \textbf{Type} & \textbf{Working Principle} & \textbf{Specification} & \textbf{Application} & \textbf{Ref} \\
    \midrule
    \multirow{6}{*}{\centering \uprightcell{\bfseries MZI-based PIC}}
    & Optical phase shift between split paths due to molecular binding 
    & Sensitivity: $6.8 \times 10^{-6}$ RIU; LOD: $\sim$1 ng/mL; Footprint: $3.5 \times 0.6$ mm$^{2}$
    & Streptavidin--biotin sensing 
    & \cite{C72ref113} \\

    & Spiral MZI with enhanced evanescent interaction 
    & LOD: 0.25 pg/mm$^{2}$; Spiral length: 1.8 mm; Channel width: 190 $\mu$m  
    & Antibody--antigen detection 
    & \cite{C72ref114} \\

    & 3D laser-written MZI with orthogonal fluidic arm 
    & LOD: $1\times10^{-4}$ RIU; Spatial resolution: 10 $\mu$m; Channel: 150 $\mu$m 
    & Spatially resolved label-free sensing 
    & \cite{C72ref115} \\

    & Rib waveguide MZI (Si$_3$N$_4$)
    & LOD: $7\times10^{-6}$ RIU; Sensor length: 15 mm; Core: 250 nm  
    & Water pollutant detection 
    & \cite{C72ref8} \\

    \midrule
    \multirow{5}{*} {\centering \uprightcell{\bfseries MRR-based PIC}}
    & Subwavelength grating ring -- resonance shift from antigen binding 
    & LOD: 1.31 fM; Detection time: 15 min; Q-factor: $\sim$30,000 
    & SARS-CoV-2 \& influenza detection 
    & \cite{C72ref9} \\

    & SiN ring array with photonic packaging 
    & Q-factor: $>4\times10^{4}$
    & Respiratory antibody profiling 
    & \cite{C72ref10} \\

    & Slotted plasmonic ring with enhanced field confinement 
    & Sensitivity: 1609 nm/RIU; Shift: 29.6 nm (for $\Delta$RI = 0.0184); Slot: 10 nm 
    & Water pollutant detection 
    & \cite{C72ref116} \\

    & Rectangular semi-ring waveguide biosensor 
    & LOD: 0.12 ng/mL; Linear: 0.1--100 ng/mL; Width: 800 nm 
    & Hepatitis B virus detection 
    & \cite{C72ref117} \\

    & Photonic crystal ring + PN phase shifter 
    & Tuning range: 1.3--1.7 $\mu$m; Radius: 6 $\mu$m; Phase shift: $\pi$ 
    & Tunable biomedical sensing 
    & \cite{C72ref118} \\

    & Racetrack dual-ring interrogation 
    & Sensitivity: 110 nm/RIU; Delay: 10 ns @100 MHz; Tolerance: $\pm$8 nm 
    & Liquid refractometry 
    & \cite{C72ref119} \\
    \end{tabular}%
}
\end{adjustbox}

\label{tab:pic_comparison}
\end{table}

		\setcounter{figure}{0}
		\setcounter{equation}{0}
		\setcounter{table}{0}
\label{appendix_Algorithms}
\chapter{Algorithms}
\vspace{-1 cm}

\begin{breakablealgorithm}
\caption{PL-NMF}
\label{alg:plnmf}
\begin{spacing}{1}
\small
\begin{algorithmic}[1]
\Require mixture matrix $\mathbf{Y}\in\mathbb{R}_+^{2\times T}$, 
$\{\lambda_1,\lambda_2,\alpha,L\}_{\text{heart}}$, 
$\{\lambda_1,\lambda_2,\alpha,L\}_{\text{lung}}$
\Ensure Estimated signals $\mathbf{x}_{\text{heart}},\mathbf{x}_{\text{lung}}\in\mathbb{R}_+^{1\times T}$

\For{$j = 1$ \textbf{to} $2$}
    \If{$j = 1$}
        \State $\{\lambda_1,\lambda_2,\alpha,L\} \leftarrow \{\lambda_1,\lambda_2,\alpha,L\}_{\text{heart}}$
    \Else
        \State $\{\lambda_1,\lambda_2,\alpha,L\} \leftarrow \{\lambda_1,\lambda_2,\alpha,L\}_{\text{lung}}$
    \EndIf

    \State Initialize $\mathbf{A}_1,\mathbf{A}_2,\ldots,\mathbf{A}_L\in\mathbb{R}_+^{2\times2}$, $\mathbf{X}\in\mathbb{R}_+^{2\times T}$
    \State $\mathbf{Y} \leftarrow \lambda_1.\mathbf{Y} + \lambda_2$
    \Repeat
        \For{each layer $l = 1$ \textbf{to} $L$}
            \State $\mathbf{A}_l \leftarrow \mathbf{A}_l \!\cdot\!
            \left(
                \dfrac{\sum_{t=1}^{T} x_{jt}
                \left[
                    \left(\dfrac{y_{it}}
                    {[\mathbf{A}_1\mathbf{A}_2\ldots\mathbf{A}_l\mathbf{X}]_{it}}\right)^{\!\alpha}\right]}
                {\sum_{t=1}^{T} x_{jt}}
            \right)^{\!1/\alpha}$
            \State$\mathbf{X} \leftarrow \mathbf{X} \!\cdot\!
            \left(
                \dfrac{\sum_{i=1}^{I}
                [\mathbf{A}_l^\top\!\ldots\!\mathbf{A}_2^\top\!\mathbf{A}_1^\top]_{ij}
                \left[
                    \left(\dfrac{y_{it}}
                    {[\mathbf{A}_1\mathbf{A}_2\ldots\mathbf{A}_l\mathbf{X}]_{it}}\right)^{\!\alpha}\right]}
                {\sum_{i=1}^{I}[\mathbf{A}_l^\top\!\ldots\!\mathbf{A}_2^\top\!\mathbf{A}_1^\top]_{ij}}
            \right)^{\!1/\alpha}$
            \State Normalize $\mathbf{A}_l,\mathbf{X}$
        \EndFor
    \Until a stopping criterion is met

    \For{each $\mathbf{X}$ row $i = 1$ \textbf{to} $2$}
        \State Calculate $p_i$ using Equation~\ref{eq:14}
    \EndFor

    \If{$j = 1$}
        \State $\mathbf{x}_{\text{heart}} = \mathbf{X}(\arg\min_i p_i,:)$
    \Else
        \State $\mathbf{x}_{\text{lung}} = \mathbf{X}(\arg\max_i p_i,:)$
    \EndIf
\EndFor
\end{algorithmic}
\end{spacing}
\vspace{0.5 cm}
\end{breakablealgorithm}

\vspace{1cm}
\begin{breakablealgorithm}
\caption{LingoNMF}
\label{alg:lingonmf}
\begin{spacing}{1}
\small
\begin{algorithmic}[1]
\Require mixture matrix $\mathbf{Y}\in\mathbb{R}_+^{2\times T}$, 
$\{\lambda_1,\lambda_2,\alpha,L\}_{\text{heart}}$, 
$\{\lambda_1,\lambda_2,\alpha,L\}_{\text{lung}}$
\Ensure Estimated signals $\mathbf{x}_{\text{heart}},\mathbf{x}_{\text{lung}}\in\mathbb{R}_+^{1\times T}$

\For{$j = 1$ \textbf{to} $2$}
    \If{$j = 1$}
        \State $\{\lambda_1,\lambda_2,\alpha,L\} \leftarrow \{\lambda_1,\lambda_2,\alpha,L\}_{\text{heart}}$
    \Else
        \State $\{\lambda_1,\lambda_2,\alpha,L\} \leftarrow \{\lambda_1,\lambda_2,\alpha,L\}_{\text{lung}}$
    \EndIf

    \State Initialize $\mathbf{A}_1,\mathbf{A}_2,\ldots,\mathbf{A}_L\in\mathbb{R}_+^{2\times2}$, $\mathbf{X}\in\mathbb{R}_+^{2\times T}$
    \State $\mathbf{Y} \leftarrow \lambda_1.\mathbf{Y} + \lambda_2$
    \Repeat
        \For{each layer $l = 1$ \textbf{to} $L$}
            \State $\mathbf{A}_l \leftarrow \mathbf{A}_l \!\cdot\!
            \left(
                \dfrac{\sum_{t=1}^{T} x_{jt}
                \left[
                    \left(\dfrac{y_{it}}
                    {[\mathbf{A}_1\mathbf{A}_2\ldots\mathbf{A}_l\mathbf{X}]_{it}}\right)^{\!\alpha}\right]}
                {\sum_{t=1}^{T} x_{jt}}
            \right)^{\!1/\alpha}$
            \State$\mathbf{X} \leftarrow \mathbf{X} \!\cdot\!
            \left(
                \dfrac{\sum_{i=1}^{I}
                [\mathbf{A}_l^\top\!\ldots\!\mathbf{A}_2^\top\!\mathbf{A}_1^\top]_{ij}
                \left[
                    \left(\dfrac{y_{it}}
                    {[\mathbf{A}_1\mathbf{A}_2\ldots\mathbf{A}_l\mathbf{X}]_{it}}\right)^{\!\alpha}\right]}
                {\sum_{i=1}^{I}[\mathbf{A}_l^\top\!\ldots\!\mathbf{A}_2^\top\!\mathbf{A}_1^\top]_{ij}}
            \right)^{\!1/\alpha}$
            \State Normalize $\mathbf{A}_l,\mathbf{X}$
            \State Compute $\widehat{\mathbf{f}}_{f}$ from $\mathbf{X}$ using Equation~\ref{eq:16}
            \State Extract features from $\mathbf{X}$
            \State Prompt LLM $\rightarrow$ receive $\mathbf{f}_{f} = [f_{\mathrm{heart}}, f_{\mathrm{lung}}]$
        \EndFor
    \Until a stopping criterion is met

    \For{each $\mathbf{X}$ row $i = 1$ \textbf{to} $2$}
        \State Calculate $p_i$ using Equation~\ref{eq:14}
    \EndFor

    \If{$j = 1$}
        \State $\mathbf{x}_{\text{heart}} = \mathbf{X}(\arg\min_i p_i,:)$
    \Else
        \State $\mathbf{x}_{\text{lung}} = \mathbf{X}(\arg\max_i p_i,:)$
    \EndIf
\EndFor
\end{algorithmic}
\end{spacing}
\vspace{0.5 cm}
\end{breakablealgorithm}

\vspace{1cm}

\begin{algorithm} [H]
\caption{Chem-NMF}
\label{alg:chemnmf}
\begin{spacing}{1}
\begin{algorithmic}[1]
\Require Input data $\mathbf{Y}\in\mathbb{R}_+^{I\times T}$, layer rank $\mathbf{R}=[R_1,..,R_L]$, $\alpha$, $bf$
\Ensure Output activation map $\mathbf{X}^{(L)} \in \mathbb{R}^{R_L \times T}$, basis features $\mathbf{A}_{tot} \in \mathbb{R}^{I \times R_L}$

\State $\mathbf{Y}^{(0)} = \mathbf{Y}$; \quad $\mathbf{A}_{tot} = \mathbf{I}_I$
\For{$\ell=1$ \textbf{to} $L$}
    \If{$\ell=1$}
        \State Random $\mathbf{A}^{(1)} \in \mathbb{R}_+^{I\times R_1}$, $\mathbf{X}^{(1)} \in \mathbb{R}_+^{R_1\times T}$
    \Else
        \State Random $\mathbf{X}^{(\ell)} \in \mathbb{R}_+^{R_\ell\times T}, \mathbf{A}_{rand}\in \mathbb{R}_+^{R_{\ell-1}\times R_\ell}$
        \State $\mathbf{A}_{base} = \mathrm{mean}(\mathbf{A}^{(\ell-1)}) \cdot \mathds{1}_{R_{\ell-1}\times R_\ell}$
        \State $\mathbf{A}^{(\ell)} = (1-bf)\mathbf{A}_{rand} + bf\mathbf{A}_{base}$
    \EndIf
    \Repeat

    \State $\widehat{\mathbf{Y}}^{(\ell-1)} \gets \mathbf{A}^{(\ell)} \mathbf{X}^{(\ell)}$

    \State $\mathbf{X}^{(\ell)} \gets \mathbf{X}^{(\ell)} \odot
        \left(
          \frac{ (\mathbf{A}^{(\ell)})^\top 
                 \left( \mathbf{Y}^{(\ell-1)} \oslash \widehat{\mathbf{Y}}^{(\ell-1)} \right)^{\alpha} }
               { (\mathbf{A}^{(\ell)})^\top \mathds{1}_I \, \mathds{1}_T^\top }
        \right)^{1/\alpha}$
    
    \State $\mathbf{A}^{(\ell)} \gets \mathbf{A}^{(\ell)} \odot
        \left(
          \frac{ \left( \mathbf{Y}^{(\ell-1)} \oslash \widehat{\mathbf{Y}}^{(\ell-1)} \right)^{\alpha}
                 (\mathbf{X}^{(\ell)})^\top }
               { \mathds{1}_I \, (\mathbf{X}^{(\ell)} \mathds{1}_T)^\top }
        \right)^{1/\alpha}$

    \Until a stopping criterion is met
    \State $\mathbf{A}_{tot} \gets \begin{cases}
        \mathbf{A}^{(1)}, & \ell=1 \\
        \mathbf{A}_{tot}\mathbf{A}^{(\ell)}, & \ell>1
    \end{cases}$
    \State $\mathbf{Y}^{(\ell)} \gets \mathbf{X}^{(\ell)}$
\EndFor
\State \Return $\mathbf{A}_{tot}, \mathbf{X}^{(L)}, \{\mathbf{A}^{(\ell)}\}, \{\mathbf{X}^{(\ell)}\}$
\end{algorithmic}
\end{spacing}
\vspace{0.5 cm}
\end{algorithm}

		\setcounter{figure}{0}
		\setcounter{equation}{0}
		\setcounter{table}{0}
\label{appendix_Proofs}
\chapter{Mathematical Proofs}

\section{Theorem 4.4.1}
\label{sec:thm441}
\begin{proof}

We differentiate the cost function $D_{\alpha}(\mathbf{Y} \parallel \mathbf{A}\mathbf{X})$ with respect to $x_{jt}$:
\begin{equation}
\frac{\partial D}{\partial x_{jt}} = \frac{1}{\alpha} \sum_{i} a_{ij} \left[ 1 - \left( \frac{y_{it}}{[\mathbf{A}\mathbf{X}]_{it}} \right)^{\alpha} \right].
\end{equation}

We employ a projected (transformed) gradient descent approach:
\begin{equation}
\Phi(x_{jt}) \leftarrow \Phi(x_{jt}) - \eta_{jt} \frac{\partial D}{\partial \Phi(x_{jt})},
\label{eq:grad}
\end{equation}

where we define $\Phi(x) = x^{\alpha}$, and choose the learning rate as:

\begin{equation}
\eta_{jt} = \frac{\alpha^2 \Phi(x_{jt})}{x_{jt}^{1-\alpha} \sum_{i} a_{ij}}.
\label{eq:eta}
\end{equation}

Applying this transformation and using the chain rule, we obtain:
\begin{equation}
\frac{\partial D}{\partial \Phi(x_{jt})} 
= \frac{\partial D}{\partial x_{jt}} \cdot \frac{\partial x_{jt}}{\partial \Phi(x_{jt})} 
= \frac{1}{\alpha} \sum_{i} a_{ij} \left[ 1 - \left( \frac{y_{it}}{[\mathbf{A}\mathbf{X}]_{it}} \right)^{\alpha} \right] \cdot \frac{1}{\alpha x_{jt}^{\alpha-1}}.
\label{eq:devgrad}
\end{equation}

Since $\Phi(x_{jt}) = x_{jt}^{\alpha}$, we substitute Equation~\ref{eq:eta} and Equation~\ref{eq:devgrad} into Equation~\ref{eq:grad}, yielding:
\allowdisplaybreaks
\begin{align}
x_{jt} & \leftarrow \Phi^{-1} \left( \Phi(x_{jt}) - \eta_{jt} \frac{\partial D}{\partial \Phi(x_{jt})} \right) \nonumber \\
&\leftarrow \left( x_{jt}^{\alpha} - \frac{\alpha^2 x_{jt}^{\alpha}}{x_{jt}^{1-\alpha} \sum_{i} a_{ij}} \cdot \frac{\partial D}{\partial \Phi(x_{jt})} \right)^{1/\alpha} \nonumber \\
&\leftarrow \left( x_{jt}^{\alpha} - \frac{\alpha^2 x_{jt}^{\alpha}}{x_{jt}^{1-\alpha} \sum_{i} a_{ij}} \cdot \frac{1}{\alpha} \sum_{i} a_{ij} \left[ 1 - \left( \frac{y_{it}}{[\mathbf{A}\mathbf{X}]_{it}} \right)^{\alpha} \right] \cdot \frac{1}{\alpha x_{jt}^{\alpha-1}} \right)^{1/\alpha} \nonumber \\
&\leftarrow \Bigg( x_{jt}^{\alpha} - \frac{\alpha^2 x_{jt}^{\alpha}}{x_{jt}^{1-\alpha} \sum_{i=1}^{I} a_{ij}} \cdot \frac{1}{\alpha} \sum_{i} a_{ij} \cdot \frac{1}{\alpha x_{jt}^{\alpha-1}} \nonumber \\ 
&\quad + \frac{\alpha^2 x_{jt}^{\alpha}}{x_{jt}^{1-\alpha} \sum_{i} a_{ij}} \cdot \frac{1}{\alpha} \sum_{i} a_{ij} \left( \frac{y_{it}}{[\mathbf{A}\mathbf{X}]_{it}} \right)^{\alpha} \cdot \frac{1}{\alpha x_{jt}^{\alpha-1}} \Bigg)^{1/\alpha} \nonumber \\ 
&\leftarrow \left( x_{jt}^{\alpha} - x_{jt}^{\alpha} + \frac{x_{jt}^{\alpha}}{\sum_{i=1}^{I} a_{ij}} \cdot \sum_{i} a_{ij} \left( \frac{y_{it}}{[\mathbf{A}\mathbf{X}]_{it}} \right)^{\alpha} \right)^{1/\alpha} \nonumber \\
&\leftarrow \left( \frac{x_{jt}^{\alpha}}{\sum_{i} a_{ij}} \cdot \sum_{i} a_{ij} \left( \frac{y_{it}}{[\mathbf{A}\mathbf{X}]_{it}} \right)^{\alpha} \right)^{1/\alpha}. \nonumber
\end{align}
\begin{equation}
    x_{jt} \leftarrow x_{jt} \left( \frac{\sum\limits_{i} a_{ij} \left( \frac{y_{it}}{[\mathbf{A}\mathbf{X}]_{it}} \right)^{\alpha}}{\sum\limits_{i} a_{ij}} \right)^{\frac{1}{\alpha}}.
\end{equation}

Similarly, we derive the update rule for $a_{ij}$ as:  

\begin{equation}
a_{ij} \leftarrow a_{ij} \left( \frac{\sum\limits_{t=1}^{T} x_{jt} \left( \frac{y_{it}}{[\mathbf{A}\mathbf{X}]_{it}} \right)^{\alpha}}{\sum\limits_{t=1}^{T} x_{jt}} \right)^{\frac{1}{\alpha}}. 
\end{equation}
\qedhere \text{ QED.}
\end{proof}
\section{Lemma 4.4.1}
\label{sec:lm441}
\begin{proof} 
We have two conditions:

\textbf{(i) Identity Condition:} \textit{Setting $\mathbf{X}' = \mathbf{X}$ in the auxiliary function $G(\mathbf{X}, \mathbf{X}')$ recovers the original $F(\mathbf{X})$, such that  $G(\mathbf{X}, \mathbf{X}) = F(\mathbf{X})$.}

Setting $\mathbf{X}' = \mathbf{X}$, we simplify $\zeta_{itj}$ as:
\[
\zeta_{itj} = \frac{a_{ij} x_{jt}}{\sum_{j=1}^{J} a_{ij} x_{jt}} 
= \frac{a_{ij} x_{jt}}{[\mathbf{A}\mathbf{X}]_{it}}.
\]

Substituting $\zeta_{itj}$ into $G(\mathbf{X}, \mathbf{X})$ and simplifying, we get:
\begin{align}
G(\mathbf{X}, \mathbf{X}) 
&= \frac{1}{\alpha(\alpha-1)} \sum_{ijt} 
y_{it} \frac{a_{ij} x_{jt}}{[\mathbf{A}\mathbf{X}]_{it}} 
\left[ \left( \frac{[\mathbf{A}\mathbf{X}]_{it}}{y_{it}} \right)^{1-\alpha} 
+ (\alpha-1) \frac{[\mathbf{A}\mathbf{X}]_{it}}{y_{it}} - \alpha \right] \nonumber\\
&= \frac{1}{\alpha(\alpha-1)} \sum_{it} 
y_{it} \frac{\sum_{j=1}^{J} a_{ij} x_{jt}}{[\mathbf{A}\mathbf{X}]_{it}} 
\left[ \left( \frac{[\mathbf{A}\mathbf{X}]_{it}}{y_{it}} \right)^{1-\alpha} 
+ (\alpha-1) \frac{[\mathbf{A}\mathbf{X}]_{it}}{y_{it}} - \alpha \right] \nonumber\\
&= \frac{1}{\alpha(\alpha-1)} \sum_{it} 
y_{it} \frac{[\mathbf{A}\mathbf{X}]_{it}}{[\mathbf{A}\mathbf{X}]_{it}} 
\left[ \left( \frac{[\mathbf{A}\mathbf{X}]_{it}}{y_{it}} \right)^{1-\alpha} 
+ (\alpha-1) \frac{[\mathbf{A}\mathbf{X}]_{it}}{y_{it}} - \alpha \right] \nonumber\\
&= \frac{1}{\alpha(\alpha-1)} \sum_{it} 
\Big( [\mathbf{A}\mathbf{X}]_{it}^{1-\alpha} y_{it}^\alpha 
+ (\alpha-1) [\mathbf{A}\mathbf{X}]_{it}^\alpha - \alpha y_{it} \Big) 
= F(\mathbf{X}).
\end{align}

\qedhere \text{ QED.}

\vspace{1em}
\textbf{(ii) Upper Bound Condition:}  
\textit{The auxiliary function $G(\mathbf{X}, \mathbf{X}')$ provides an upper bound on $F(\mathbf{X})$, such that $G(\mathbf{X}, \mathbf{X}') \geq F(\mathbf{X})$.}

\begin{definition}[Jensen’s Inequality]
Let $f(z)$ be a convex function. For any weights $w_j \geq 0$ such that $\sum_{j} w_j = 1$, we have:
\begin{equation}
    f\left(\sum_{j} w_j z_j\right) \leq \sum_{j} w_j f(z_j).
\end{equation}
\end{definition} 
We consider the function associated with the $\alpha$-divergence:
\begin{equation}
    f(z) = \frac{1}{\alpha(\alpha-1)} \left[ z^{1-\alpha} + (\alpha -1)z - \alpha \right].
\end{equation}
Its first derivative is:
\begin{equation}
    f'(z) = \frac{1}{\alpha(\alpha-1)} \left[ (1-\alpha)z^{-\alpha} + (\alpha-1) \right].
\end{equation}
Differentiating again, we obtain:
\begin{equation}
    f''(z) = \frac{1}{\alpha(\alpha-1)} \left[ -\alpha(1-\alpha)z^{-\alpha-1} \right].
\end{equation}
Rewriting this:
\begin{equation}
    f''(z) = z^{-\alpha-1}.
\end{equation}

Since $z^{-\alpha-1} \geq 0$ for $z > 0$, we conclude $f(z)$ is convex. Applying Jensen’s inequality, we obtain:
\begin{equation}
    f\left( \sum_{j}\frac{a_{ij} x_{jt}}{y_{it}} \right) 
    \leq \sum_{j} \zeta_{itj} f\left(\frac{a_{ij} x_{jt}}{y_{it}\zeta_{itj}} \right),
\end{equation}
where the weights $\zeta_{itj}$ are defined as:
\begin{equation}
    \zeta_{itj} = \frac{a_{ij} x'_{jt}}{\sum_{j=1}^{J} a_{ij} x'_{jt}}, 
    \quad \sum_{j} \zeta_{itj} = 1, \quad \zeta_{itj} \geq 0.
\end{equation}

Multiplying both sides by $y_{it}$ and summing over all $i$ and $t$, we get:
\begin{equation}
   F(\mathbf{X}) = \sum_{it} y_{it} f\left( \sum_{j}\frac{a_{ij} x_{jt}}{y_{it}} \right) 
   \leq \sum_{itj} y_{it} \zeta_{itj} f\left(\frac{a_{ij} x_{jt}}{y_{it}\zeta_{itj}} \right) 
   = G(\mathbf{X}, \mathbf{X}').
\end{equation}
\qedhere $\square$ \text{ QED.}
\end{proof}

\section{Theorem 4.4.2}
\label{sec:thm442}
\begin{proof} 

Consider the function associated with the $\alpha$-divergence:
\begin{equation}
    f(z) = \frac{1}{\alpha(\alpha-1)} \left[ z^{1-\alpha} + (\alpha -1)z - \alpha \right].
\end{equation}
Its first derivative is:
\begin{equation}
    f'(z) = \frac{1}{\alpha(\alpha-1)} \left[ (1-\alpha)z^{-\alpha} + (\alpha-1) \right].
    \label{eq:firstderivative}
\end{equation}
Rewriting $F(\mathbf{X})$ and $G(\mathbf{X},\mathbf{X}')$ as:
\begin{equation}
   F(\mathbf{X})= \sum_{it} y_{it} f\left( \sum_{j}\frac{ a_{ij} x_{jt} }{y_{it}}\right).
\end{equation}

\begin{equation}
   G(\mathbf{X},\mathbf{X}')=\sum_{itj} y_{it} \zeta_{itj} f\left(\frac{a_{ij} x_{jt}}{y_{it}\zeta_{itj}} \right).
\end{equation}

We minimize $G(\mathbf{X}, \mathbf{X}')$ by setting the gradient to zero:

\begin{equation}
    \frac{\partial G(\mathbf{X}, \mathbf{X}')}{\partial x_{jt}} = \sum_{i} y_{it} \zeta_{itj} f'\left(\frac{a_{ij} x_{jt}}{y_{it}\zeta_{itj}}\right) \cdot \frac{\partial}{\partial x_{jt}} \left( \frac{a_{ij} x_{jt}}{y_{it}\zeta_{itj}} \right)=0.
\end{equation}

Since \( \frac{\partial}{\partial x_{jt}} \left( \frac{a_{ij} x_{jt}}{y_{it}\zeta_{itj}} \right) = \frac{a_{ij}}{y_{it}\zeta_{itj}} \), we get:

\begin{equation}
    \frac{\partial G(\mathbf{X}, \mathbf{X}')}{\partial x_{jt}} = \sum_{i} y_{it} \zeta_{itj} f'\left(\frac{a_{ij} x_{jt}}{y_{it}\zeta_{itj}}\right) \frac{a_{ij}}{y_{it}\zeta_{itj}}=0.
\end{equation}

Substituting \( f'(z) \) from Equation~\ref{eq:firstderivative} into the expression:

\begin{align}
    \frac{\partial G(\mathbf{X}, \mathbf{X}')}{\partial x_{jt}} = & \sum_{i} y_{it}\zeta_{itj} \cdot\frac{1}{\alpha(\alpha-1)} \left[ (1-\alpha) \left( \frac{a_{ij} x_{jt}}{y_{it}\zeta_{itj}} \right)^{-\alpha} + (\alpha-1) \right] \frac{a_{ij}}{y_{it}\zeta_{itj}} \nonumber \\
    & = \frac{1}{\alpha}\sum_{i}  a_{ij}\left[ 1- \left( \frac{a_{ij} x_{jt}}{y_{it}\zeta_{itj}} \right)^{-\alpha} \right] =0.
\end{align}

Rearranging the equation for $\alpha \neq 0$:

\begin{equation}
    \sum_{i} a_{ij} = \sum_{i} a_{ij} \left( \frac{a_{ij} x_{jt}}{y_{it} \zeta_{itj}} \right)^{-\alpha}= \sum_{i} a_{ij} \left( \frac{y_{it} \zeta_{itj}}{a_{ij} x_{jt}} \right)^{\alpha}.
\end{equation}

Dividing both sides by \( \sum_{i}a_{ij} \):

\begin{equation}
    1 = \frac{\sum_{i} a_{ij} \left( \frac{y_{it} \zeta_{itj}}{a_{ij} x_{jt}} \right)^{\alpha}}{\sum_{i} a_{ij}}.
\end{equation}

Substituting \( \zeta_{itj} \) from Equation~\ref{eq:zeta} into the expression:

\begin{equation}
    1 = \frac{\sum_{i} a_{ij} \left( \frac{y_{it}a_{ij}x'_{jt}}{a_{ij}x_{jt}\sum_{i} a_{ij} x'_{jt}} \right)^{\alpha}}{\sum_{i}a_{ij}} = \frac{\sum_{i} a_{ij} \left( \frac{y_{it}}{\sum_{i} a_{ij}x'_{jt}} \right)^{\alpha}\cdot \left( \frac{x'_{jt}}{x_{jt}} \right)^{\alpha}}{\sum_{i}a_{ij}},
\end{equation}

which leads to:

\begin{equation}
    \left(\frac{x_{jt}}{x'_{jt}}\right) = \left[ \frac{\sum\limits_{i} a_{ij} \left( \frac {y_{it}} {\sum\limits_{i} a_{ij}x'_{jt}}\right)^{\alpha}}{\sum\limits_{i} a_{ij}} \right]^{1/\alpha},
\end{equation}

which suggests the following update rule for \( x_{jt} \):

\begin{equation}
    x_{jt} \leftarrow x_{jt} \left( \frac{\sum\limits_{i} a_{ij} \left( \frac{y_{it}}{[\mathbf{A}\mathbf{X}]_{it}} \right)^{\alpha}}{\sum\limits_{i} a_{ij}} \right)^{\frac{1}{\alpha}}.
\end{equation}

Since \( G(\mathbf{X}, \mathbf{X}') \) is an auxiliary function for \( F(\mathbf{X}) \), minimizing \( G(\mathbf{X}, \mathbf{X}') \) at each step ensures that \( F(\mathbf{X}) \) is non-increasing: 

\begin{equation}  
    F(\mathbf{X}^{(t+1)}) \leq G(\mathbf{X}^{(t+1)}, \mathbf{X}^{(t)}) \leq G(\mathbf{X}^{(t)}, \mathbf{X}^{(t)}) = F(\mathbf{X}^{(t)}).
\end{equation}  

\qedhere \text{ QED.}

\end{proof}

\section{Lemma 4.4.2}
\label{sec:lm442}
\begin{proof}
Each layer solves the following minimization problem:
\begin{equation}
    (\mathbf{A}^{(l)}, \mathbf{X}^{(l)}) = \arg \min_{\mathbf{A}, \mathbf{X}} D_{\alpha}(\mathbf{X}^{(l-1)} \parallel \mathbf{A}\mathbf{X}).
\end{equation}
This ensures:
\begin{equation}
    D_{\alpha}(\mathbf{X}^{(l-1)} \parallel \mathbf{A}^{(l)} \mathbf{X}^{(l)}) \leq D_{\alpha}(\mathbf{X}^{(l-1)} \parallel \mathbf{A}^{(l-1)} \mathbf{X}^{(l-1)}).
    \label{eq:this}
\end{equation}
Applying the non-increasing property in Equation~\ref{eq:Gfunc} to two consecutive layers, we get:
\begin{equation}
     D_{\alpha}(\mathbf{X}^{(l-1)} \parallel \mathbf{A}^{(l-1)} \mathbf{X}^{(l-1)})\leq D_{\alpha}(\mathbf{X}^{(l-2)} \parallel \mathbf{A}^{(l-1)} \mathbf{X}^{(l-1)}).
     \label{eq:that}
\end{equation}

Combining Equation~\ref{eq:this} and Equation\ref{eq:that}, we derive:
\begin{equation}
    D_{\alpha}(\mathbf{X}^{(l-1)} \parallel \mathbf{A}^{(l)} \mathbf{X}^{(l)}) \leq D_{\alpha}(\mathbf{X}^{(l-1)} \parallel \mathbf{A}^{(l-1)} \mathbf{X}^{(l-1)})\leq D_{\alpha}(\mathbf{X}^{(l-2)} \parallel \mathbf{A}^{(l-1)} \mathbf{X}^{(l-1)}).
\end{equation}

Thus, by definition,
\begin{equation}
    D_l \leq D_{l-1}.
\end{equation}
\qedhere \text{ QED.}
\end{proof}
\section{Theorem 4.4.3}
\label{sec:thm443}
\begin{proof}
Let $M_l$ be the maximum divergence along the optimization path at layer $l$,
and let $D_l$ denote the divergence at layer $l$. Assume that
\begin{enumerate}
    \item $M_l$ is non-increasing for all sufficiently large $l$;
    \item $D_l$ is non-increasing for all $l > 1$ (by Lemma~4.4.2);
    \item there exists $L_\Delta \in \mathbb{N}$ such that for all $l \ge L_\Delta$,
    \begin{equation}
        0 \le D_{l-2} - D_{l-1} \le M_{l-1} - M_l.
    \end{equation}
\end{enumerate}
Define
\begin{equation}
    \mu_l = M_{l-1} - M_l,\qquad
    \delta_l = D_{l-1} - D_l,\qquad
    \xi_l = M_l - D_{l-1}.
\end{equation}
Then, for all $l > 1$,
\begin{equation}
\begin{aligned}
\xi_l - \xi_{l-1}
&= (M_l - D_{l-1}) - (M_{l-1} - D_{l-2}) \\
&= (M_l - M_{l-1}) + (D_{l-2} - D_{l-1}) \\
&= -\,\mu_l + \delta_{l-1}.
\end{aligned}
\end{equation}
By Lemma~4.4.2, $D_l$ is non-increasing, hence
\begin{equation}
    \forall\,l \ge 3:\quad \delta_{l-1} = D_{l-2} - D_{l-1} \ge 0.
\end{equation}
Since $M_l$ is non-increasing for all sufficiently large $l$, there exists
$L_M \in \mathbb{N}$ such that
\begin{equation}
    \forall\,l \ge L_M:\quad \mu_l = M_{l-1} - M_l \ge 0.
\end{equation}
By assumption (3), there exists $L_\Delta$ such that for all $l \ge L_\Delta$,
\begin{equation}
    0 \le \delta_{l-1} \le \mu_l.
\end{equation}
Set $L^* := \max\{L_M, L_\Delta, 3\}$. Then, for all $l \ge L^*$,
\begin{equation}
\begin{aligned}
0 \le \delta_{l-1} \le \mu_l
&\Longrightarrow -\,\mu_l + \delta_{l-1} \le 0 \\
&\Longrightarrow \xi_l - \xi_{l-1} \le 0 \\
&\Longrightarrow \xi_l \le \xi_{l-1}.
\end{aligned}
\end{equation}
Define the probability weights
\begin{equation}
    P_l := \frac{1}{Z}\, e^{-\beta \xi_l}, \qquad \beta > 0,
\end{equation}
where $Z$ is a normalization constant. Since $x \mapsto e^{-\beta x}$ is strictly
decreasing for $\beta > 0$, the monotonicity of $(\xi_l)$ implies that, for all
$l \ge L^*$,
\begin{equation}
    \xi_l \le \xi_{l-1}
    \ \Longrightarrow\ 
    \frac{1}{Z} e^{-\beta \xi_l} \ge \frac{1}{Z} e^{-\beta \xi_{l-1}}
    \ \Longrightarrow\ 
    P_l \ge P_{l-1}.
\end{equation}
\qedhere \text{ QED.}
\end{proof}

\section{Lemma 4.4.3}
\label{sec:lm443}
\begin{proof}
Since $D_l, M_l \ge 0$ are non-increasing and lower bounded, the sequences $M_l$ and $D_l$ converge to finite limits $M_\infty$ and $D_\infty$, respectively.
\begin{equation}
\lim_{l\to\infty} M_l = M_\infty,\qquad \lim_{l\to\infty} D_l = D_\infty.
\end{equation}
For any $\varepsilon>0,\;\exists\,N_M,N_D\in\mathbb{N}$ such that:
\begin{equation}
l\ge N_M \Rightarrow |M_l-M_\infty|<\tfrac{\varepsilon}{2},\qquad l\ge N_D \Rightarrow |D_l-D_\infty|<\tfrac{\varepsilon}{2}.
\end{equation}
Let $N_\xi=\max\{N_M,\;N_D\}$. Then for all $l\ge N_\xi$:
\begin{equation}
\begin {aligned}
&|\xi_l-(M_\infty-D_\infty)|
=|(M_l-M_\infty)-(D_{l-1}-D_\infty)| \\
&\le |M_l-M_\infty|+|D_{l-1}-D_\infty|
<\varepsilon.
\end {aligned}
\end{equation}
Which implies:
\begin{equation}
\lim_{l\to\infty}\xi_l = M_\infty - D_\infty. 
\end{equation}
By continuity property we have:
\begin{equation}
\lim_{l\to\infty} P_l
=\lim_{l\to\infty}\tfrac{1}{Z}e^{-\beta\xi_l}
=\tfrac{1}{Z}e^{-\beta\lim_{l\to\infty}\xi_l}
=\tfrac{1}{Z}e^{-\beta(M_\infty-D_\infty)}
=P_\infty.
\end{equation}
\qedhere \text{ QED.}
\end{proof}

\section{Lemma 4.4.4}
\label{sec:lm444}
\begin{proof}
At layer $l$, the escape probability is $P_l$. This implies:
\begin{equation}
\mathbb{P}(\text{no escape at layer }l \mid S_{l-1}) = 1 - P_l.
\end{equation}
For $l=1$,
\begin{equation}
\mathbb{P}(S_1)= 1- P_l \vert_{l=1}=1-P_1.
\end{equation}
The law of conditional probability states:
\begin{equation}
\mathbb{P}(S_l)=\mathbb{P}(S_{l-1})\cdot \mathbb{P}(\text{no escape at layer }l \mid S_{l-1})=\mathbb{P}(S_{l-1})\cdot(1 - P_l).
\end{equation}
Assume for some $m \ge 1$:
\begin{equation}
\mathbb{P}(S_m)=\prod_{j=1}^{m}(1-P_j).
\end{equation}
Thus:
\begin{equation}
\mathbb{P}(S_{m+1})
=\mathbb{P}(S_m)(1-P_{m+1})
=\Big(\prod_{j=1}^{m}(1-P_j)\Big)(1-P_{m+1})
=\prod_{j=1}^{m+1}(1-P_j).
\end{equation}
By induction, the claim holds for all $n \ge 1$. Therefore:
\begin{equation}
   \mathbb{P}(S_n)=\prod_{l=1}^{n}(1-P_l). 
\end{equation}
\qedhere \text{ QED.}
\end{proof}

		\setcounter{figure}{0}
		\setcounter{equation}{0}
		\setcounter{table}{0}
\end{appendix}


\begingroup
\footnotesize
\printbibliography
\endgroup

\label{NumDocumentPages}

\end{document}